\RequirePackage{fix-cm}
\documentclass[a4paper,10pt,twoside]{article}

\RequirePackage[a4paper,head=1cm]{geometry} 

\usepackage{mathptmx}      % use Times fonts if available on your TeX system
\usepackage{latexsym}

\usepackage{graphicx}
\usepackage{latexsym}
\usepackage[english]{babel}
\usepackage{graphicx}
\usepackage[matrix,arrow,curve]{xy}
\usepackage{amssymb}
\usepackage{amsmath}
\usepackage{amsthm}
\usepackage{array}
\usepackage{float}
\usepackage{graphicx}
\usepackage{color}
\usepackage{makeidx}
\usepackage{hyperref}
\hypersetup{colorlinks,%
            citecolor=blue,%
            filecolor=blue,%
            linkcolor=blue,%
            urlcolor=blue}

\newtheorem{theorem}{Theorem}[section]
\newtheorem{lemma}{Lemma}
\newtheorem{corollary}{Corollary}

\newtheorem{proposition}{Proposition}
\newtheorem{definition}{Definition}
\newtheorem{remark}{Remark}

\newtheorem{acknowledgements}{Acknowledgements}

\newcommand{\C}{\mathbb{C}ov}

\newcommand{\E}{\mathbb{E}}
\newcommand{\V}{\mathbb{V}}

\newcommand{\R}{\mathbb{R}}
\newcommand{\p}{\partial}
\newcommand{\y}{\mathbf{y}}
\newcommand{\Y}{\mathbf{Y}}
\newcommand{\Var}{\mathbb{V}ar}
\renewcommand{\P}{f}
\author{Salima El Kolei}
\date{}
\title{ Parametric estimation of  hidden stochastic model by contrast minimization and deconvolution: application to the Stochastic Volatility Model}

\begin{document}
\footnotetext[1]{
This is a reprint of the original article published by Springer-Verlag Berlin Heidelberg Edition in Metrika, 2013, Journal 184 article 430, \url{http://link.springer.com/article/10.1007}. DOI: \url{10.1007/s00184-013-0430-3}. This reprint differs from the original in pagination and typographic details.}
\footnotetext[2]{Received: 7 March 2012}
\footnotetext[3]{Affiliation: S. El Kolei,
              Laboratoire de Math\'{e}matiques J.A. Dieudonn\'{e}\\
              UMR n 7351 CNRS UNSA\\
              Universit\'{e} de Nice - Sophia Antipolis\\
              06108 Nice Cedex 02 France \\
              Tel.: +33-04-92-07-62-56\\
 \href{salima@unice.fr}{salima@unice.fr}
}

\maketitle

\begin{abstract}

We study a new parametric approach for particular hidden stochastic models such as the Stochastic Volatility model. This method is based on contrast minimization and deconvolution. \\
After proving consistency and asymptotic normality of the estimation leading to asymptotic confidence intervals, we provide a thorough numerical study, which compares most of the classical methods that are used in practice (Quasi Maximum Likelihood estimator, Simulated Expectation Maximization Likelihood estimator and Bayesian estimators). We prove that our estimator clearly outperforms the Maximum Likelihood Estimator in term of computing time, but also most of the other methods. We also show that this contrast method is the most robust with respect to non Gaussianity of the error and also does not need any tuning parameter.\\

 
\end{abstract}

\section{Introduction}

This paper is concerned with the particular \emph{hidden stochastic model}:

\begin{equation}\label{mod1}
\left\lbrace\begin{array}{ll}
Y_i=X_{i}+\varepsilon_{i}\\
X_{i+1}=b_{\phi_0}(X_{i})+\eta_{i+1},
\end{array}
\right.
\end{equation}
where $(\varepsilon_{i})_{i\geq 1}$ and $(\eta_{i})_{i\geq 1}$ are two independent sequences of independent and identically distributed (i.i.d) centered random variables with variance $\sigma^{2}_{\epsilon}$ and $\sigma^{2}_{0}$. It is assumed that the variance $\sigma^{2}_{\epsilon}$ is known. The terminology \emph{hidden} comes from the unobservable character of the process $(X_i)_{i\geq 1}$ since the only available observations are $Y_{1},\cdots, Y_n$.\\ 
The dynamics of the process $X_i$ is described by a measurable function $b_{\phi_0}$ which depends on an unknown parameter $\phi_0$ and by a sequence of i.i.d centered random variables with unknown variance $\sigma_0^2$. We denote by $\theta_0$ the vector of parameters governing the process $X_i$ and suppose that the model is correctly specified: that is, $\theta_0$ belongs to the parameter space $\Theta \subset \R^{r}$, with $r \in \mathbb{N}^{*}$.\\

 Inference in hidden Markov models is a real challenge and has been studied by many authors (see \cite{moulines}, \cite{DoFr01}, \cite{MR1747395}). K.C. Chanda provided in \cite{KCD95} an asymptotically normal estimator for the vector of parameters $\theta_0$ by using modified Yule Walker equation but it assumes that the function $b_{\phi_0}$ is linear in $\phi_0$ and $X_i$, so the model (\ref{mod1}) is reduced to an autoregressive model with measurement error.\\
 Recently, in \cite{Dou11}, the authors propose an efficient estimator of the vector of parameters $\theta_0$ for nonlinear function $b_{\phi_0}$. They prove that their \emph{Maximisation Likelihood Estimator} (MLE) is consistent and asymptotically normal. The main difficulty with their approach comes from the unobservable character of the process $X_i$ making the calculus of the exact likelihood intractable in practice: the likelihood is only available in the form of a multiple integral, so exact likelihood methods require simulations and have therefore an intensive computational cost. In many case, the MLE has to be approximated. A popular approach to approximate the MLE consists in using Monte Carlo Markov Chain (MCMC) simulation techniques. Thanks to the development of these methods, the MLE has known a huge progress and Bayesian estimations have received more attention (see \cite{SaRg93}). Another method for performing the MLE consists in using the Expectation-Maximization (EM) algorithm proposed by Dempster et al. in 1977 (see \cite{dempster}). Nevertheless, since $X_i$ is unobservable, this method requires to introduce a MCMC in the Expectation step. Although these methods are used in practice, they are expensive from a computational point of view. \\
 Some authors have proposed Sequential Monte Carlo algorithms (SMC) known as Particles Filtering methods which allow to approximate the likelikood. The computational cost is reduced by a recursive construction. We refer to the book of \cite{DoFr01} and \cite{moulines} for a complete review of these methods.\\
Particle Markov Chain Monte Carlo (PMCMC) is another method for estimating the model (\ref{mod1}). This method combines Particles filtering methods and MCMC methods to estimate the vector of parameters $\theta_0$. From a computational point of view, this approach is expensive and we refer the reader to \cite{MR2758115} for more details. In \cite{peters}, they propose an adaptive PMCMC method to estimate ecological hidden stochastic models.\\

We propose here an approach based on M-estimation: It consists in the optimisation of a well-chosen contrast function (see \cite{van} chapter p.41 for a partial review) and deconvolution strategy. The deconvolution problem is encountered in many statistical situations where the observations are collected with random errors. In this approach, a method for estimating the parameter $\phi_0$ has been proposed by F. Comte and M. Taupin (see \cite{comte3}). Their procedure of estimation is based on a modified least squared minimization. In the same perspective, J. Dedecker, A. Samson and M-L. Taupin in \cite{RLA11} propose also the same procedure of estimation based on a weighted least squared estimation: Their assumptions on the process $X_i$ are less restrictive than those proposed by F. Comte and M. Taupin and they provide consistent estimation of the parameter $\phi_0$ with a parametric rate of convergence in a very general framework. Their general estimator is based on the introduction of a kernel deconvolution density and depends on the choice of a weight function.\\

 The approach proposed here is different: it is not based on a weighted least squared estimation so that the choice of the weight function is not encountered in this paper. Moreover, it allows to estimate both the parameters $\phi_0$ and $\sigma^{2}_0$. Our principle of estimation relies on the Nadaraya-Watson strategy and is proposed by F. Comte et al. in \cite{CfLc11} in a non parametric case to estimate the function $b_{\phi}$ as a ratio of an estimate of $l_{\theta}=b_{\phi}f_{\theta}$ and an estimate of $f_{\theta}$, where $f_{\theta}$ represents the invariant density of the $X_i$. We propose to adapt their approach in a parametric context and suppose that the form of the stationary density $f_{\theta_0}$ is known up to some unknown parameter $\theta_0$. Our work is purely parametric but we go further in this direction by proposing an analytical expression of the asymptotic variance matrix $\Sigma(\hat{\theta}_{n})$ which allows to construct confidence interval. Furthermore, this approach is much less greedy from a computational point of view than the MLE and its implementation is straighforward.\\

\textbf{Applications:} Applications include epidemiology, meterology, neuroscience, ecology (see \cite{MR2850220}) and finance (see \cite{johannes}). For example, our approach can be applied to the five ecological state space models described in \cite{peters}.  Although the scope of our method is general, we have chosen to focus on the so-called discrete time Stochastic Volatility model (SV) introduced by Taylor in 1982 (see \cite{Ta82}). We also investigate the behavior of our method on the simpler autoregressive process AR(1) with measurement noise which has been widely studied and on which our method can be more easily understood and compared with other ones. Our procedure allows to estimate the parameters of a large class of discrete Stochastic Volatility models (ARCH-E model, Vasicek model, Merton model..), which is a real challenge in financial application.\\

\emph{ (i) Gaussian Autoregressive AR(1) with measurement noise:} It has the following form:
\begin{equation}\label{gaussintro}
\left\lbrace\begin{array}{ll}
Y_{i+1}=X_{i+1}+\varepsilon_{i+1}\\
X_{i+1}=\phi_0 X_{i}+\eta_{i+1},
\end{array}
\right.
\end{equation}	
where $\varepsilon_{i+1}$ and  $\eta_{i+1}$ are two centered Gaussian random variables with variance $\sigma_{\epsilon}^{2}$ assumed to be known and $\sigma^{2}_{0}$ assumed to be unknown. Additionally, we assume that $|\phi_0|<1$ which implies the stationary and ergodic property of the process $X_i$ (see \cite{Do94}).\\

\emph{(ii) SV model:} It is directly connected to the type of diffusion process used in asset-pricing theory ( see \cite{Melino}):
\begin{equation}\label{svintro}
\left\lbrace\begin{array}{ll}
R_{i+1}=\exp\left(\frac{X_{i+1}}{2}\right)\xi_{i+1},\\
X_{i+1}=\phi_{0}X_{i}+\eta_{i+1}, 
\end{array}
\right.
\end{equation}
where $\xi_{i+1}$ and  $\eta_{i+1} $ are two centered Gaussian random variables with variance $\sigma_{\xi}^{2}$ assumed to be known and equal to one and $\sigma^{2}_{0}$ assumed to be unknown. The variables $R_{i+1}$ represent the returns and $X_{i+1}$ is the log-volatility process.\\
By applying a log-transformation $Y_{i+1}=\log(R^{2}_{i+1})-\E[\log(\xi^{2}_{i+1})]$ and $\varepsilon_{i+1}=\log(\xi^{2}_{i+1})-\E[\log(\xi^{2}_{i+1})]$, the SV model  is a particular version of (\ref{mod1}).  We assume that $|\phi_0|<1$ and we refer the reader to \cite{JT00} for the mixing properties of stochastic volatility models.\\ 

Most of the computational problems stem from the assumptions that the innovation of the returns are Gaussian which translates into a logarithmic chi-square distribution when the model (\ref{svintro}) is transformed in a linear state space model. Many authors have ignored it in their implementation and many authors use some mixture of Gaussian to approximate the log-chi-square density. For example, in the Quasi-Maximum Likelihood (QML) method implemented by Jacquier, Polson and Rossi in \cite{JaPo94} and in the Simulated Expectation-Maximization Likelihood estimator proposed (SIEMLE) by Kim, Shephard, and Chib in \cite{shep} they used a mixture of Gaussian distribution to approximate the log-chi-square distribution. Harvey \cite{harvey} used the Kalman filter to estimate the likelihood of the transform state space model, hence the model was also assumed to be Gaussian.\\

\textbf{Organization of the paper:} The first purpose of the paper is to present our estimator and its statistical properties in Section \ref{Procedure}: Under weak assumptions, we show that it is a consistent and asymptotically normal estimator.\\ 
The second purpose of this paper consists in comparing our contrast estimator with different estimations: the QML, the SIEMLE and Bayesian estimators. Section \ref{Application} contains the numerical study: In Section \ref{implementation} we give the parameter estimates and the comparison with others ones for simulation data and Section \ref{Application_real_data} contains the study on real data. We compare our contrast estimator with other ones on the SP$\&$500 and FTSE index.  From a computational point of view, we show that the implementation of our estimator is straightforward and it is faster than the SIEMLE (see Table [\ref{computing_siem}] in Section \ref{comp-time}). Besides, we show that our estimator outperforms the QML and Bayesian estimators.\\

\textbf{Notations:} We denote by: $u^{*}(t)=\int_{}^{}e^{itx}u(x)dx$ the Fourier transform of the function $u(x)$ and $\left \langle u,v \right \rangle=\int_{}^{}u(x)\overline{v}(x)dx$ with $v\overline{v}=|v|^{2}$. We write $||u||_{2}=\left(\int_{}^{} |u(x)|^{2} dx\right)^{1/2}$ the norm of $u(x)$ on the space of functions $\mathbb{L}^{2}(\R)$. By property of the Fourier transform, we have $(u^{*})^{*}(x)=2\pi u(-x)$ and $\left\langle u_1,u_2\right\rangle=\frac{1}{2\pi}\left\langle u^{*}_1,u^{*}_2\right\rangle$. The vector of the partial derivatives of $f$ with respect to (w.r.t) $\theta$ is denoted by $\nabla_{\theta}f$ and the Hessian matrix of $f$ w.r.t $\theta$ is denoted by $\nabla^{2}_{\theta}f$. The Euclidean norm matrix, that is, the square root of the sum of the squares of all its elements will be written by $\left\|A\right\|$. We denote by $\Y_i$ the pair $(Y_i,Y_{i+1})$ and $\y_{i}=(y_{i},y_{i+1})$ is a given realisation of $\Y_i$. \\
In the following, $\mathbb{P}, \E, \Var$ and $\C$ denote respectively the probability $\mathbb{P}_{\theta_0}$, the expected value $\E_{\theta_0}$, the variance $\Var_{\theta_0}$ and the covariance $\C_{\theta_0}$ when the true parameter is $\theta_0$. Additionally, we write $\mathbf{P}_n$ (\emph{resp.} $\mathbf{P}$) the empirical expectation (\emph{resp.} theoretical), that is: for any stochastic variable $X$: $\mathbf{P}_{n}(X) = \frac{1}{n}\sum_{i=1}^{n} X_i$ (\emph{resp.} $\mathbf{P}(X)=\E[X]$).\\
$\newline$

\subsection{Procedure: Contrast estimator}\label{Procedure}

Hereafter, we propose explicit estimators of the parameter $\theta_0$. This estimator called the contrast estimator is based on minimization of suitable functions of the observations usually called ``contrasts functions". We refer the reader to \cite{van} for a general account on this notion. Furthermore, in this part, we use the contrast function proposed by \cite{CfLcRy10}, that is: 
 
\begin{equation}\label{contraste}
\mathbf{P}_{n}m_{\theta}=\frac{1}{n}\sum_{i=1}^{n}m_{\theta}(\Y_i),
\end{equation} with $n$ the number of observations and:
\begin{equation*}
m_{\theta}(\y_{i}): (\theta, \y_{i})\in (\Theta \times \mathbb{R}^{2})\mapsto m_{\theta}(\y_{i})=||l_{\theta}||^{2}_{2}-2y_{i+1}u^{*}_{l_{\theta}}(y_{i}),
\end{equation*}
where the function $l_{\theta}$ and $u_{v}$ are given by:

\begin{equation}
l_{\theta}(x)=b_{\phi}(x)f_{\theta}(x) \text{ and }\quad u_{v}(x)=\frac{1}{2\pi}\frac{v^{*}(-x)}{f_{\varepsilon}^{*}(x)}
\end{equation}

with $f_{\theta}$ the invariant density of $X_i$.\\

\textbf{Some assumptions}. As our procedure relies on the Fourier deconvolution strategy, in order to construct our estimator, we assume that the density of the noise $\varepsilon_i$, denoted by $f_\varepsilon$, is fully known and  belongs to $\mathbb{L}_2(\mathbb{R})$, and for all $x \in \mathbb{R}$ $f^{*}_{\varepsilon}(x)\neq 0$. Furthermore, we assume that the function $l_{\theta}$ belongs to $\mathbb{L}_1(\mathbb{R})\cap \mathbb{L}_2(\mathbb{R})$.  The function $u_{l_{\theta}}$ must be integrable.\\
For the statistical study, the key assumption is that the process $(X_i)_{i \geq 1}$ is stationary and ergodic  (see \cite{JT00} for a definition).\\

\begin{remark}
In this paper we consider the situation in which the observation noise variance is known. This assumption which is not in general the case in practice is necessary for the identifiability of the model and so is a standard assumption for state-space models given in (\ref{mod1}).\\
There is some restrictions on the distribution of the observation and process errors in the Nadaraya-Watson approach. It is known that the rate of convergence for estimating the function $l_{\theta}$ is related to the rate of decrease of $f^{*}_{\varepsilon}$. In particular, the smoother $f_{\varepsilon}$, the slower the rate of convergence for estimating is (The Gaussian, log-chi squared or Cauchy distributions are super-smooth. The Laplace distribution is ordinary smooth). This rate of convergence can be improved by assuming some additional regularity conditions on $l_{\theta}$.  There is a good discussion about this subject in \cite{CfLcRy10} and \cite{RLA11}.

\end{remark}

\textbf{The procedure} Let us explain the choice of the contrast function and how the strategy of deconvolution works. Obviously, as the model (\ref{mod1}) shows, the $\Y_i$ are not i.i.d. However, by assumption, they are stationary ergodic, so the convergence of $\mathbf{P}_{n}m_{\theta}$ to $\mathbf{P}m_{\theta}=\E\left[m_{\theta}(\Y_{1})\right]$ as $n$ tends to the infinity is provided by the Ergodic Theorem. Moreover, the limit $\E\left[m_{\theta}(\Y_{1})\right]$ of the contrast function can be explicitly computed: 

\begin{equation*}
\E\left[m_{\theta}(\Y_{1})\right]=\left\|l_{\theta}\right\|_{2}^{2}-2\E\left[Y_{2}u_{l_{\theta}}^{*}(Y_{1})\right].
\end{equation*} By Eq.(\ref{mod1}) and under the independence assumptions of the noise $(\varepsilon_2)$ and $(\eta_2)$, we have:

\begin{equation*}
\E\left[ Y_{2}  u^*_{l_{\theta}}(Y_{1})\right] = \E\left[ b_{\phi_{0}}(X_1)  u^*_{l_{\theta}}(Y_{1})\right].
\end{equation*} Using Fubini's Theorem and Eq.(\ref{mod1}), we obtain:

\begin{eqnarray}
\E\left[ b_{\phi_{0}}(X_1)  u^*_{l_{\theta}}(Y_{1})\right] &=& \E\left[b_{\phi_{0}}(X_1) \int e^{iY_{1}z} u_{l_{\theta}}(z) dz \right]\nonumber\\
&=&\E\left[b_{\phi_{0}}(X_1) \int \frac{1}{2\pi}\frac{1}{f_{\varepsilon}^*(z)}e^{iY_{1}z} (l_{\theta}(-z))^*dz  \right]\nonumber\\
&=&\frac{1}{2\pi} \int \E\left[b_{\phi_{0}}(X_1)e^{i(X_1+\varepsilon_1)z} \right] \frac{1}{f_{\varepsilon}^*(z)} (l_{\theta}(-z))^* dz \nonumber\\
&=&\frac{1}{2\pi} \int \frac{\E\left[ e^{i\varepsilon_1z}\right]}{f_{\varepsilon}^*(z)} \E\left[b_{\phi_{0}}(X_1)e^{iX_1z}\right] (l_{\theta}(-z))^* dz\nonumber\\
&=&\frac{1}{2\pi} \E\left[ b_{\phi_{0}}(X_1) \int e^{iX_1z} (l_{\theta}(-z))^* dz \right]\nonumber\\
&=&\frac{1}{2\pi} \E\left[ b_{\phi_{0}}(X_1) \left((l_{\theta}(-X_1))^{*}\right)^*\right]\nonumber\\
&=& \E\left[ b_{\phi_{0}}(X_1)l_{\theta}(X_1)\right].\nonumber\\
&=&\int b_{\phi_{0}}(x)f_{\theta_{0}}(x)b_{\phi}(x)f_{\theta}(x)dx \nonumber\\
&=&\left\langle l_{\theta}, l_{\theta_{0}} \right\rangle. \label{decon_strategy}
\end{eqnarray} 

Then, 

\begin{eqnarray}\label{pmtheta}
\E\left[m_{\theta}(\Y_{1})\right]&=&\left\|l_{\theta}\right\|_{2}^{2}-2\left\langle l_{\theta}, l_{\theta_{0}} \right\rangle,\\
&=&\left\|l_{\theta}-l_{\theta_0}\right\|_{2}^{2}-\left\|l_{\theta_{0}}\right\|_{2}^{2}. 
\end{eqnarray} 
Under the uniqueness assumption \textbf{(CT)} given just later this quantity is minimal when $\theta$=$\theta_{0}$. Hence, the associated minimum-contrast estimators $\widehat{\theta}_n$ is defined as any solution of:  

\begin{equation}\label{min}
\widehat{\theta}_n=\arg\min_{\theta \in \Theta}\mathbf{P}_{n}m_{\theta}.
\end{equation}

\begin{remark} 

One can see in the deconvolution strategy described in Eq.(\ref{decon_strategy}) that temporal correlation in the observation or latent process errors can be authorized. The procedure still be applicable but the covariance matrix $\Omega_{j-1}(\theta_0)$ for the CLT has not an analytic expression in this case since the use of the Fourier deconvolution approach does not work.
\\

We refer the reader to \cite{Do94} for the proof that if $X_i$ is an ergodic process then the process $Y_i$, which is the sum of an ergodic process with an i.i.d. noise, is again stationary ergodic. Furthermore, by the definition of an ergodic process, if $Y_i$ is an ergodic process then the couple $\Y_i=(Y_i, Y_{i+1})$ inherits the property (see \cite{JT00}).\\
\end{remark}

\subsection{Asymptotic properties of the Contrast estimator}

Our proof holds under the following assumptions. For the reader convenience, we denote by \textbf{(E)} (\emph{resp.} \textbf{(C)} and \textbf{(T)}) the assumptions which serve us for the existence (\emph{resp.} Consistency and Central Limit Theorem). If the same assumption is needed for two results, for example for the existence and the consistency, it is denoted by \textbf{(EC)}.\\

\begin{flushleft}
\textbf{(ECT):} The parameter space $\Theta$ is a compact subset of $\mathbb{R}^{r}$ and $\theta_0$ is an element of the interior of $\Theta$.\\
\textbf{(C):} (Local dominance): $\E\left[\sup_{\theta \in \Theta}\left|Y_{2}u^{*}_{l_{\theta}}(Y_{1})\right|\right]<\infty$.\\
\textbf{(CT):} The application $\theta \mapsto \mathbf{P}m_{\theta}$ admits an unique minimum and its Hessian matrix denoted by $V_{\theta}$ is non-singular in $\theta_0$.\\ 
\textbf{(T):} (Regularity): We assume that the function $l_{\theta}$ is twice continuously differentiable w.r.t $\theta \in \Theta$ for any $x$ and measurable w.r.t $x$ for all $\theta$ in $\Theta$. Additionally, each coordinate of $\nabla_{\theta}l_{\theta}$ and each coordinate of $\nabla^{2}_{\theta}l_{\theta}$ belong to $\mathbb{L}_1(\mathbb{R})\cap \mathbb{L}_2(\mathbb{R})$  and each coordinate of $u_{\nabla_{\theta}l_{\theta}}$ and $u_{\nabla_{\theta}^{2}l_{\theta}}$ have to be integrable as well.\\
\qquad (Moment condition): For some $\delta >0$ and for $j \in \left\{1,\cdots,r\right\}$: 
$$\E\left[\left|Y_{2}u^*_{\frac{\p l_{\theta}}{\p \theta_j}}(Y_{1})\right|^{2+\delta}\right]<\infty.$$\\
\qquad (Hessian Local dominance): For some neighbourhood $\mathcal{U}$ of $\theta_0$ and for $j,k \in \left\{1,\cdots,r\right\}$:
$$\E\left[\sup_{\theta \in \mathcal{U}}\left|Y_{2}u^{*}_{\frac{\p^2 l_{\theta}}{\p \theta_j\p \theta_k}}(Y_{1})\right|\right]<\infty.$$
\end{flushleft}

Let us introduce the matrix:

\begin{equation*}
\Sigma(\theta)=V_{\theta}^{-1} \Omega(\theta) V_{\theta}^{-1'} \text{ with } \Omega(\theta)=\Omega_{0}(\theta)+2 \sum_{j=2}^{+\infty} \Omega_{j-1}(\theta),
\end{equation*}
where $\Omega_{0}(\theta)=\Var\left(\nabla_{\theta}m_{\theta}(\mathbf{Y_{1}})\right)$ and $\Omega_{j-1}(\theta)=\C\left(\nabla_{\theta}m_{\theta}(\mathbf{Y_{1}}),\nabla_{\theta}m_{\theta}(\mathbf{Y_{j}})\right)$

\begin{theorem}\label{MR}
Under our assumptions, let $\widehat{\theta}_{n}$ be the minimum-contrast estimator defined by (\ref{min}). Then:
\begin{equation*}
\widehat{\theta}_{n} \longrightarrow \theta_{0}  \qquad \text{ in probability }\text{ as }n \rightarrow \infty.  
\end{equation*}

Moreover, if $\Y_i$ is geometrically ergodic (see Definition \ref{def_geo} in Appendix \ref{appen}), then:
\begin{equation*}
\sqrt{n}(\widehat{\theta}_{n}-\theta_{0})\rightarrow \mathcal{N}\left(0,\Sigma(\theta_{0})\right)   \qquad  \text{ in law} \text{ as }n \rightarrow \infty. 
\end{equation*}
\end{theorem}

 The following corollary gives an expression of the matrix $\Omega(\theta_0)$ and $V_{\theta_0}$ of Theorem \ref{MR} for the practical implementation:

\begin{corollary}\label{lele}
Under our assumptions, the matrix $\Omega(\theta_0)$ is given by:

$$\Omega(\theta_0) = \Omega_{0}(\theta_0)+ 2\sum_{j=2}^{+\infty} \Omega_{j-1}(\theta_0),$$
where:

\begin{equation*}
\Omega_{0}(\theta_0)= 4\E \left[Y^{2}_{2}\left(u^{*}_{\nabla_{\theta}l_{\theta}}(Y_{1})\right)^{2}\right] -4\E\left[b_{\phi_{0}}(X_1) \left(\nabla_{\theta}l_{\theta}(X_1) \right)\right]\E\left[b_{\phi_{0}}(X_1) \left(\nabla_{\theta}l_{\theta}(X_1) \right)\right] ',\\
\end{equation*}
and, the covariance terms are given by:

\begin{equation*}
\Omega_{j-1}(\theta_0)=4\left[\tilde{C}_{j-1}-\E\left[b_{\phi_{0}}(X_1)\left(\nabla_{\theta}l_{\theta}(X_1)\right)\right]\E\left[b_{\phi_{0}}(X_1)\left(\nabla_{\theta}l_{\theta}(X_1)\right)\right] '\right],
\end{equation*}
where $\tilde{C}_{j-1}=\E\left[b_{\phi_{0}}(X_1)\left(\nabla_{\theta}l_{\theta}(X_1)\right)\left(b_{\phi_{0}}(X_j)\nabla_{\theta}l_{\theta}(X_j)\right)'\right] $ and the differential $\nabla_{\theta}l_{\theta}$ is taken at point $\theta=\theta_0$.\\

Furthermore, the Hessian matrix $V_{\theta_0}$ is given by:

\begin{eqnarray*}
\left(\left[V_{\theta_0}\right]_{j,k}\right)_{1\leq j,k\leq r}&=&2\left( \left\langle \frac{\p l_{\theta}}{\p  \theta_k}, \frac{\p l_{\theta}}{\p \theta_j}\right\rangle \right)_{j,k} \text { at point $\theta=\theta_0$.} 
\end{eqnarray*}
\end{corollary}

Let us now state the strategy of the proof, the full proof is given in Appendix \ref{proof_result}. Clearly, the proof of Theorem \ref{MR} relies on M-estimators properties and on the deconvolution strategy. The existence of our estimator follows from regularity properties of the function $l_{\theta}$ and compactness argument of the parameter space, it is explained in Appendix \ref{EoE}. The key of the proof consists in proving the asymptotic properties of our estimator. This is done by splitting the proof into two parts: we first give the consistency result in Appendix \ref{CoE} and then give the asymptotic normality in Appendix \ref{ANoE}. Let us introduce the principal arguments:\\
	
The main idea for proving the consistency of a M-estimator comes from the following observation: if $\mathbf{P}_{n}m_{\theta}$ converges to $\mathbf{P}m_{\theta}$ in probability, and if the true parameter solves the limit minimization problem, then, the limit of the argminimum $\widehat{\theta}_n$ is $\theta_0$. By using an argument of uniform convergence in probability and by compactness of the parameter space, we show that the argminimum of the limit is the limit of the argminimum. A standard method to prove the uniform convergence is to use \emph{the Uniform Law of Large Numbers} (see Lemma \ref{ULLN} in Appendix \ref{appen}). Combining these arguments with the dominance argument \textbf{(C)} give the consistency of our estimator, and then, the first part of Theorem \ref{MR}.\\

The asymptotic normality follows essentially from Central Limit Theorem for a mixing process (see \cite{Ga04}). Thanks to the consistency, the proof is based on a moment condition of the Jacobian vector of the function $m_{\theta}(\y)$ and on a local dominance condition of its Hessian matrix. 
 To refer to likelihood results, one can see these assumptions as a moment condition of the score function and a local dominance condition of the Hessian.\\

\section{ Applications}\label{Application}

\subsection{ Contrast estimator for the Gaussian AR(1) model with measurement noise:}\label{arr}

Consider the following autoregressive process AR(1) with measurement noise:
  
\begin{equation}\label{cac}
\left\lbrace\begin{array}{ll}
Y_i=X_{i}+\varepsilon_{i}\\
X_{i+1}=\phi_0 X_{i}+\eta_{i+1},
\end{array}
\right.
\end{equation}

The noises $\varepsilon_i$ and $\eta_i$ are supposed to be centered Gaussian randoms with variance respectively $\sigma^2_{\varepsilon}$ and $\sigma^2_0$.  We assume that $\sigma^2_{\varepsilon}$ is known. Here, the unknown vector of parameters is $\theta_0=(\phi_0,\sigma^2_0)$ and for stationary and ergodic properties of the process $X_i$, we assume that the parameter $\phi_0$ satisfies $|\phi_0|<1$ (see \cite{Do94}). The functions $b_{\phi}$ and $l_{\theta}$ are defined by:

\begin{eqnarray*}
&& b_{\phi}(x):\ (x,\theta) \in (\R\times \Theta) \mapsto b_{\phi}(x)=\phi x,\\ 
&& l_{\theta}(x):\ (x,\theta) \in (\R\times \Theta) \mapsto l_{\theta}(x)=b_{\phi}(x)f_{\theta}(x)=\frac{\phi}{\sqrt{2\pi\gamma^{2}}} x \exp\left(-\frac{1}{2\gamma^{2}}x^{2}\right),
\end{eqnarray*}
where $\gamma^{2}=\frac{\sigma^{2}}{1-\phi^{2}}$. The vector of parameter $\theta$ belongs to the compact subset $\Theta$ given by $\Theta = [-1+r; 1-r]\times[\sigma^{2}_{min}; \sigma^{2}_{max}]$ with $\sigma^{2}_{min}\geq \sigma_{\varepsilon}^2+\overline{r}$ where $r$,  $\overline{r}$, $\sigma^{2}_{min}$ and $\sigma^{2}_{max}$ are positive real constants. We consider this subset since by stationary of $X_i$, the parameter $|\phi|<1$ and by construction the function $u^{*}_{l_{\theta}}$ is well defined for $\sigma^{2}> \sigma_{\varepsilon}^2(1-\phi^2)$ with $\phi \in [-1+r; 1-r]$ which is implied by $\sigma^{2}>\sigma_{\varepsilon}^2$. The contrast estimator defined in (\ref{Procedure}) has the following form:

\begin{equation}\label{contraste_application}
\widehat{\theta}_n= \arg \min_{\theta \in \Theta}\left\{ \frac{\phi^{2}\gamma}{4\sqrt{\pi}}-\sqrt{\frac{2}{\pi}}\frac{\phi \gamma^{2}}{n(\gamma^{2}-\sigma^{2}_{\epsilon})^{3/2}}\sum_{j=1}^{n} Y_{j+1} Y_{j}\exp\left(-\frac{1}{2}   \frac{Y^{2}_{j}}{(\gamma^{2}-\sigma^{2}_{\epsilon})}\right)\right\}
\end{equation} 

with $n$ the number of observations. Theorem \ref{MR} applies for $\theta_0=(0.7, 0.3)$ and the corresponding result for the  Gaussian AR(1) model is given in Appendix \ref{AppGauss}. As we already mentioned, Corollary \ref{lele} allows to compute confidence intervals: For all $i=1,2$:

\begin{equation*}
\mathbb{P}\left(\hat{\theta}_{n,i}-z_{1-\alpha/2}\sqrt{\frac{\mathbf{e}_{i}'\Sigma (\hat{\theta}_{n})\mathbf{e}_{i}}{n}}\leq \theta_{0,i}\leq \hat{\theta}_{n,i}+z_{1-\alpha/2}\sqrt{\frac{\mathbf{e}_{i}'\Sigma(\hat{\theta}_{n})\mathbf{e}_{i}}{n}}\right)\rightarrow 1-\alpha, 
\end{equation*}
as $n \rightarrow \infty$ where $z_{1-\alpha/2}$ is the $1-\alpha/2$ quantile of the Gaussian law, $\theta_{0,i}$ is the $i^{th}$ coordinate of $\theta_0$ and $\mathbf{e}_{i}$ is the $i^{th}$ coordinate of the vector of the canonical basis of $\R^2$. The covariance matrix $\Sigma(\hat{\theta}_{n})$ is computed in Lemma \ref{hessienne_application} in Appendix \ref{EFCM}.\\

\subsection{ Contrast estimator for the SV model:}\label{svv}

We consider the following SV model:

\begin{equation}\label{svintro}
\left\lbrace\begin{array}{ll}
R_{i+1}=\exp\left(\frac{X_{i+1}}{2}\right)\xi_{i+1},\\
X_{i+1}=\phi_{0}X_{i}+\eta_{i+1}, 
\end{array}
\right.
\end{equation}

The noises $\xi_{i+1}$ and  $\eta_{i+1} $ are two centered Gaussian random variables with standard variance $\sigma_{\xi}^{2}$ assumed to be known and $\sigma^{2}_{0}$. We assume that $|\phi_0|<1$ and we refer the reader to \cite{JT00} for the mixing properties of this model.\\
 
By applying a log-transformation $Y_{i+1}=\log(R^{2}_{i+1})-\E[\log(\xi^{2}_{i+1})]$ and $\varepsilon_{i+1}=\log(\xi^{2}_{i+1})-\E[\log(\xi^{2}_{i+1})]$, the log-transform SV model is given by:
 
 \begin{equation}\label{svappl}
\left\lbrace\begin{array}{ll}
Y_{i+1}=X_{i+1}+\varepsilon_{i+1}\\
X_{i+1}=\phi_{0}X_{i}+\eta_{i+1}, 
\end{array}
\right.
\end{equation}
 
The Fourier transform of the noise $\varepsilon_{i+1}$ is given by:

$$f^{*}_{\varepsilon}(x)=\frac{1}{\sqrt{\pi}}2^{ix}\Gamma(\frac{1}{2}+ix)e^{-i\mathcal{E}x}$$ where $\mathcal{E}=\E[\log(\xi^{2}_{i+1})]=-1.27$ and $\Var[\log(\xi^{2}_{i+1})]$= $\sigma^{2}_{\varepsilon}=\frac{\pi^2}{2}$. Here, $\Gamma$ represents the gamma function given by:

\begin{equation*}
\Gamma: u\rightarrow \int_{0}^{+\infty}t^{u-1}e^{-t}dt \qquad \forall u \in \mathbb{C} \text { such that } \mathcal{R}_{e}(u)>0.
\end{equation*}

The vector of parameters $\theta=(\phi,\sigma^2)$ belongs to the compact subset $\Theta$ given by $[-1+r; 1-r]\times[ \sigma^{2}_{min} ;  \sigma^{2}_{max}]$ with $r$, $\sigma^{2}_{min}$ and $\sigma^{2}_{max}$ positive real constants. \\

Our contrast estimator (\ref{Procedure}) is given by:

\begin{equation}\label{contraste_application2}
\widehat{\theta}_n= \arg \min_{\theta \in \Theta}\left\{ \frac{\phi^{2}\gamma}{4\sqrt{\pi}}-\frac{2}{n}\sum_{i=1}^{n}Y_{i+1}u^{*}_{l_{\theta}}(Y_i)\right\},
\end{equation} with $u_{l_{\theta}}(y)=\frac{1}{2\sqrt{\pi}}\left(\frac{-i\phi y\gamma^2\exp\left(\frac{-y^2}{2}\gamma^2\right)}{\exp\left(-i\mathcal{E}y\right)2^{i y}\Gamma\left(\frac{1}{2}+i y\right)}\right)$.\\

Theorem \ref{MR} applies for $\theta_0=(0.7, 0.3)$ and by Slutsky's Lemma we also obtain confidence intervals. We refer the reader to Appendix \ref{AppSV} for the proof.

\subsection{Comparison with the others methods}

\subsubsection{QML Estimator}\label{qml}

For the SV model, the QML estimator, proposed independently by Harvey et al.(1994) (see \cite{harvey}) is based on the log-transform model given in (\ref{svappl}). Making as if the $\varepsilon_i$ were Gaussian in the log-transform of the SV model, the Kalman filter \cite{ka} can be applied in order to obtain the quasi likelihood function of $Y_{1:n}=(Y_{1},\cdots, Y_n)$ where $n$ is the sample data length. For the AR(1) and  the log-transform of the SV model, the log-likelihood $l(\theta)$ is given by: 

\begin{equation*}
l(\theta)=\log f_{\theta}(Y_{1:n})=-\frac{n}{2}\log(2\pi)-\frac{1}{2}\sum_{i=1}^{n}\log F_i-\frac{1}{2}\sum_{i=1}^{n}\frac{\nu_{i}^{2}}{F_i},
\end{equation*} where $\nu_{i}$ is the one-step ahead prediction error for $Y_i$, and $F_i$ is the corresponding mean square error. More precisely, the two quantities are given by:

\begin{eqnarray*}
\nu_i=(Y_i-\hat{Y}_i^{-}) \text{ and } F_i=\Var_{\theta}[\nu_i]=P_i^{-}+\sigma^{2}_{\varepsilon},
\end{eqnarray*} where $\hat{Y}_i^{-}=\E_{\theta}[Y_i| Y_{1:i-1}]$ is the one-step ahead prediction for $Y_i$ and $P_i^{-}=\Var_{\theta}[(X_i-\hat{X}_{i}^{-})^{2}]$ is the one-step ahead error variance for $X_i$.\\

Hence, the associated estimator of $\theta_0$ is defined as a solution of:
\begin{equation*}
\hat{\theta}_n=\arg\max_{\theta \in \Theta}l(\theta).
\end{equation*}

Note that this procedure can be inefficient: the method does not rely on the exact likelihood of the $Z_{1:n}$ and approximating the true log-chi-square density by a normal density  can be rather inappropriate (see Figure [\ref{approxim}] below).

\begin{figure}[H]
\includegraphics[width=129mm, height=60mm]{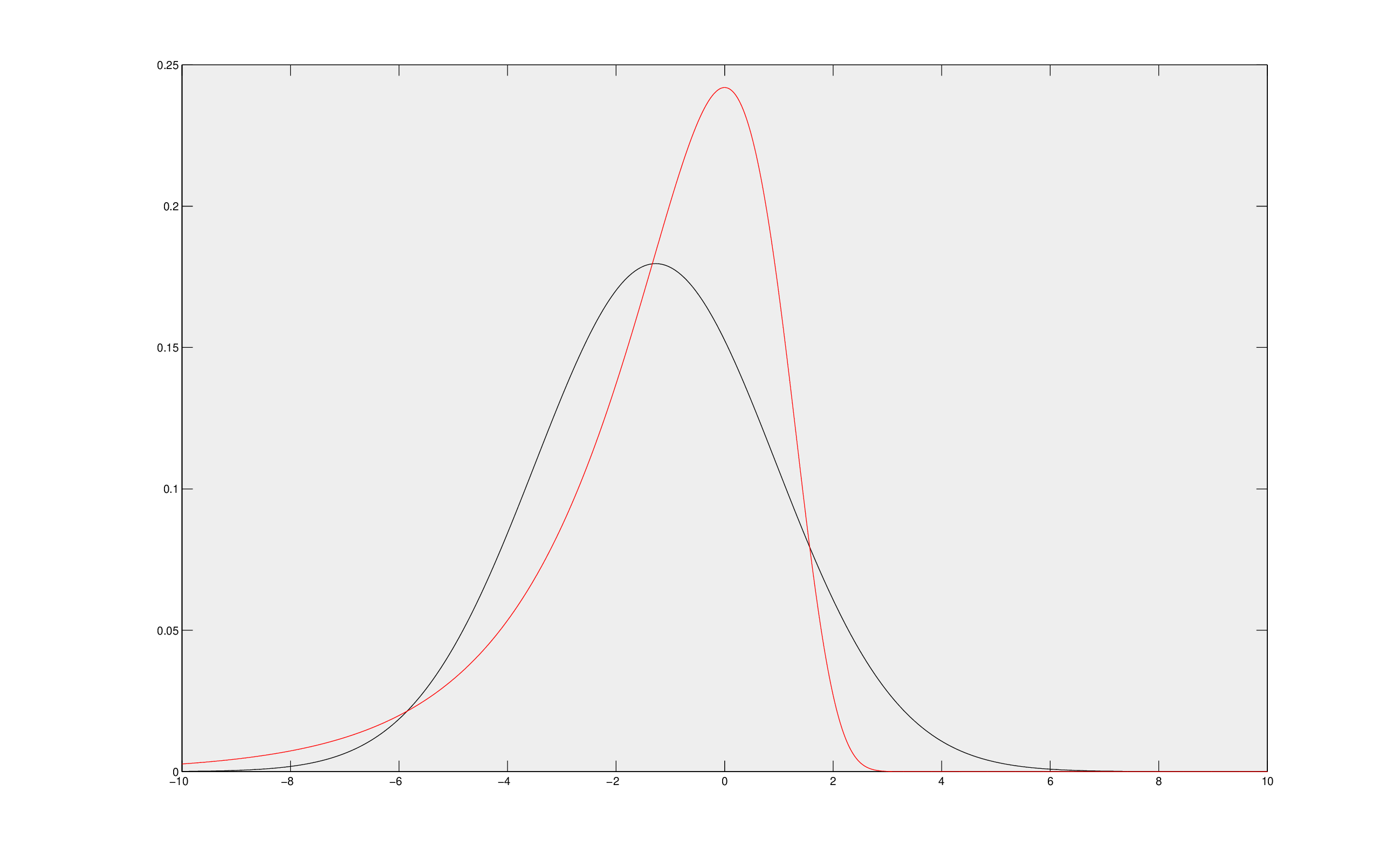}
\caption{\normalsize  Approximation of the $\log$-chi-square density (Red) by a Gaussian density with mean $\mathcal{E}=-1.27$ and variance $\sigma^{2}_{\varepsilon}=\frac{\pi^2}{2}$ (Black). }
\label{approxim} % Give a unique label
\end{figure}

\subsubsection{Particle filters estimators: Bootstrap, APF and KSAPF }\label{bay}

For the particle filters, the vector of parameters $\theta=(\phi,\sigma^2)$ is supposed random obeying the prior distribution assumed to be known. We propose to use the Kitagawa and al.'s approach (see \cite{DoFr01} chapter 10 p.189) in which the parameters are supposed time-varying: $\theta_{i+1}=\theta_{i}+\mathcal{G}_{i+1}$ where $\mathcal{G}_{i+1}$ is a centered Gaussian random with a variance matrix $Q$ supposed to be known. Now, we consider the augmented state vector $\tilde{X}_{i+1}=(X_{i+1}, \theta_{i+1})'$ where $X_{i+1}$ is the hidden state variable and $\theta_{i+1}$ the unknown vector of parameters. In this paragraph, we use the terminology of the particle filtering method, that is: we call particle a random variable. The sequential particle estimation of the vector $\tilde{X}_{i+1}$ consists in a combined estimation of $X_{i+1}$ and $\theta_{i+1}$. For initialisation the distribution of $X_1$ \footnote{ To avoid confusions between the true value $\theta_0$ and the initial value $\theta_1$ in the Bayesian algorithms, we start the algorithms with $i=1$.} conditionally to $\theta_1$ is given by the stationary density $f_{\theta_1}$.  \\

For the comparison with our contrast estimator (\ref{Procedure}), we use the three methods: the Bootstrap filter, the Auxiliary Particle filter (APF) and the Kernel Smoothing Auxiliary Particle filter (KSAPF). We refer the reader
to \cite{DoFr01}, \cite{PiSh07} and \cite{liu} for a complete revue of these methods.

\begin{remark} Let us underline some particularity of  the combined state and parameters estimation: For the Bootstrap and APF estimator, an important issue concerns the choice of the parameter variance $Q$ since the parameter is itself unobservable. If one can choose an optimal variance $Q$ the APF estimator could be a very good estimator since with arbitrary variance the results are acceptable (see Table [\ref{compar2}]). In practice, Q is chosen by an empirical optimization. The KSAPF is an enhanced version of the APF and depends on a smooth factor $0<h<1$ (see \cite{liu}). Therefore, the choice of $h$ is another problem in practice. \\

A common approach to estimate the vector of parameters is to maximize the likelihood. Nevertheless, for state space models, the main difficulty with the Maximum Likelihood Estimator (MLE) \index{MLE,  Maximum Likelihood Estimator} comes from the unobservable character of the state $x_t$ making the calculus of the likelihood untractable in practice: the likelihood is only available in the form of a multiple integral, so exact likelihood methods require simulations and have therefore an intensive computational cost. In many cases, the MLE has to be approximated. A popular approach to approximate it consists in using MCMC simulation techniques (see \cite{SaRg93} and  \cite{cape}). Another approach to approximate the likelihood consists in using particles filtering algorithms. Recently, in \cite{MR2649602} the authors propose an approach of Integrated Nested Laplace Approximations to obtain approximations of the likelihood.\\
In \cite{chopin2} the authors propose a sequential $SMC^{2}$ algorithm which allows an efficient approximation of the complete distribution $p(x_{0:t}, \theta \vert y_{1:t})$. Their approach is an extension of the Iterated Batch Importance Sampling (IBIS) proposed in \cite{MR1929161}. In \cite{MR2758115} the authors develop a general algorithm which is a MCMC algorithm that uses the particles filter to approximate the intractable density $p_{\theta}(y_{1:n})$ combined with a MCMC step that samples from $p(\theta \vert y_{1:n})$. They show that their PMCMC algorithm admits as stationary density the distribution of interest $p(x_{0:t}, \theta \vert y_{1:t})$. There exist others methods and we refer the reader to \cite{MR2416438}, \cite{singh} for more details. \\
\end{remark}

\subsection{A simulation study}\label{implementation}

 For the AR(1) and SV model, we sample the trajectory of the $X_i$ with the parameters $\phi_0=0.7$ and $\sigma_0^{2}=0.3$. Conditionally to the trajectory, we sample the variables $Y_i$  for $i=1\cdots n$ where $n$ represents the number of observations. We take $n=1000$ and $\sigma^{2}_{\varepsilon}=0.1$ for the two models. This means that we consider the following model:
 
\begin{equation*}
\left\lbrace\begin{array}{ll}
R_{i+1}=\exp\left(\frac{X_{i+1}}{2}\right)\xi_{i+1}^{\beta},\\
X_{i+1}=\phi_{0}X_{i}+\eta_{i+1}, 
\end{array}
\right.
\end{equation*}
 
with $\beta=\frac{1}{\sqrt{5}\pi}$. In this case, the Fourier transform of $\varepsilon_{i+1}$ is given by: $f^{*}_{\varepsilon}(y)=\exp\left(-i\tilde{\mathcal{E}}y\right)\frac{2^{i\beta y}}{\sqrt{\pi}}\Gamma\left(\frac{1}{2}+i\beta y\right)$ with $\tilde{\mathcal{E}}=\beta\mathcal{E}$(see Appendix \ref{AppSV}).\\

For the three methods, we take a number of particles $M$ equal to $5000$.  Note that for the Bayesian procedure (Bootstrap, APF and KSAPF), we need a prior on $\theta$, and this only at the first step. The prior for $\theta_1$ is taken to be the Uniform law and conditionally to $\theta_1$ the distribution of $X_1$ is the stationary law:

\begin{equation*}
\left\lbrace\begin{array}{ll}
p(\theta_1)=\mathcal{U}(0.5, 0.9)\times \mathcal{U}(0.1, 0.4)\\
f_{\theta_1}(X_1)=\mathcal{N}\left(0, \frac{\sigma^{2}_{1}}{1-\phi_1^2}\right)
\end{array}
\right.
\end{equation*}

We take $h=0.1$ for the KSAPF and $Q=\begin{pmatrix}
0.6.10^{-6} & 0\\
0 & 0.1.10^{-6} 
\end{pmatrix}$ for the APF and Bootstrap filter. 	

\begin{remark}
Note that, in practice,  there is no constraint on the parameters for the  contrast function contrary to the particle filters where we take the stationary law for $p_{\theta}(X_0)$ and the Uniform law around the true parameters. Hence, we bias favourably the particle filters.  
\end {remark}	 

\subsection{Numerical Results}

In the numerical section we compare the different estimations: the QML estimator defined in Section \ref{qml}, the Bayesian estimators defined in Section \ref{bay} and our contrast estimator defined in Section \ref{Procedure}. For the comparison of the computing time, we also compare our contrast estimator with the SIEMLE proposed by Kim, Shepard and Chib (see Appendix \ref{siem} and \cite{shep}  for more general details).

\subsubsection{Computing time}\label{comp-time}

From a theoretical point of view, the MLE is asymptotically efficient. However, in practice since the states $(X_{1}\cdots, X_{n})$ are unobservable and the SV model is non Gaussian, the likelihood is untractable. We have to use numerical methods to approximate it. In this section, we illustrate the SIEMLE which consists in approximating the likelihood and applying the Expectation-Maximisation algorithm introduced by Dempster \cite{dempster} to find the parameter $\theta$.\\
To illustrate the SIEMLE for the SV model, we run an estimator with a number of observations $n$ equal to $1000$.  Although the estimation is good the computing time is very long compared with the others methods (see Tables [\ref{time_gaussian}] and [\ref{computing_siem}]). This result illustrates the numerical complexity of the SIEMLE (see Appendix \ref{siem}). Therefore, in the following, we only compare our contrast estimator with the QML and Bayesian estimators. The results are illustrated by Figure [\ref{time_gaussian}]. We can see that our contrast estimator is the fastest for the Gaussian AR(1) model. The QML is the most rapid for the SV model since it assumes that the measurement errors are Gaussian but we show in Figures [\ref{boxplot_compa1}], [\ref{boxplot_compa2}] and [\ref{compar2}] that it is a biased estimator with large mean square error. For our algorithm, for the Gaussian AR(1)  model, the function $u^{*}_{l_{\theta}}$ has an explicit expression but for the SV model, the function $u^{*}_{l_{\theta}}$ is approximated numerically since the Fourier transform of the function $u_{l_{\theta}}$ has not an explicit form. This explains why our algorithm is slower on the SV model than on the Gaussian AR(1) model.\footnote{We use a quadrature method implemented in Matlab to approximate the Fourier transform of $u_{l_{\theta}}(y)$. One can also use the FFT method and we expect that the contrast estimator will be more rapid in this case.}  In spite of this approximation, our contrast estimator is fast and its implementation is straightforward.

\begin{table}[H]
\caption{\normalsize  \label{time_gaussian} Comparison of the computing time (CPU in seconds) and the MSE with respect to the number of observations $n=200$ up to $1500$ for the Gaussian AR(1) and the SV models. The number of particles in Bayesian estimations is $M=5000$ particles and the number of estimators is $N=100$ for the MSE (see Eq.(\ref{MSE})).}
\begin{center}
\begin{tabular}{lllllll}
\hline\noalign{\smallskip}
 & n & SV &  &  & AR(1) &   \\
\noalign{\smallskip}\hline\noalign{\smallskip}
&  & CPU & MSE & CPU & MSE  \\
\noalign{\smallskip}\hline\noalign{\smallskip}
Contrast & 200 & 4.2695  &  0.0425 & 0.032146 &0.0411\\
& 300 &  5.1015  &  0.0453 &  0.022588 &0.0398 \\
& 400 & 7.0502   & 0.0239 & 0.028062 & 0.0374\\
   & 500 & 6.9109   & 0.0175 & 0.026517 & 0.0306\\
   & 750 & 11.8555  &  0.0117 & 0.031353 & 0.0218\\
   & 1000 & 20.4074  &  0.0078 & 0.056931 & 0.0133\\
   & 1500 & 29.3910  &  0.0061 & 0.08432 & 0.0091 \\
\noalign{\smallskip}\hline
\noalign{\smallskip}\hline
Bootstrap filter & 200 & 41.4780  &  0.0275 & 85.65 & 0.0225\\
 & 300 &    57.5201  &  0.0261 & 103.7212 & 0.0211\\
 & 400 &  67.9421   & 0.0248 & 155.0456 & 0.0199\\
 & 500 & 107.9450  &  0.0228 & 169.5578 & 0.0187\\
 & 750 & 138.0307  &  0.0186 & 241.1891 & 0.0154\\
 & 1000 & 192.2166 &  0.0174 & 318.5656 & 0.0133\\
 & 1500 & 158.3680   & 0.0166 & 388.7098 & 0.0122\\
 \noalign{\smallskip}\hline
\noalign{\smallskip}\hline
APF  & 200 &  19.4471  &  0.0209 &  49.6784 & 0.0138 \\
 & 300 &    39.2457   & 0.0182 &   69.3421 & 0.0125 \\
  & 400 & 46.9590   & 0.0123 &   86.9111 & 0.0118\\
 & 500 &  54.5811  &  0.0189 & 108.9087 & 0.0112\\
 & 750 &  91.5288   & 0.0171 & 166.3432 & 0.0100\\
 & 1000 & 105.1695  &  0.0163 & 189.5432 & 0.0087\\
  & 1500 &  122.1278  &  0.0159 & 326.7654 & 0.0074\\
\noalign{\smallskip}\hline
\noalign{\smallskip}\hline
KSAPF & 200 & 32.8328   & 0.0131 & 55.039200 & 0.0121\\
  & 300 & 47.4919    & 0.0129 & 90.691115 & 0.0116\\
  & 400 & 58.3216 & 0.0118 & 107.767974 & 0.110\\
  & 500 & 66.3554   & 0.0114 & 127.565273 & 0.102\\
  & 750 & 76.4818  &  0.0103 & 173.311428 & 0.0086 \\
  & 1000 & 93.8846   &  0.0093 & 246.09729 & 0.0073\\
  & 1500 &  151.7971   &  0.0084 & 376.8976 & 0.0068\\
\noalign{\smallskip}\hline
\noalign{\smallskip}\hline
QML & 200 &  0.0268 & 0.172 & 0.0283 & 0.0444 \\
    & 300 & 0.0201 & 0.164 & 0.0312 & 0.0331\\
    & 400 & 0.0532 & 0.153 & 0.0386 & 0.0336\\
    & 500 & 0.0675 & 0.146 & 0.0476 & 0.0327 \\
   & 750 &  0.1046 & 0.132 & 0.0631 & 0.0311\\
    & 1000 & 0.0702 & 0.118 & 0.0712 & 0.0278 \\
    & 1500 & 0.2148 & 0.110  & 0.0854 & 0.0253\\
\noalign{\smallskip}\hline
\end{tabular}
\end{center}
\end{table}

\newpage
\begin{table}[H]
\caption{\label{computing_siem}\normalsize SIEMLE estimation for the SV model. The number of observations is $n=1000$ and the number of sweeps for the Gibbs sampler is $\tilde{M}=100$ (see Appendix \ref{siem}). }
\begin{center}
\begin{tabular}{llllll}
\hline\noalign{\smallskip}
$\phi_0$ & $\sigma^{2}_0$ & $\hat{\phi}_n$ & $\hat{\sigma}^{2}_n$ & CPU (sec)   \\
\noalign{\smallskip}\hline\noalign{\smallskip}
 0.7 & 0.3 & 0.667 & 0.2892& 74300  \\
\noalign{\smallskip}\hline
\end{tabular}
\end{center}
\end{table}

\subsubsection{Parameter estimates}\label{co}

For the AR(1) Gaussian model, we run $N=1000$ estimates for each method (QML, APF, KSAPF and Bootsrap filter) and $N=500$ for the SV model. The number of observations $n$ is equal to $1000$ for the two models.\\
In order to compare with others the performance of our estimator, we compute for each method the Mean Square Error (MSE) defined by:

\begin{equation}\label{MSE}
MSE=\frac{1}{N}\left(\sum_{j=1}^{N}(\hat{\phi}_{j}-\phi_0)^{2}+(\hat{\sigma}^{2}_{j}-\sigma^{2}_0)^{2}\right),
\end{equation}

We  illustrate by boxplots the different estimates (see Figures [\ref{boxplot_compa1}] and [\ref{boxplot_compa2}]). We also illustrate in Figure [\ref{compar2}] the MSE for each estimator computed by equation(\ref{MSE}). We can see that, for the parameter $\phi_0$, the QML estimator is better for the Gaussian AR(1) model than for the SV model (see Figure [\ref{boxplot_compa1}]). Indeed, the Gaussianity assumption is wrong for the SV model. Moreover, the estimate of $\sigma^{2}_0$ by QML is very bad for the two models (see Figure [\ref{boxplot_compa2}]) and its corresponding boxplots have the largest dispersion meaning that the QML method is not very stable. The Bootstrap, APF and KSAPF have also a large dispersion of their boxplots, in particular for the parameter $\phi_0$ (see Figure [\ref{boxplot_compa1}]). Besides, the Booststrap filter is less efficient than the APF and KSAPF. 
For the Gaussian and SV model, the boxplots of our contrast estimator show  that our estimator is the most stable with respect to $\phi_0$ and we obtain similar results for $\sigma^{2}_0$. The MSE is better for the SV model and the smallest for our contrast estimator.

\newpage

\begin{figure}[H]
\includegraphics[width=174mm, height=90mm]{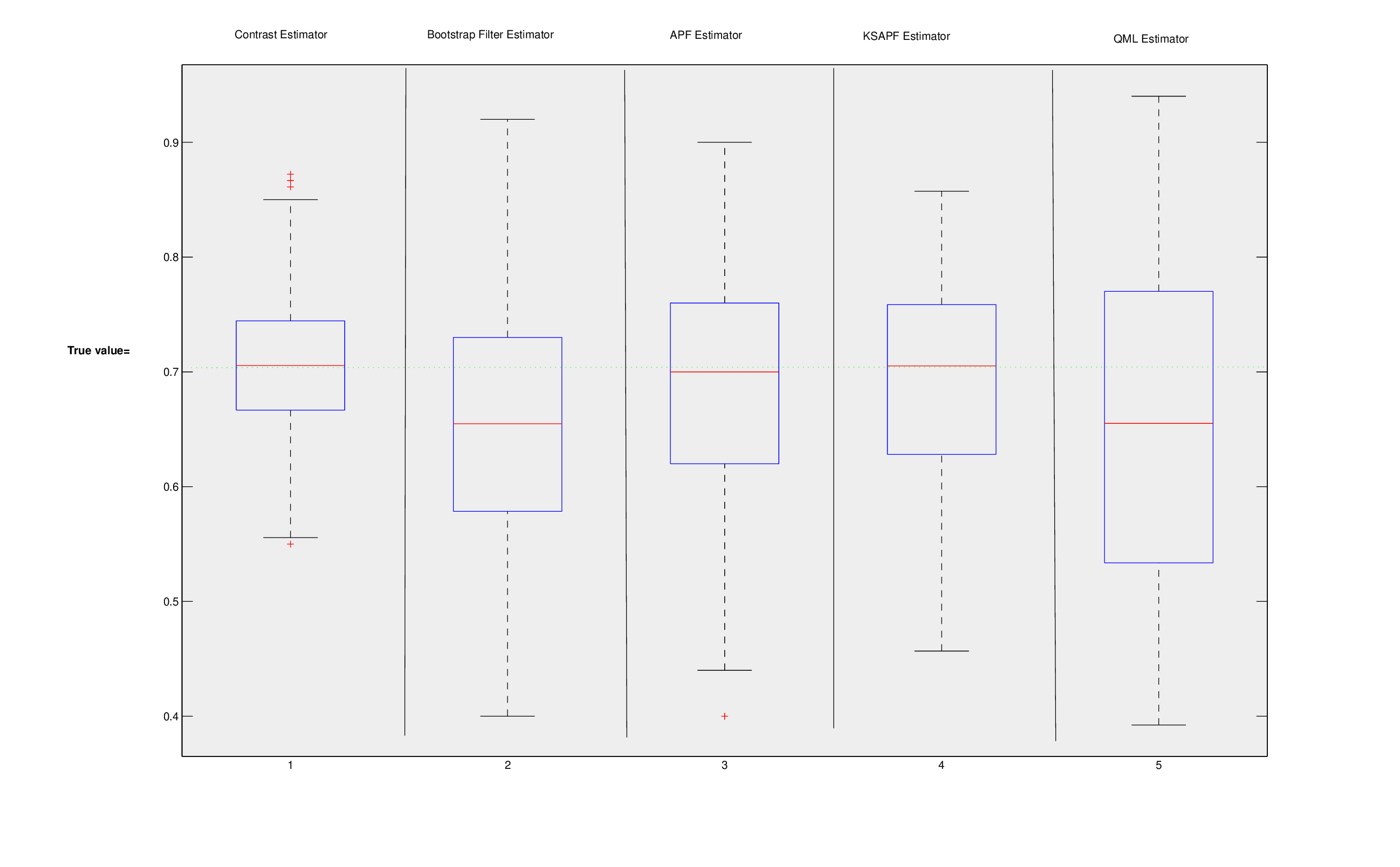}\\
\includegraphics[width=174mm, height=90mm]{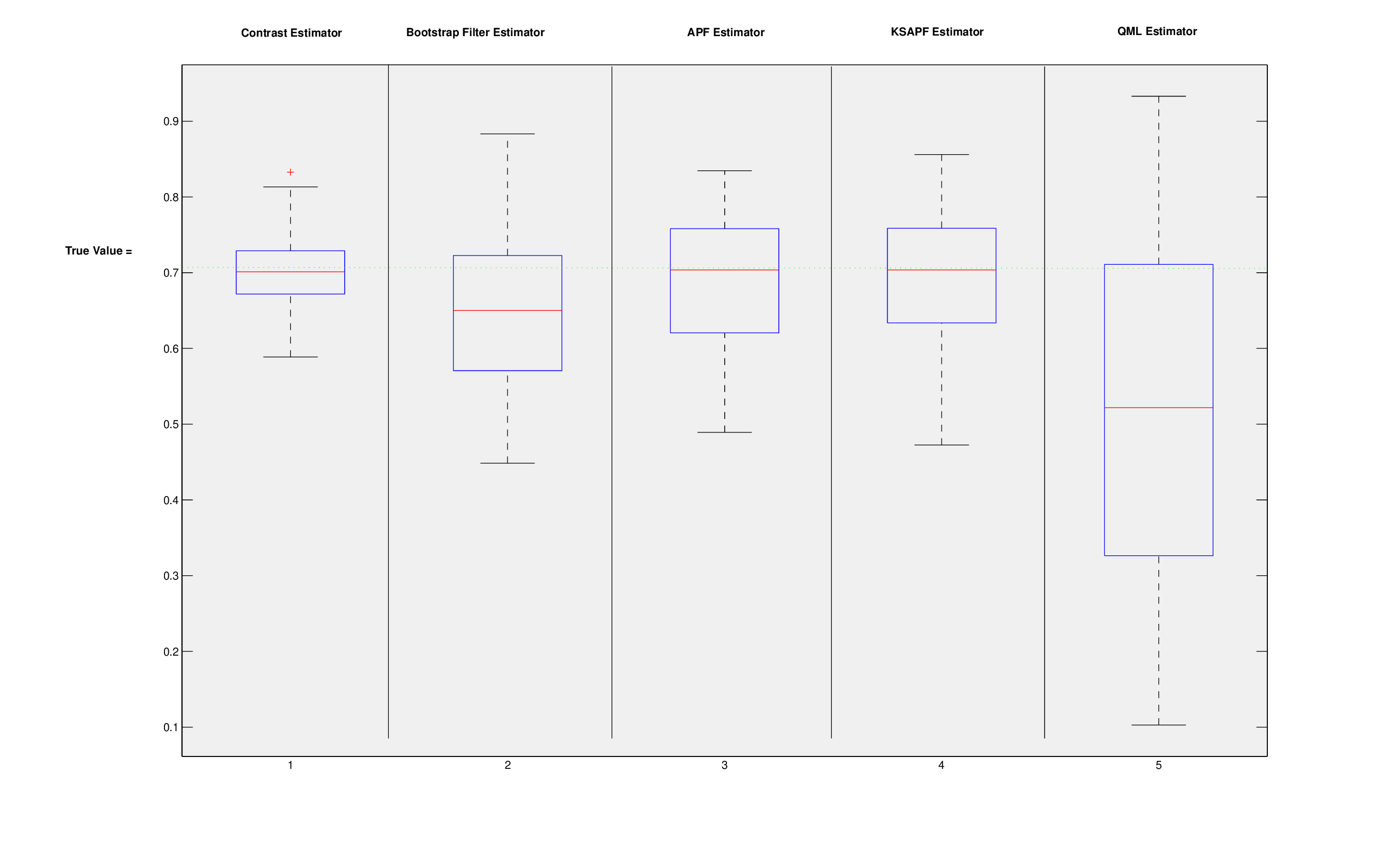}
\caption{\normalsize Boxplot of $\phi$. True value: $\phi_0=0.7$. Top Panel: Gaussian AR(1) model. Bottom Panel: SV model.}
\label{boxplot_compa1}      % Give a unique label
\end{figure}

\begin{figure}[H]
\includegraphics[width=174mm, height=90mm]{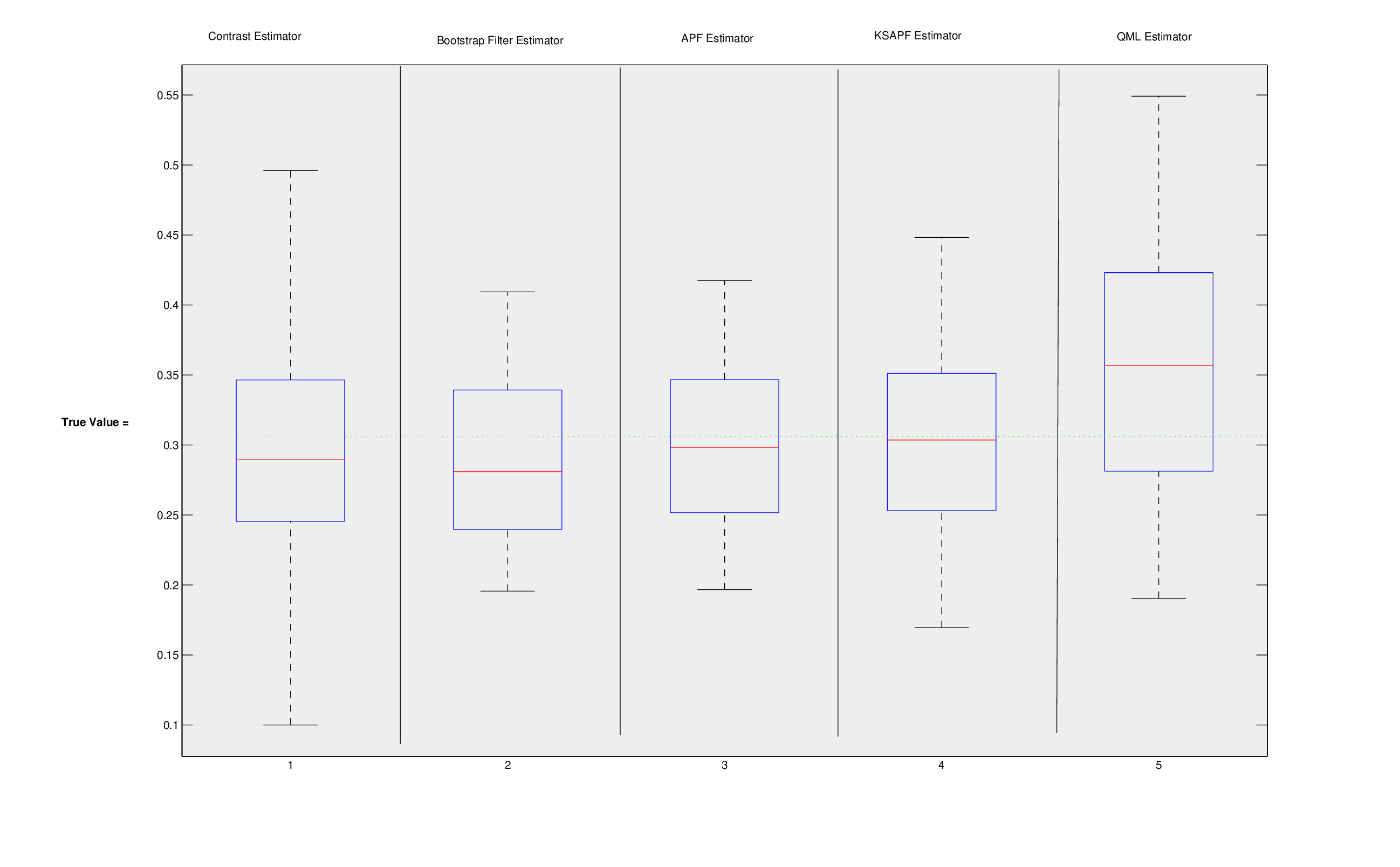}\\
\includegraphics[width=174mm, height=90mm]{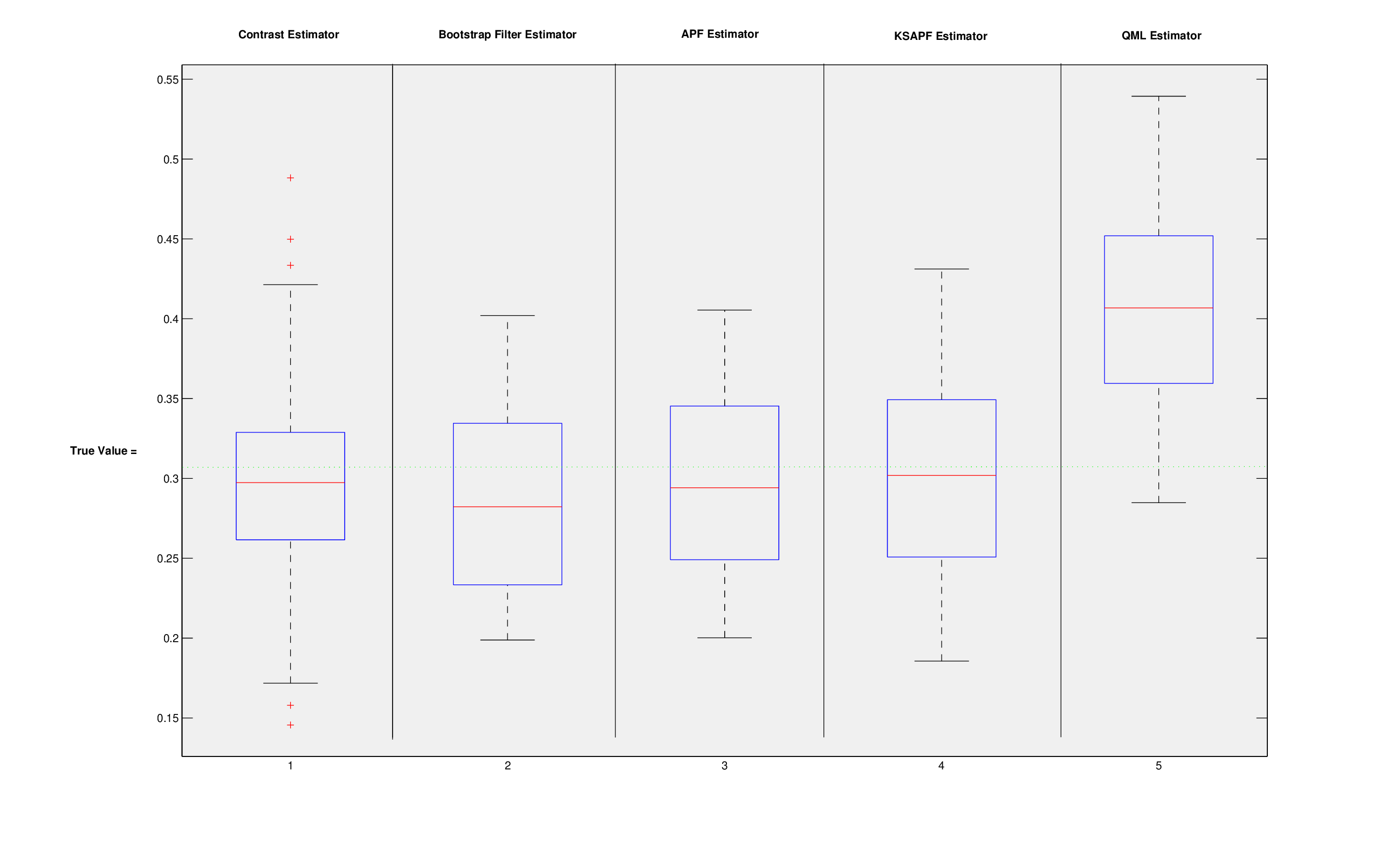}
\caption{\normalsize Boxplot of $\sigma^2$. True value: $\sigma^{2}_0=0.3$. Left: Gaussian AR(1) model. Right: SV model.}
\label{boxplot_compa2}          % Give a unique label
\end{figure}

\begin{figure}[H]
\includegraphics[width=129mm, height=40mm]{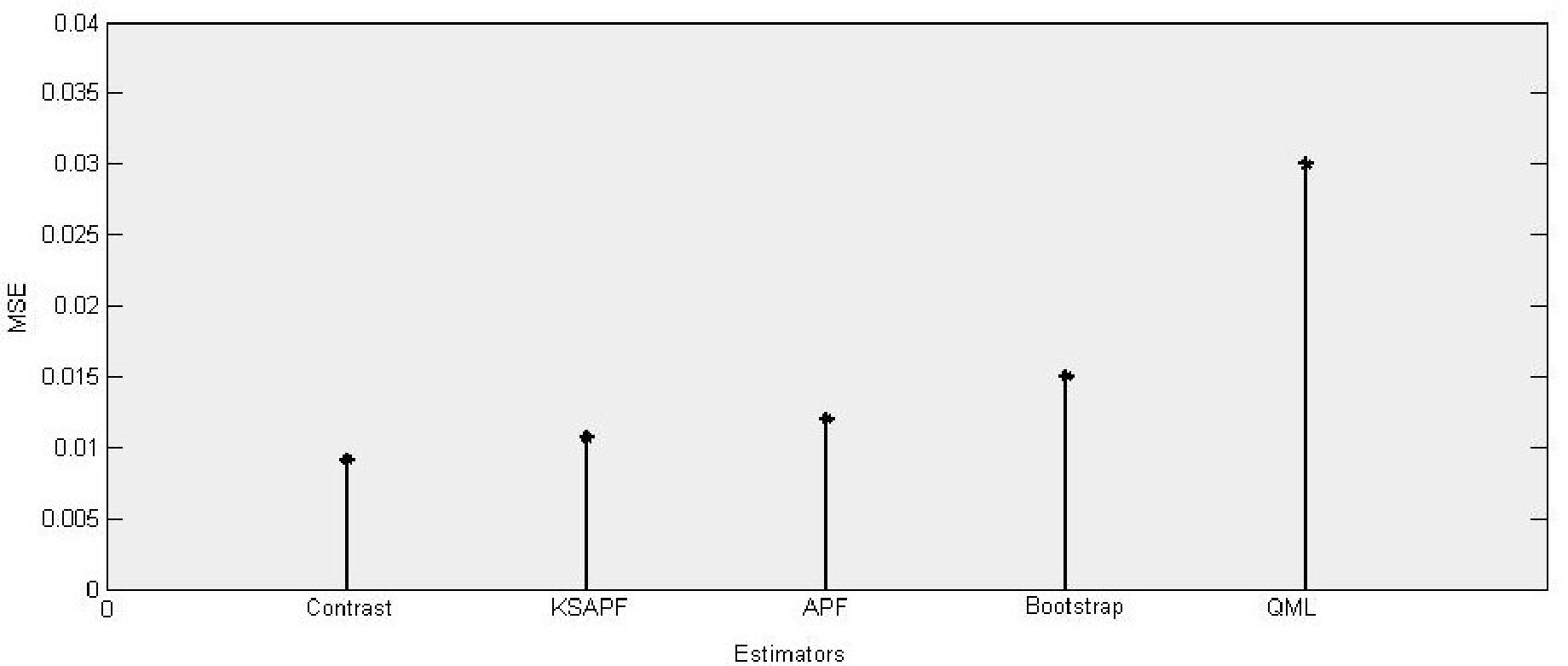}\\
\includegraphics[width=129mm, height=40mm]{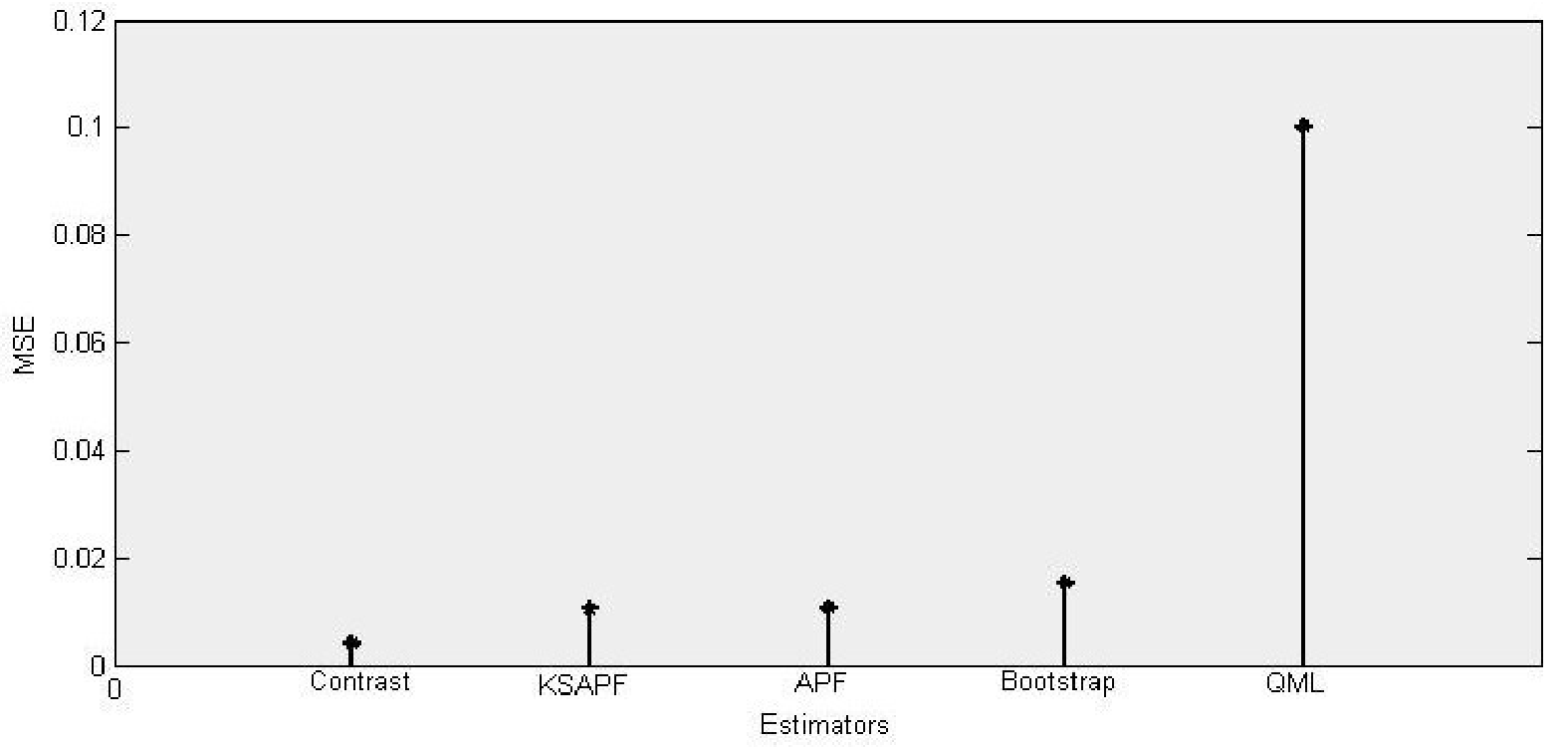}
\caption{\normalsize MSE computed by Eq.(\ref{MSE}). Top Panel: Gaussian AR(1) model. Bottom Panel: SV model.}  
\label{compar2} 
\end{figure}

\subsubsection{Confidence Interval of the contrast estimator}

To illustrate the statistical properties of our contrast estimator, we  compute for each model the confidence intervals computed with the confidence level $1-\alpha$ equal to $0.95$ for $N=1$ estimator and the coverages for $N=1000$ with respect to the number of observations. The coverage corresponds to the number of times for which the true parameter $\theta_{0,i}, i=1,2$ belongs to the confidence interval. The results are illustrated by the Figures [\ref{Covera1}]-[\ref{ic1}] and [\ref{ic2}]: for the Gaussian and SV models, the coverage converges to $95\%$ for a small number of observations.  As expected, the confidence interval decreases with the number of observations,. Note that of course a MLE confidence interval would be smaller since the MLE is efficient but the corresponding computing time would be huge. 

\begin{figure}[H]
\includegraphics[width=129mm, height=35mm]{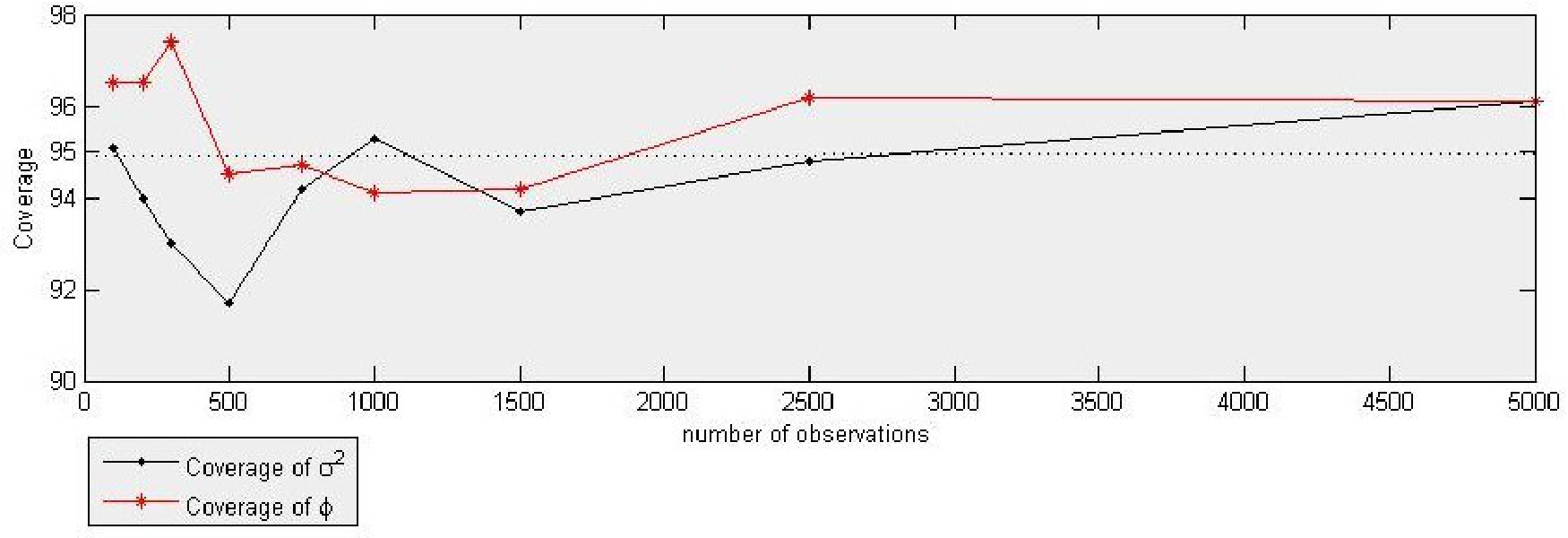}\\
\includegraphics[width=129mm, height=35mm]{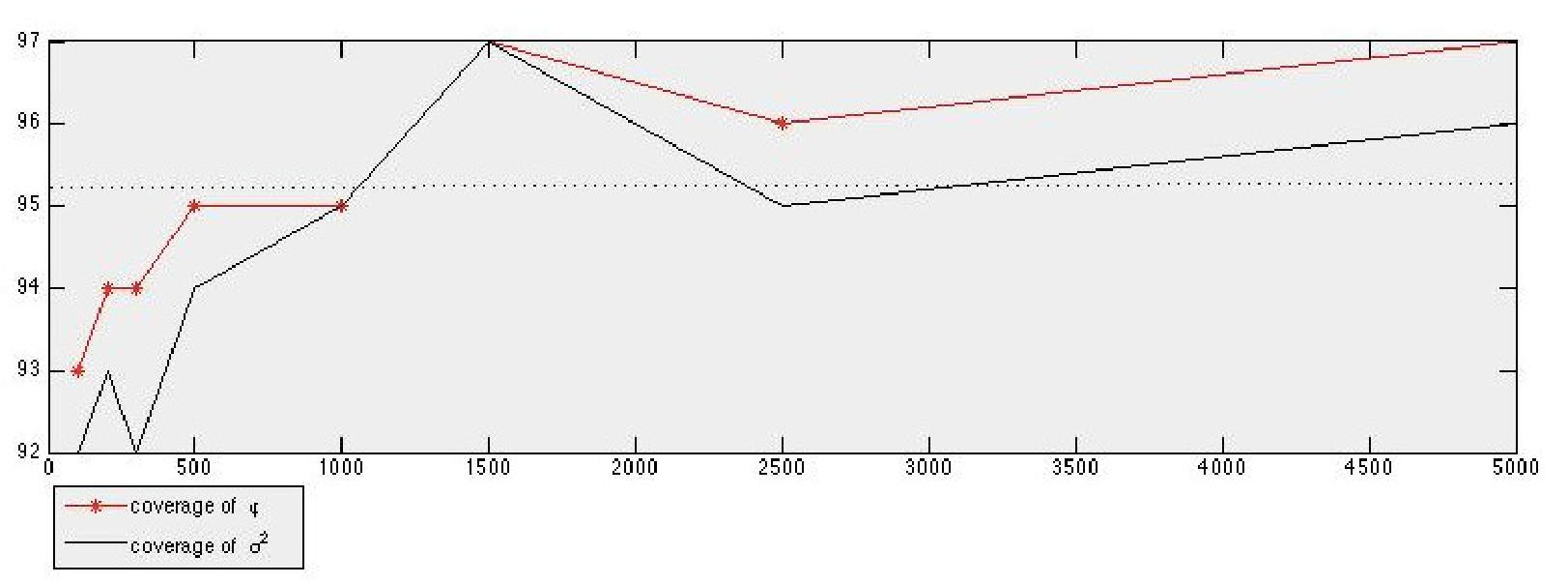}
\caption{\normalsize Coverage with respect to the number of observations $n=100$ up to $5000$ for $N=1000$ estimators . Top Panel: Gaussian AR(1) model. Bottom Panel: SV model.}
\label{Covera1}       
\end{figure}

\begin{figure}[H]
\includegraphics[width=129mm, height=35mm]{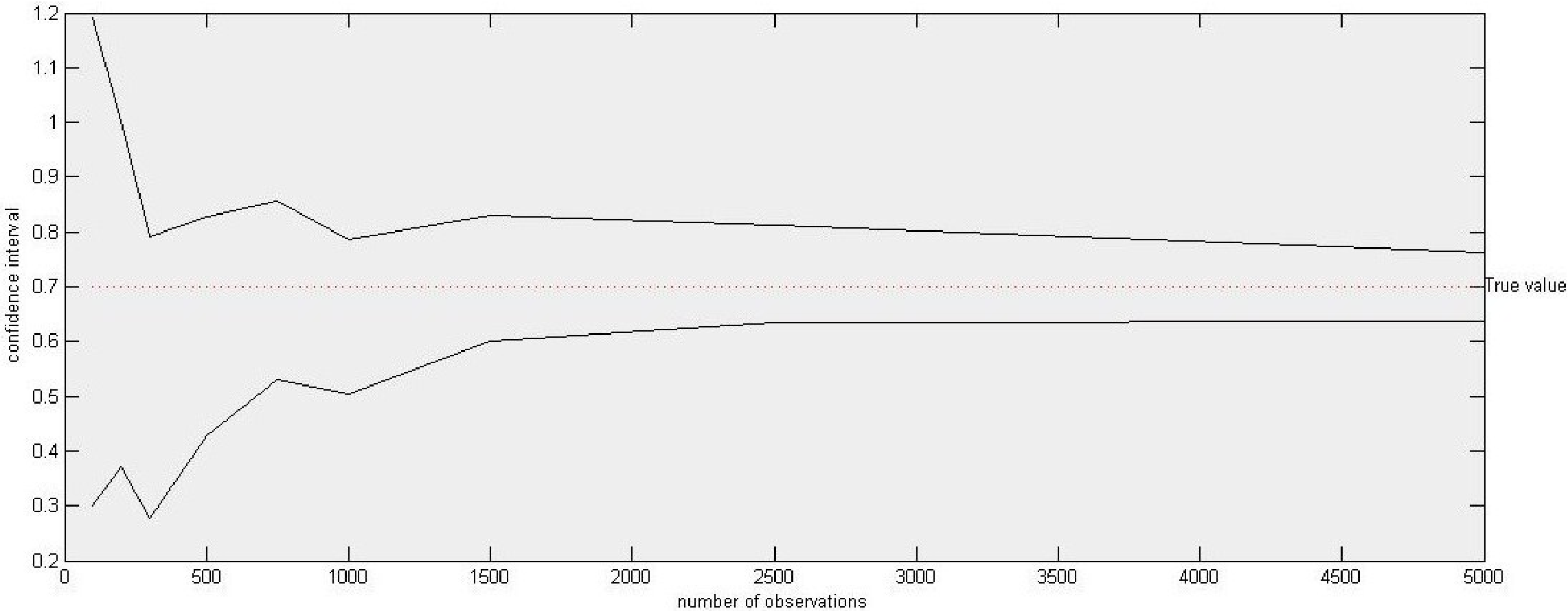}\\
\includegraphics[width=129mm, height=35mm]{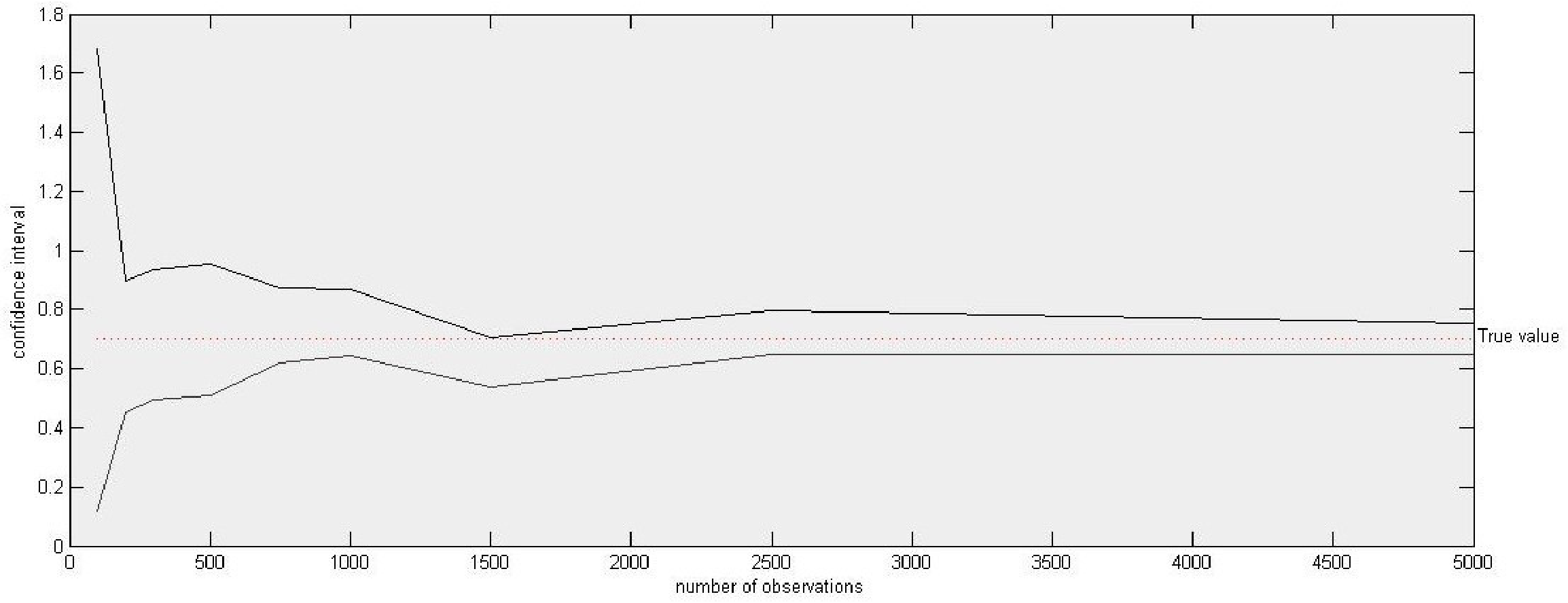}
\caption{\normalsize Confidence interval for the parameter $\phi_0$ with respect to the number of observations $n=100$ up to $5000$ for $N=1$ estimator. Top Panel: Gaussian AR(1) model. Bottom Panel: SV model.}
\label{ic1}      
\end{figure}

\begin{figure}[H]
\includegraphics[width=129mm, height=35mm]{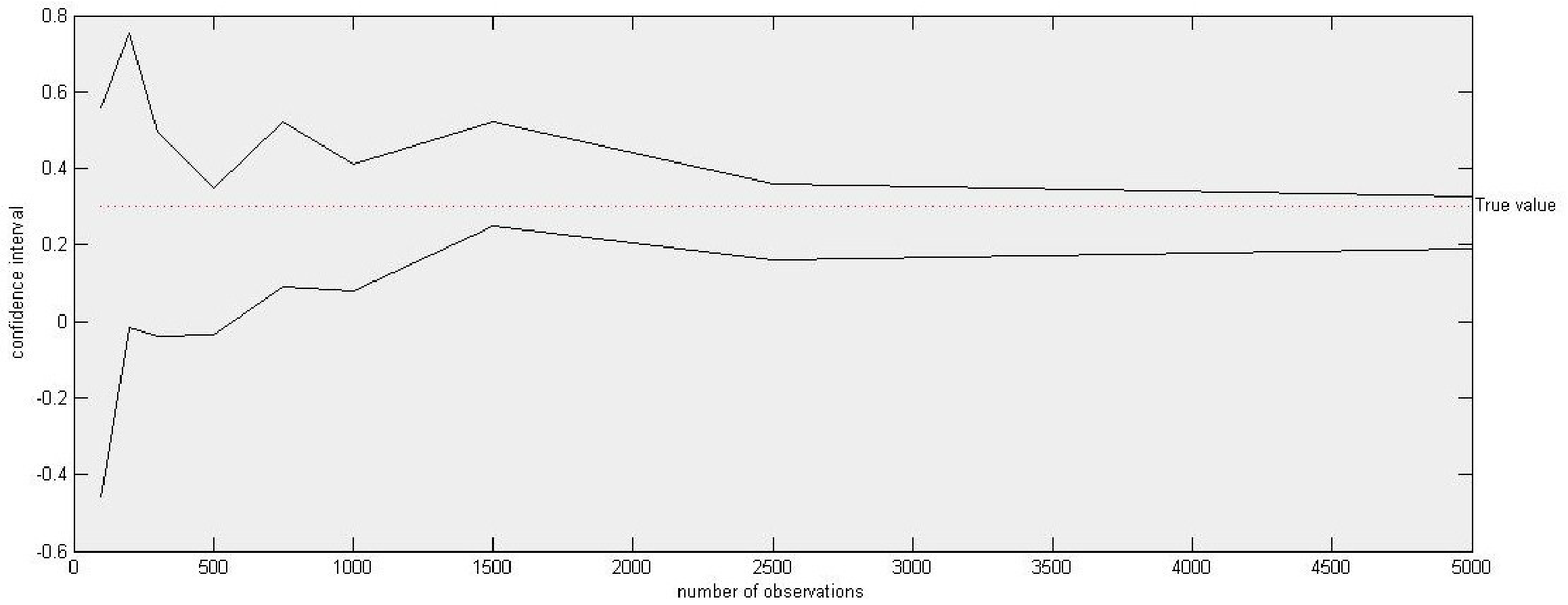}\\
\includegraphics[width=129mm, height=35mm]{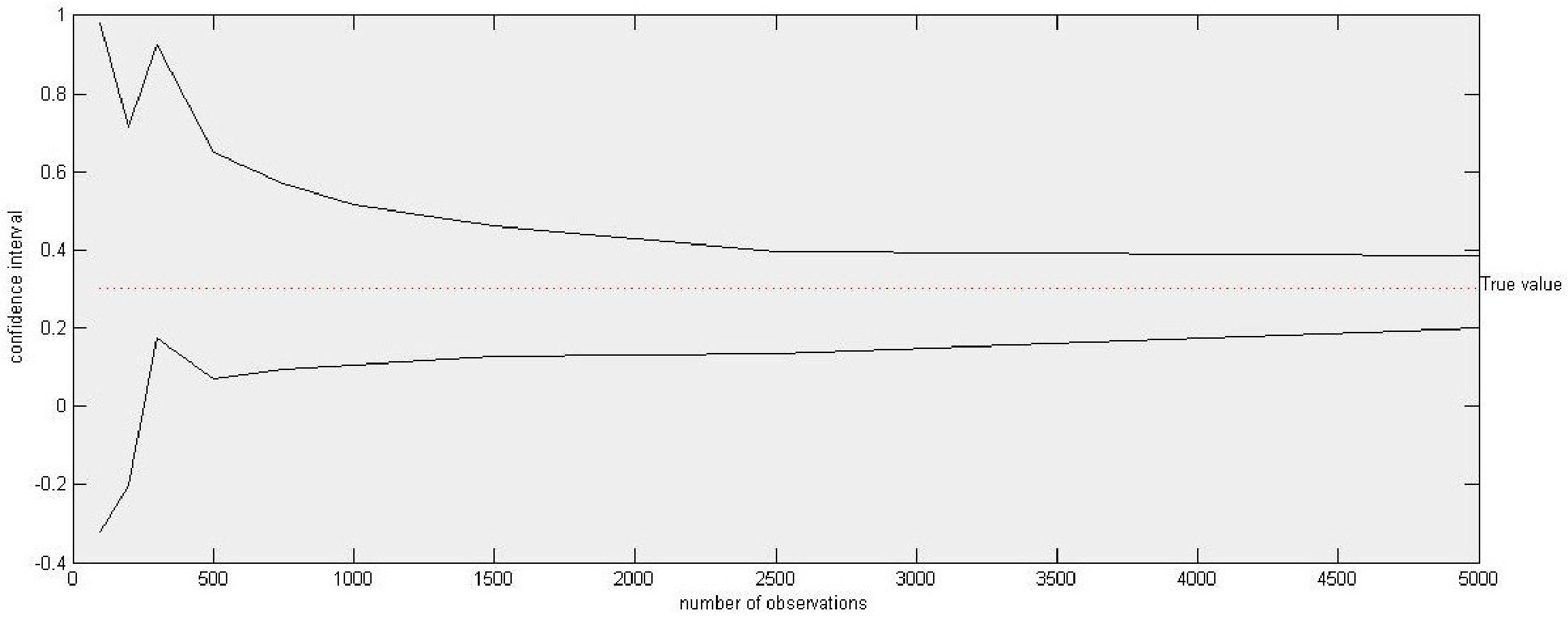}
\caption{\normalsize Confidence interval for the parameter $\sigma^2_0$ with respect to the number of observations $n=100$ up to $5000$ for $N=1$ estimator. Top Panel: Gaussian AR(1) model. Bottom Panel: SV model.}
\label{ic2}      
\end{figure}

\subsection{Application to Real Data}\label{Application_real_data}

The data consist of daily observations on FTSE stock price index and S$\&$P500 stock price index. The series taken in boursorama.com are closing prices from January, 3, 2004 to January, 2, 2007 for the FTSE and S$\&$P500 leaving a sample of $759$ observations for the two series.\\
The daily prices $S_i$ are transformed into compounded rates returns centered around their sample mean $c$ for self-normalization (see \cite{mat} and \cite{GeHaRe96}) $R_i=100\times\log\left(\frac{S_i}{S_{i-1}}\right)-c$. We want to model those data by the SV model defined in (\ref{svappl}) leading to :
\begin{eqnarray*}
Y_i&=&\log(R_i^2)-\E[\log(\xi_i^2)]\\
&=&\log(R_i^2)+1.27
\end{eqnarray*} Those data are represented on Figure [\ref{RES}]. 

\begin{figure}[H]
\includegraphics[width=66mm, height=70mm]{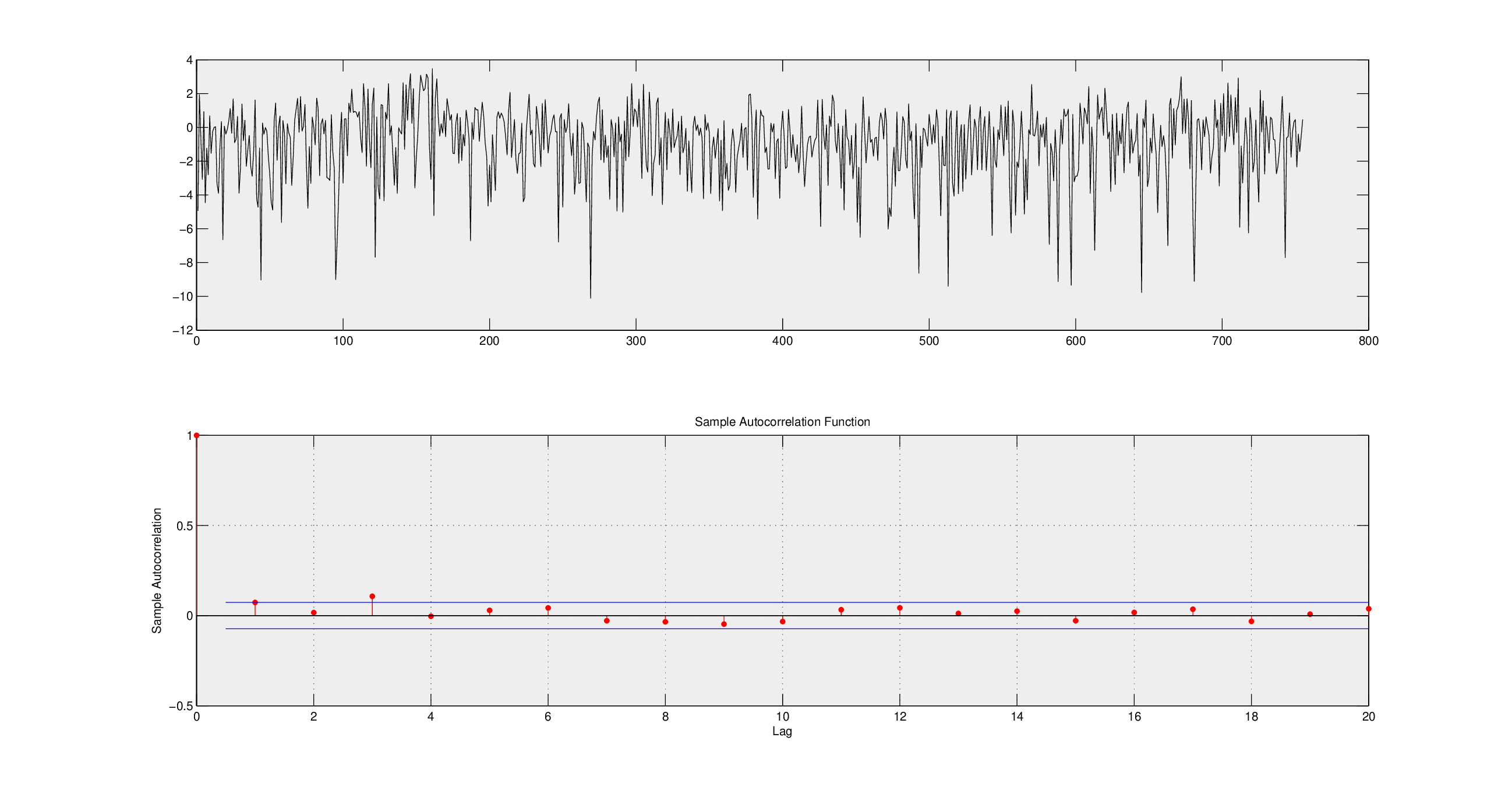}
\includegraphics[width=66mm, height=70mm]{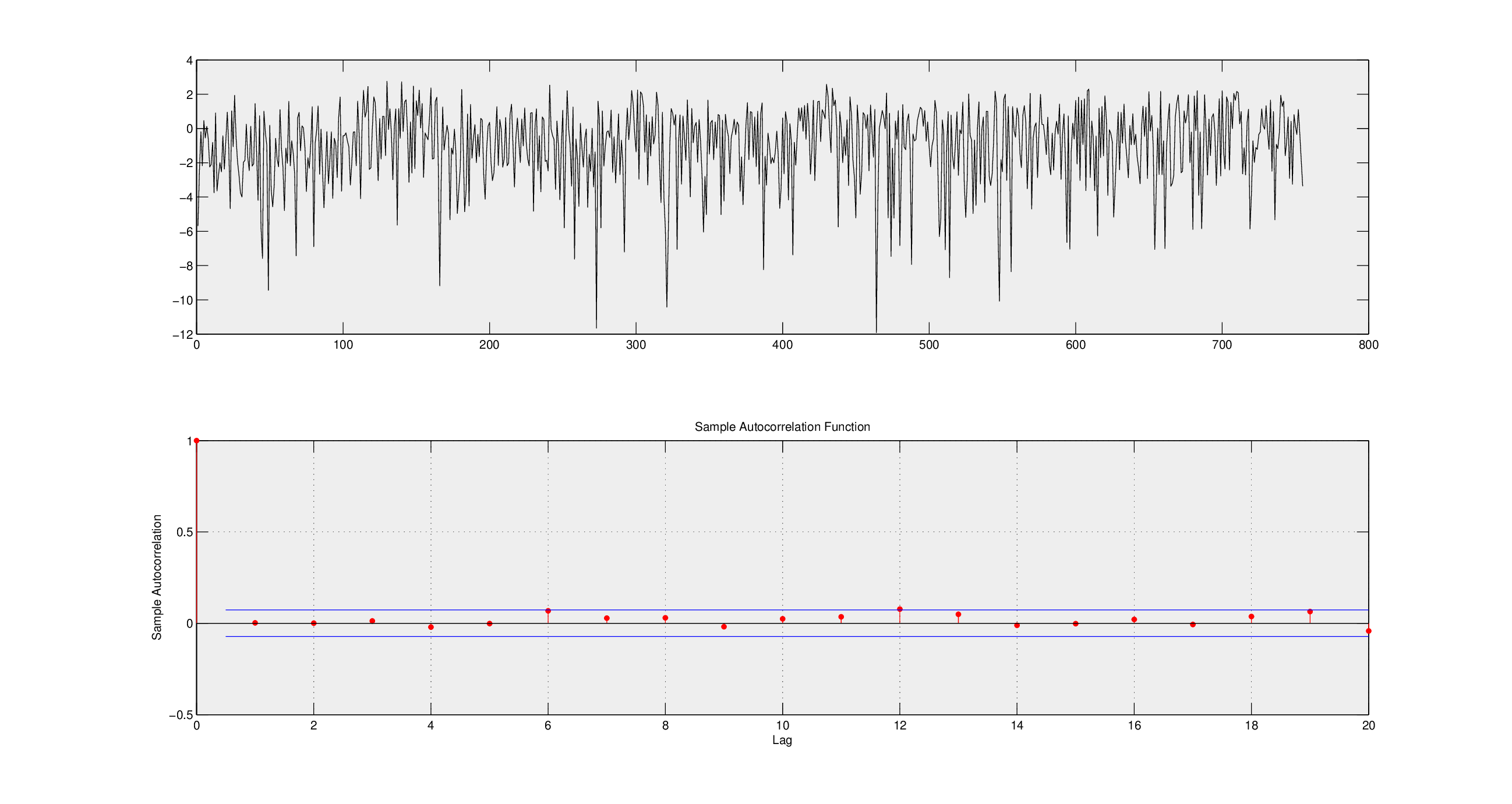}
\caption{\normalsize Top Left Panel: Graph of $Y_i$= FTSE. Top Right Panel:  Graph of $Y_i$= SP500. Bottom Left Panel: Autocorrelation of $Y_i$=FTSE.  Bottom Right Panel: Autocorrelation of $Y_i$=SP500.}
\label{RES}   
\end{figure}

\subsubsection{Parameter Estimates}

In the empirical analysis, we compare the QML, the Bootstrap filter, the APF and the KSAPF estimators. The last one is our contrast estimator. The variance of the measurement noise is $\sigma^2_{\varepsilon}=\frac{\pi^2}{2}$, that is $\beta$ is equal to $1$ (see Section \ref{implementation}). Table [\ref{resu}] summarises the parameter estimates and the computing time for the five methods. For initialization of the Bayesian procedure, we take the Uniform law for the parameters $p(\theta_1)=\mathcal{U}(0.4, 0.95)\times \mathcal{U}(0.1, 0.5)$ and the stationary law for the log-volatility process $X_1$, i.e, $f_{\theta_1}(X_1)=\mathcal{N}\left(0, \frac{\sigma^{2}_{1}}{1-\phi_1^2}\right)$.\\

The estimates of $\phi$ are in full accordance with results reported in previous studies of SV models. This parameter is in general close to 1 which implies persistent logarithmic volatility data. We compute the corresponding confidence intervals at level $5\%$ (see Table [\ref{icc}]). For the SP500 and the FTSE, note that the Bootstrap filter and the QML are not in the confidence interval for the two parameters $\phi$ and $\sigma^2$. These results are consistent with the simulations where we showed that both methods were biased for the SV model (see Section \ref{co}).  Note also that as expected the computing time for the QML is the shortest because it assumes Gaussianity which is probably not the case here. Except of QML, the contrast is the fastest method. The results are presented in Table [\ref{resu}] below.

\begin{table}[H]
\caption{\normalsize \label{resu} Parameter estimates: $n=1000$ and the number of particles $M=5000$ for the particle filters.}
\begin{center}
\begin{tabular}{lllllll}
\hline\noalign{\smallskip}
Index &  & FTSE &  &  & SP500 &   \\
\noalign{\smallskip}\hline\noalign{\smallskip}
& $\hat{\phi}_n$ & $\hat{\sigma}^{2}_n$ & CPU & $\hat{\phi}_n$ & $\hat{\sigma}^{2}_n$ & CPU \\
\noalign{\smallskip}\hline\noalign{\smallskip}
Contrast &  0.69 & 0.27 & 26 & 0.78 & 0.13 & 38  \\
\noalign{\smallskip}\hline
Bootstrap filter &  0.91 & 0.15 & 204 & 0.830 & 0.247 & 214  \\
\noalign{\smallskip}\hline
APF &  0.693 & 0.29  & 169 & 0.734 & 0.108 & 182  \\
\noalign{\smallskip}\hline
KSAPF &  0.697 & 0.29 & 152 & 0.80 & 0.12 & 175  \\
\noalign{\smallskip}\hline
QML & 0.649 & 0.08 & 0.07 &  0.895& 0.257 & 0.1 \\
\noalign{\smallskip}\hline
\end{tabular}
\end{center}
\end{table}

\begin{table}[H]
\caption{\normalsize  \label{icc} Confidence interval at level $5\%$.}
\begin{center}
\begin{tabular}{lllll}
\hline\noalign{\smallskip}
Index &   Confidence Interval   \\
\noalign{\smallskip}\hline\noalign{\smallskip}
&  $\phi$ &   $\sigma^{2}$ \\
\noalign{\smallskip}\hline\noalign{\smallskip}
FTSE & [ 0.6627  ; 0.7173]& [0.1771  ; 0.3629]  \\
\noalign{\smallskip}\hline
SP500 & [ 0.7086  ; 0.8514]& [ 0.0278 ; 0.2322]  \\
\noalign{\smallskip}\hline
\end{tabular}
\end{center}
\end{table}

\subsection{Summary and Conclusions}

In this paper we propose a new method to estimate an hidden stochastic model on the form (\ref{mod1}). This method is based on the deconvolution strategy and leads to a consistent and asymptotically normal estimator. We empirically study the performance of our estimator for the Gaussian AR(1) model and SV model and we are able to construct a confidence interval (see Figures [\ref{ic1}] and [\ref{ic2}]). As the boxplots [\ref{boxplot_compa1}] and [\ref{boxplot_compa2}] show, only the Contrast, the APF, and the KSAPF estimators are comparable. Indeed the QML and the Bootstrap Filter estimators are biased and their MSE are bad, and in particular, the QML method is the worst estimator (see Figure [\ref{compar2}]). One can see that the QML estimator proposed by Harvey et al. is not suitable for the SV model because the approximation of the log-chi-square density by the Gaussian density is not robust (see Figure [\ref{approxim}]). Furthermore, if we compare the MSE of the three Sequential Bayesian estimation, the KSAPF estimator is the best method. From a Bayesian point of view, it is known that the Bootstrap filter is less efficient than the APF and KSAPF filter since by using the density transition as the importance density, the propagation step of the particles will be made without taking care the observations (see \cite{DoFr01}).\\

Among the three estimators (Contrast, APF, and KSAPF) which give good results our estimator outperforms the others in a MSE aspect (see Figure [\ref{compar2}]). Moreover, as we already mentioned, in the combined state and parameters estimation the difficulties are the choice of $Q$, $h$ and the prior law since the results depend on these choices. In the numerical section, we have used the stationary law for the variable $X_1$  and this choice yields good results but we expect that the behavior of the Bayesian estimation will be worse for another prior.
The implementation of the contrast estimator is the easiest and it leads to confidence intervals with a larger variance than the SIEMLE but at a smaller computing cost, in particular for the AR(1) Gaussian model (see Table [\ref{time_gaussian}]). Furthermore, the contrast estimator does not require an arbitrary choice of parameter in practice. 

\newpage

\appendix

\small
\section{M-Estimator}\label{appen}

\begin{definition}{Geometrical ergodic process} \label{def_geo}\\

Denote by $Q^n(x,.)$ the transition kernel at step $n$ of a (discrete-time) stationary Markov chain $(X_n)_n$ which started at $x$ at time $0$. That is, $Q^n(x,F) = \mathbb{P}(X_n \in F | X_0 =x)$. Let $\pi$ denote the stationary law of $X_n$ and let $f$ be any measurable function. We call mixing coefficients $(\beta_n)_n$ the coefficients defined by, for each $n$:
\begin{equation*}
\beta_n = \int \left[ \sup_{||f||_{\infty}\leq 1}\left|Q^n(x,f) - \pi(f) \right| \right] \pi(dx),
\end{equation*}
where $\pi(f) = \int f(y) \pi(dy)$. We say that a process is geometrically ergodic if the decreasing of the sequence of the mixing coefficients $(\beta_n)_n$ is geometrical, that is:
$$\exists\ 0<\eta<1,\text{ such that } \beta_n \leq \eta^n.$$
\end{definition}

The following results are the main tools for the proof of Theorem \ref{MR}.\\

Consider the following quantities:

\begin{equation*}
\mathbf{P}_{n}h_{\theta}=\frac{1}{n}\sum_{i=1}^{n}h_{\theta}(Y_i);\quad \mathbf{P}_{n}S_{\theta}=\frac{1}{n}\sum_{i=1}^{n}\nabla_{\theta}h_{\theta}(Y_i) \text{ and }  \mathbf{P}_{n}H_{\theta}=\frac{1}{n}\sum_{i=1}^{n}\nabla^{2}_{\theta}h_{\theta}(Y_i)
\end{equation*}

where $h_{\theta}(y)$ is real function from $\Theta \times \mathcal{Y}$ with value in $\R$.\\

\begin{lemma}{Uniform Law of Large Numbers (\textbf{ULLN})(see \cite{newey} for the proof.)}\label{ULLN}\qquad\\

Let $(Y_i)$ be an ergodic stationary process and suppose that:
\begin{enumerate}
\item   $h_{\theta}(y)$ is continuous in $\theta$ for all $y$ and measurable in $y$ for all $\theta$ in the compact subset $\Theta$.
\item   There exists a function $s(y)$(called the dominating function) such that $\left|h_{\theta}(y)\right|\leq s(y)$ for all $\theta \in \Theta$ and $\E[s(Y_1)]<\infty$. Then:
\end{enumerate}

\begin{equation*}
\sup_{\theta \in \Theta}\left|\mathbf{P}_{n}h_{\theta}-\mathbf{P}h_{\theta}\right|\rightarrow 0 \qquad \text{ in probability as n } \rightarrow  \infty. 
\end{equation*}

Moreover, $\mathbf{P}h_{\theta}$ is a continuous function of $\theta$.\\
\end{lemma}

\begin{proposition}[Proposition 7.8 p. 472 in \cite{Fy00}. The proof is in \cite{At85} Theorem 4.1.5.]\label{fumio}

Suppose that:
\begin{enumerate}
\item $\theta_0$ is in the interior of $\Theta$.
\item $h_{\theta}(y)$ is twice continuously differentiable in $\theta$ for any $y$.
\item  The Hessian matrix of the application $\theta \mapsto \mathbf{P}h_{\theta}$ is non-singular.
\item $\sqrt{n}\mathbf{P}_{n}S_{\theta} \stackrel{}{\rightarrow} \mathcal{N}(0, \Omega(\theta_0))$ in law as n $\rightarrow \infty$, with $\Omega(\theta_0)$ a positive definite matrix.
\item Local dominance on the Hessian: for some neighbourhood $\mathcal{U}$ of $\theta_0$:
\begin{equation*}
\E\left[\sup_{\theta \in \mathcal{U} }\left\|\nabla_{\theta}^{2}h_{\theta}(Y_{1})\right\|\right]<\infty,
\end{equation*}
\end{enumerate}
so that, for any consistent estimator $\hat{\theta}$ of $\theta_0$ we have: $\mathbf{P}_{n}H_{\hat{\theta}} \rightarrow \E[\nabla^{2}_{\theta}h_{\theta}(Y_1)]$ in probability as n $\rightarrow \infty$.\\

Then, $\hat{\theta}$ is asymptotically normal with asymptotic covariance matrix given by:
\begin{equation*}
\Sigma(\theta_{0})=\E[\nabla^{2}_{\theta}h_{\theta}(Y_1)]^{-1} \Omega(\theta_{0})\E[\nabla^{2}_{\theta}h_{\theta}(Y_1)]^{-1}
\end{equation*}
where the differential $\nabla^{2}_{\theta}h_{\theta}(Y_1)$ is taken at point $\theta=\theta_0$.
\end{proposition}

\begin{proposition}[The proof is in \cite{Ga04}]\label{galin}\quad\\

Let $Y_i$ be an ergodic stationary Markov chain and let g: $\mathcal{Y}$ $\rightarrow$ $\mathbb{R}$ a borelian function. Suppose that $Y_i$ is geometrically ergodic and  $\E\left[|g(Y_1)|^{2+\delta}\right]<\infty$ for some $\delta >0$. Then, when $n \rightarrow \infty$,
\begin{equation*}
\sqrt{n}(\mathbf{P}_{n}g-\mathbf{P}g)\stackrel{}{\rightarrow}\mathcal{N}(0,\sigma^{2}_{g}) \text{ in law,}
\end{equation*}
where $\sigma^{2}_{g}:=Var\left[(g(Y_{1})\right]+2\sum_{j=1}^{\infty}Cov \left(g(Y_{1}), g(Y_{j})\right)<\infty$
\end{proposition}

\section{Proofs of Theorem \ref{MR}}\label{proof_result}

For the reader convenience we split the proof of Theorem \ref{MR} into three parts: in Subsection \ref{EoE}, we give the proof of the existence of our contrast estimator defined in (\ref{Procedure}). In Subsection \ref{CoE}, we prove the consistency, that is, the first part of Theorem \ref{MR}. Then, we prove the asymptotic normality of our estimator in  Subsection \ref{ANoE}, that is, the second part of Theorem \ref{MR}.  The Section \ref{PC} is devoted to  Corollary \ref{lele}. Finally, in  Section \ref{appendiceB} we prove that Theorem \ref{MR} applies for the AR(1) and SV models.

\subsection{ Proof of the existence and measurability of the M-Estimator}\label{EoE}

By assumption, the function $\theta \mapsto \left\|l_{\theta}\right\|_{2}^2$ is continuous. Moreover, $l^{*}_{\theta}$ and then $u^{*}_{l_{\theta}}(x)=\frac{1}{2\pi}\int e^{ixy}\frac{l^{*}_{\theta}(-y)}{f^{*}_{\varepsilon}(y)}dy$ are continuous w.r.t $\theta$. In particular, the function $m_{\theta}(\y_{i})=\left\|l_{\theta}\right\|_{2}^2-2y_{i+1}u^{*}_{l_{\theta}}(y_{i})$ is continuous w.r.t $\theta$. Hence, the function $\mathbf{P}_{n}m_{\theta}=\frac{1}{n}\sum_{i=1}^{n}m_{\theta}(\Y_i)$ is continuous w.r.t $\theta$ belonging to the compact subset $\Theta$. So, there exists $\tilde{\theta}$ belongs to $\Theta$ such that: 
 \begin{equation*}
\inf_{\theta \in \Theta}\mathbf{P}_{n}m_{\theta}=\mathbf{P}_{n}m_{\tilde{\theta}}.\qed
 \end{equation*}
 
\subsection{ Proof of the Consistency}\label{CoE}

By assumption $l_{\theta}$ is continuous w.r.t $\theta$ for any $x$ and measurable w.r.t $x$ for all $\theta$ which implies the continuity and the measurability of the function $\mathbf{P}_{n}m_{\theta}$ on the compact subset $\Theta$. Furthermore, the local dominance assumption \textbf{(C)} implies that $\E\left[\sup_{\theta \in \Theta }\left|m_{\theta}(\Y_i)\right|\right]$ is finite. Indeed,

\begin{eqnarray*}
\left|m_{\theta}(\y_{i})\right|&=&\left|\left\|l_{\theta}\right\|_{2}^2-2y_{i+1}u^{*}_{l_{\theta}}(y_{i})\right|\nonumber\\
&\leq& \left\|l_{\theta}\right\|_{2}^2+2\left|y_{i+1}u^{*}_{l_{\theta}}(y_{i})\right|.\label{dom}
\end{eqnarray*} 

As $\left\|l_{\theta}\right\|_{2}^2$ is continuous on the compact subset $\Theta$, $\sup_{\theta \in \Theta }\left\|l_{\theta}\right\|_{2}^2$ is finite. Therefore, $\E\left[\sup_{\theta \in \Theta }\left|m_{\theta}(\Y_i)\right|\right]$ is finite if $\E\left[\sup_{\theta \in \Theta }\left|Y_{i+1}u^{*}_{l_{\theta}}(Y_i)\right|\right]$ is finite. Lemma \emph{ULLN} \ref{ULLN} gives us the uniform convergence in probability of the contrast function: for any $\varepsilon>0$:

\begin{equation*}
\lim_{n\rightarrow +\infty} \mathbb{P}\left(\sup_{\theta \in \Theta}\left|\mathbf{P}_{n}m_{\theta}-\mathbf{P}m_{\theta}\right|\leq \varepsilon\right)=1.
\end{equation*}

Combining the uniform convergence with Theorem 2.1 p. 2121 chapter 36 in \cite{robert} yields the weak (convergence in probability) consistency of the estimator.\qed

\begin{remark}

In most applications, we do not know the bounds for the true parameter. So the compactness assumption is sometimes restrictive, one can replace the compactness assumption by: $\theta_0$ is an element of the interior of a convex parameter space $\Theta \subset \mathbb{R}^{r}$. Then, under our assumptions except the compactness, the estimator is also consistent. The proof is the same and the existence is proved by using convex optimization arguments. One can refer to \cite{Fy00} for this discussion.
\end{remark}

%%%%%%%%%%%%%%%%%%%%%%%%%%%%%%%%%%%%%%%%%%%%ASYMPTOTIC NORMALITY%%%%%%%%%%%%%%%%%%%%%%%%%%%%%%%%%%%%%%%%%%%%
\subsection{ Proof of the asymptotic normality}\label{ANoE}

The proof is based on the following Lemma: 
 
\begin{lemma}\label{F}

Suppose that the conditions of the consistency hold. Suppose further that:

\begin{enumerate}
\item $\Y_i$ geometrically ergodic.
\item (Moment condition): for some $\delta >0$ and for each $j\in \left\{1,\cdots ,r\right\}:$
$$\E\left[\left|\frac{\p m_{\theta}(\Y_{1})}{\p \theta_j}\right|^{2+\delta}\right]<\infty$$.
\item (Hessian Local condition): For some neighbourhood $\mathcal{U}$ of $\theta_0$ and for $j,k\in \left\{1,\cdots, r\right\}$ : 
$$\E\left[\sup_{\theta \in \mathcal{U} }\left|\frac{\p^2m_{\theta}(\Y_{1})}{\p \theta_j \p \theta_k}\right|\right]< \infty.$$
\end{enumerate}

Then, $\widehat{\theta}_{n}$ defined in Eq.(\ref{min}) is asymptotically normal with asymptotic covariance matrix given by:
\begin{equation*}
\Sigma(\theta_{0})=V_{\theta_0}^{-1} \Omega(\theta_{0})V_{\theta_0}^{-1}
\end{equation*}
where $V_{\theta_0}$ is the Hessian of the application $\mathbf{P}m_{\theta}$ given in Eq.(\ref{pmtheta}).
\end{lemma}

\begin{proof}
The proof follows from Proposition \ref{fumio} and Proposition \ref{galin} and by using the fact that by assumption we have $\E[\nabla_{\theta}^{2}m_{\theta}(\Y_1)]=\nabla_{\theta}^{2}\E[m_{\theta}(\Y_1)]$. 
\end{proof}

It just remains to check that the conditions (2) and (3) of Lemma \ref{F} hold under our assumptions \textbf{(T)} .\\

\emph{Moment condition:}
As the function $l_{\theta}$ is twice continuously differentiable w.r.t $\theta$, for all $\y_{i}$ $\in$ $\R^2$, the application $m_{\theta}(\y_{i}):\ \theta \in \Theta \mapsto m_{\theta}(\y_{i})=||l_{\theta}||_{2}^2 - 2y_{i+1}u^*_{l_{\theta}}(y_{i})$ is twice continuously differentiable for all $\theta$ $\in$ $\Theta$ and its first derivatives are given by:

\begin{equation*}
\nabla_{\theta}m_{\theta}(\y_{i})= \nabla_{\theta}||l_{\theta}||_{2}^2 - 2y_{i+1}\nabla_{\theta}u^*_{l_{\theta}}(y_{i}).\\
\end{equation*}

 By assumption, for each $j\in \left\{1,\cdots,r\right\}$, $\frac{\p l_{\theta}}{\p \theta_j} \in \mathbb{L}_{1}(\R)$, therefore one can apply the Lebesgue Derivation Theorem and Fubini's Theorem to obtain :

\begin{equation}
\nabla_{\theta}m_{\theta}(\y_{i})
=\left[\nabla_{\theta}||l_{\theta}||_{2}^2 - 2y_{i+1}u^*_{\nabla_{\theta}l_{\theta}}(y_{i}) \right]\label{expres}.
\end{equation}

Then, for some $\delta >0$:

\begin{eqnarray}
\left| \nabla_{\theta}m_{\theta}(\y_{i})\right|^{2+\delta}&=&\left|\nabla_{\theta}||l_{\theta}||_{2}^2 - 2y_{i+1}u^*_{\nabla_{\theta}l_{\theta}}(y_{i}) \right|^{2+\delta}\nonumber\\
&\leq&  C_{1}\left|\nabla_{\theta}||l_{\theta}||_{2}^2\right|^{2+\delta}+C_{2}\left|y_{i+1}u^*_{\nabla_{\theta}l_{\theta}}(y_{i}) \right|^{2+\delta},\label{ca}
\end{eqnarray}
where $C_{1}$ and $C_{2}$ are two positive constants. By assumption, the function $||l_{\theta}||_{2}^2$ is twice continuously differentiable w.r.t $\theta$. Hence, $\nabla_{\theta}||l_{\theta}||_{2}^2$ is continuous on the compact subset $\Theta$ and the first term of equation (\ref{ca}) is finite. The second term is finite by the moment assumption \textbf{(T)}.\\

\emph{Hessian Local dominance:} For $j,k \in\left\{1,\cdots,r\right\}$, $\frac{\p^2 l_{\theta}}{\p\theta_j \p\theta_k} \in \mathbb{L}_{1}(\R)$, the Lebesgue Derivation Theorem gives:

\begin{equation*}
\nabla^{2}_{\theta}m_{\theta}(\y_{i})=\nabla^{2}_{\theta}||l_{\theta}||_{2}^2 - 2y_{i+1}u^*_{\nabla^{2}_{\theta}l_{\theta}}(y_{i}),
\end{equation*}
and, for some neighbourhood $\mathcal{U}$ of $\theta_0$:

\begin{equation*}
\E\left[\sup_{\theta \in \mathcal{U} }\left\|\nabla_{\theta}^{2}m_{\theta}(\Y_i)\right\|\right]\leq \sup_{\theta \in \mathcal{U}}\left\|\nabla^{2}_{\theta}||l_{\theta}||_{2}^2\right\|+2\E\left[\sup_{\theta \in \mathcal{U}}\left\|Y_{i+1}u^{*}_{\nabla_{\theta}^{2}l_{\theta}}(Y_i)\right\|\right].
\end{equation*} The first term of the above equation is finite by continuity and by compactness argument. And, the second term is finite by the Hessian local dominance assumption \textbf{(T)}.\qed

\subsection{Proof of Corollary \ref{lele}}\label{PC}
By replacing $\nabla_{\theta}m_{\theta}(\Y_{1})$ by its expression (\ref{expres}), we have:

\begin{eqnarray*}
\Omega_{0}(\theta)&=&\Var\left[\nabla_{\theta}||l_{\theta}||_{2}^2 - 2Y_{2}u^*_{\nabla_{\theta}l_{\theta}}(Y_{1}) \right]\\
&=&4 \Var\left[Y_{2}u^*_{\nabla_{\theta}l_{\theta}}(Y_{1}) \right]\\
&=&4\left[\E \left[Y_{2}^2 \left(u^*_{\nabla_{\theta}l_{\theta}}(Y_{1}) \right)\left(u^*_{\nabla_{\theta}l_{\theta}}(Y_{1}) \right)'\right]-\E \left[Y_{2} u^*_{\nabla_{\theta}l_{\theta}}(Y_{1})\right]\E \left[Y_{2} u^*_{\nabla_{\theta}l_{\theta}}(Y_{1})\right]'\right].
\end{eqnarray*}
Furthermore, by Eq.(\ref{mod1}) and by independence of the centered noise $(\varepsilon_2)$ and $(\eta_2)$, we have:

\begin{equation*}
\E\left[ Y_{2}  u^*_{\nabla_{\theta}l_{\theta}}(Y_{1})\right] = \E\left[ b_{\phi_{0}}(X_1)  u^*_{\nabla_{\theta}l_{\theta}}(Y_{1})\right].
\end{equation*}

Using Fubini's Theorem and Eq.(\ref{mod1}) we obtain:

\begin{eqnarray}\label{fub}
\E\left[ b_{\phi_{0}}(X_1)  u^*_{\nabla_{\theta}l_{\theta}}(Y_{1})\right] &=& \E\left[b_{\phi_{0}}(X_1) \int e^{iY_{1}z} u_{\nabla_{\theta}l_{\theta}}(z) dz  \right]\nonumber\\
&=&\E\left[b_{\phi_{0}}(X_1) \int \frac{1}{2\pi}\frac{1}{f_{\varepsilon}^*(z)}e^{iY_{1}z} (\nabla_{\theta}l_{\theta})^*(-z)dz \right]\nonumber\\
&=&\frac{1}{2\pi} \int \E\left[b_{\phi_{0}}(X_1)e^{i(X_1+\varepsilon_1)z} \right] \frac{1}{f_{\varepsilon}^*(z)} (\nabla_{\theta}l_{\theta})^*(-z) dz \nonumber\\
&=&\frac{1}{2\pi} \int \frac{\E\left[ e^{i\varepsilon_1z}\right]}{f_{\varepsilon}^*(z)} \E\left[b_{\phi_{0}}(X_1)e^{iX_1z}\right] (\nabla_{\theta}l_{\theta})^*(-z) dz\nonumber\\
&=&\frac{1}{2\pi} \E\left[ b_{\phi_{0}}(X_1) \int e^{iX_1z} (\nabla_{\theta}l_{\theta})^*(-z) dz \right]\nonumber\\
&=&\frac{1}{2\pi} \E\left[ b_{\phi_{0}}(X_1) \left((\nabla_{\theta}l_{\theta})^*(-X_1)\right)^*\right]\nonumber\\
&=& \E\left[ b_{\phi_{0}}(X_1) \nabla_{\theta}l_{\theta}(X_1)\right].
\end{eqnarray}

Hence,
\begin{equation*}
\Omega_{0}(\theta)=4\left(P_{2}-P_{1}\right),
\end{equation*} 
where
\begin{eqnarray*}
&&P_{1}=\E\left[ b_{\phi_{0}}(X_1) \nabla_{\theta}l_{\theta}(X_1)\right]\E\left[ b_{\phi_{0}}(X_1) \nabla_{\theta}l_{\theta}(X_1)\right]',\\
&&P_2=\E \left[Y_{2}^2 \left(u^*_{\nabla_{\theta}l_{\theta}}(Y_{1}) \right)\left(u^*_{\nabla_{\theta}l_{\theta}}(Y_{1}) \right)'\right].
\end{eqnarray*}

\emph{Calculus of the covariance matrix of Corollary (\ref{lele}):} By replacing $(\nabla_{\theta}m_{\theta}(Y_{1}))$ by its expression (\ref{expres}) we have:

\begin{eqnarray*}
\Omega_{j-1}(\theta)
&=&\C\left(\nabla_{\theta}||l_{\theta}||_{2}^2 - 2Y_{2}u^*_{\nabla_{\theta}l_{\theta}}(Y_{1}) , \nabla_{\theta}||l_{\theta}||_{2}^2 - 2Y_{j+1}u^*_{\nabla_{\theta}l_{\theta}}(Y_{j}) \right),\\
&=&4\C\left(Y_{2}u^*_{\nabla_{\theta}l_{\theta}}(Y_{1}) , Y_{j+1}u^*_{\nabla_{\theta}l_{\theta}}(Y_{j})\right),\\
&=&4\left[\E\left(Y_{2}u^*_{\nabla_{\theta}l_{\theta}}(Y_{1})Y_{j+1}u^*_{\nabla_{\theta}l_{\theta}}(Y_{j})\right)-\E\left(Y_{2}u^*_{\nabla_{\theta}l_{\theta}}(Y_{1})\right)\E\left(Y_{j+1}u^*_{\nabla_{\theta}l_{\theta}}(Y_{j})\right)'\right].
\end{eqnarray*}

By using Eq.(\ref{fub}) and the stationary property of the $Y_i$, one can replace the second term of the above equation by: 

\begin{equation*}
\E\left[b_{\phi_{0}}(X_1)\nabla_{\theta}l_{\theta}(X_1)\right]\E\left[b_{\phi_{0}}(X_1)\nabla_{\theta}l_{\theta}(X_1)\right]'.
\end{equation*}

Furthermore, by using Eq.(\ref{mod1}) we obtain:

\begin{eqnarray}
\E\left[Y_{2}Y_{j+1}u^{*}_{\nabla_{\theta}l_{\theta}}(Y_{1}) u^{*}_{\nabla_{\theta}l_{\theta}}(Y_{j})\right]
&=&\E\left[b_{\phi_{0}}(X_1)b_{\phi_{0}}(X_j)u^{*}_{\nabla_{\theta}l_{\theta}} (Y_{1})u^{*}_{\nabla_{\theta}l_{\theta}}(Y_{j})\right]\nonumber\\
&+&\E\left[b_{\phi_{0}}(X_1)\left(\eta_{j+1}+\varepsilon_{j+1}\right)u^{*}_{\nabla_{\theta}l_{\theta}}(Y_{1})u^{*}_{\nabla_{\theta}l_{\theta}}(Y_{j})\right]\label{1}\\
&+&\E\left[b_{\phi_{0}}(X_j)\left(\eta_2+\varepsilon_2\right)u^{*}_{\nabla_{\theta}l_{\theta}}(Y_{1})u^{*}_{\nabla_{\theta}l_{\theta}}(Y_{j})\right]\label{2}\\
&+&\E\left[\left(\eta_2+\varepsilon_2\right)\left(\eta_{j+1}+\varepsilon_{j+1}\right)u^{*}_{\nabla_{\theta}l_{\theta}}(Y_{1})u^{*}_{\nabla_{\theta}l_{\theta}}(Y_{j})\right]\label{3}.
\end{eqnarray}

$\newline$
By independence of the centered noise, the term (\ref{1}), (\ref{2}) and (\ref{3}) are equal to zero. Now, if we use Fubini's Theorem we have: 

\begin{equation}
\E\left[b_{\phi_{0}}(X_1)b_{\phi_{0}}(X_j)u^{*}_{\nabla_{\theta}l_{\theta}}(Y_{1})u^{*}_{\nabla_{\theta}l_{\theta}}(Y_{j})\right]=\E\left[b_{\phi_{0}}(X_1) b_{\phi_{0}}(X_j) \nabla_{\theta}l_{\theta}(X_1) \nabla_{\theta}l_{\theta}(X_j)\right].
\end{equation}

Hence, the covariance matrix is given by:
 
\begin{eqnarray*}
\Omega_{j-1}(\theta)&=&4\left(\E\left[b_{\phi_{0}}(X_1)b_{\phi_{0}}(X_j)\left(\nabla_{\theta}l_{\theta}(X_1)\right)\left(\nabla_{\theta}l_{\theta}(X_j)\right)'\right]-\E\left[b_{\phi_{0}}(X_1)\left(\nabla_{\theta}l_{\theta}(X_1)\right)\right]\E\left[b_{\phi_{0}}(X_1)\left(\nabla_{\theta}l_{\theta}(X_1)\right)\right]'\right)\\
&=&4\left(\tilde{C}_{j-1}-\E\left[b_{\phi_{0}}(X_1)\left(\nabla_{\theta}l_{\theta}(X_1)\right)\right]\E\left[b_{\phi_{0}}(X_1)\left(\nabla_{\theta}l_{\theta}(X_1)\right)\right]'\right)\\
&=&4\left(\tilde{C}_{j-1}-P_{1}\right).
\end{eqnarray*}
Finally, we obtain: $\Omega(\theta)=\Omega_{0}(\theta)+2\sum_{j >1}^{\infty}\Omega_{j-1}(\theta)$ with $\Omega_{0}(\theta)=4\left(P_2-P_{1}\right)$ and $\Omega_{j-1}(\theta)=4\left(\tilde{C}_{j-1}-P_{1}\right)$.

$\newline$
\emph{Expression of the Hessian matrix $V_{\theta}$ :} We have:

\begin{equation}
\mathbf{P}m_{\theta} = ||l_{\theta}||_{2}^2 - 2\left\langle l_{\theta}, l_{\theta_0}\right\rangle. 
\end{equation}

For all $\theta$ in $\Theta$, the application $\theta \mapsto \mathbf{P}m_{\theta}$ is twice differentiable w.r.t $\theta$ on the compact subset $\Theta$. And for $j\in \left\{1,\cdots,r\right\}$:

\begin{eqnarray*}
\frac{\p \mathbf{P}m}{\p \theta_j}(\theta)&=& 2 \left\langle \frac{\p l_{\theta}}{\p \theta_j}, l_{\theta}\right\rangle-2 \left\langle  \frac{\p l_{\theta}}{\p \theta_j}, l_{\theta_0}\right\rangle\\
&=&2 \left\langle \frac{\p l_{\theta}}{\p \theta_j}, l_{\theta}-l_{\theta_0}\right\rangle,\\
&=&0 \text{ at the point } \theta_{0},
\end{eqnarray*}
and for $j,k \in \left\{1,\cdots,r\right\}$:

\begin{eqnarray*}
\frac{\p^2 \mathbf{P}m}{\p\theta_j \p\theta_k}(\theta) &=& 2 \left(\left\langle \frac{\p^2 l_{\theta}}{\p \theta_j \theta_k}, l_{\theta}- l_{\theta_0}\right\rangle+\left\langle \frac{\p l_{\theta}}{\p  \theta_k}, \frac{\p l_{\theta}}{\p \theta_j}\right\rangle\right)_{j,k}\\
&=& 2 \left(\left\langle \frac{\p l_{\theta}}{\p  \theta_k}, \frac{\p l_{\theta}}{\p \theta_j}\right\rangle\right)_{j,k}\text{ at the point } \theta_{0}.
\end{eqnarray*}

\section{Proof of the Applications }\label{appendiceB}

\subsection{The Gaussian AR(1) model with measurement noise}\label{AppGauss}
\subsubsection{Contrast Function}
We have:
\begin{equation*}
l_{\theta}(x)=\frac{1}{\sqrt{2\pi\gamma^{2}}} \phi x \exp\left(-\frac{1}{2\gamma^2}x^2\right).
\end{equation*}
So that:

\begin{equation*}
||l_{\theta}||_{2}^{2}=\int |l_{\theta}(x)|^{2}dx =\frac{\phi^{2}\gamma}{4\sqrt{\pi}},
\end{equation*}
and the Fourier Transform of $l_{\theta}$ is given by:

\begin{eqnarray*}
l^{*}_{\theta}(y) &=&\int e^{iyx}l_{\theta}(x)dx = \int e^{iyx} \frac{1}{\sqrt{2\pi\gamma^{2}}} \phi x \exp\left(-\frac{1}{2\gamma^2}x^2\right)dx\nonumber\\
&=& -i\phi \E\left[iGe^{iyG}\right] = -i\phi \frac{\p}{\p y}\E\left[ e^{iyG}\right]\quad \text{ where } G\sim \mathcal{N}(0,\gamma^2)\nonumber\\
&=& -i\phi \frac{\p}{\p y} \left[ e^{-\frac{y^2}{2}\gamma^2}\right]\nonumber\\
&=& i\phi y \gamma^2e^{-\frac{y^2}{2}\gamma^2}.
\end{eqnarray*}

As $\varepsilon_i$ is a centered Gaussian noise with variance $\sigma_{\varepsilon}^2$, we have:

\begin{equation*}
f_{\varepsilon}(x)=\frac{1}{\sqrt{2\pi \sigma_{\varepsilon}^2}}  \exp\left(-\frac{1}{2\sigma_{\varepsilon}^2}x^2\right) \text{ and } f^*_{\varepsilon}(x)=  \exp\left(-\frac{1}{2}x^2 \sigma_{\varepsilon}^2\right). 
\end{equation*}

Define:
\begin{equation*}
\displaystyle u_{l_{\theta}}(y) =\frac{1}{2\pi} \frac{l^*_{\theta}(-y)}{f^*_{\varepsilon}(y)}.
\end{equation*}
Then:

\begin{eqnarray*}
u^{*}_{l_{\theta}}(y) &=&\frac{1}{2\pi}\int  \frac{l^*_{\theta}(-x)}{f^*_{\varepsilon}(x)}  e^{iyx}dx = \frac{-i}{2\pi} \phi \gamma^2 \int  x e^{iyx} \exp\left(\frac{x^2}{2}\sigma_{\varepsilon}^2\right) \exp\left(\frac{-x^2}{2}\gamma^2\right)dx\\
&=& \frac{-i}{2\pi}\phi \gamma^2 \frac{1}{(\gamma^2-\sigma_{\varepsilon}^2)^{1/2}} \int xe^{iyx}  (\gamma^2-\sigma_{\varepsilon}^2)^{1/2}\exp\left(-\frac{1}{2}x^2(\gamma^2-\sigma_{\varepsilon}^2)\right)dx\\
&=& -\frac{1}{\sqrt{2\pi}}\phi \gamma^2 \frac{1}{(\gamma^2-\sigma_{\varepsilon}^2)^{1/2}} \E\left[iGe^{iyG}\right] = -\frac{1}{\sqrt{2\pi}}\phi \gamma^2 \frac{1}{(\gamma^2-\sigma_{\varepsilon}^2)^{1/2}} \frac{\p}{\p y}\E\left[ e^{iyG}\right]\\
&=&  -\frac{1}{\sqrt{2\pi}}\phi \gamma^2 \frac{1}{(\gamma^2-\sigma_{\varepsilon}^2)^{1/2}} \frac{\p}{\p y} \left[ e^{-\frac{y^2}{2(\gamma^2-\sigma_{\varepsilon}^2)}}\right]\\
&=&\frac{1}{\sqrt{2\pi}} \phi \gamma^2 \frac{1}{(\gamma^2-\sigma_{\varepsilon}^2)^{3/2}}y e^{-\frac{y^2}{2(\gamma^2-\sigma_{\varepsilon}^2)}},
\end{eqnarray*}
where $G\sim \mathcal{N}\left(0,\frac{1}{(\gamma^2-\sigma_{\varepsilon}^2)} \right)$. We deduce that the function $m_{\theta}(\y_{i})$ is given by:

\begin{eqnarray*}
m_{\theta}(\y_{i})&=&||l_{\theta}||_{2}^{2}-2y_{i+1}u^*_{l_{\theta}}(y_{i})\label{m}\\
&=&\frac{\phi^{2}\gamma}{4\sqrt{\pi}} - 2y_{i} y_{i+1} \frac{1}{\sqrt{2\pi}}\phi \gamma^2 \frac{1}{(\gamma^2-\sigma_{\varepsilon}^2)^{3/2}}\exp\left(-\frac{y_{i}^2}{2(\gamma^2-\sigma_{\varepsilon}^2)}\right).\nonumber
\end{eqnarray*}

Then, the contrast estimator defined in (\ref{Procedure}) is given by: 

\begin{eqnarray*}
\widehat{\theta}_n&=&\arg \min_{\theta \in  \Theta}\mathbf{P}_{n}m_{\theta}\nonumber\\
&=& \arg \min_{\theta \in \Theta}\left\{ \frac{\phi^{2}\gamma}{4\sqrt{\pi}}-\sqrt{\frac{2}{\pi}}\frac{\phi \gamma^{2}}{n(\gamma^{2}-\sigma_{\varepsilon}^2)^{3/2}}\sum_{j=1}^{n} Y_{j+1} Y_{j}\exp\left(-\frac{1}{2}   \frac{Y^{2}_{j}}{(\gamma^{2}-\sigma^{2}_{\varepsilon})}\right)\right\}.\qed
\end{eqnarray*}

\subsubsection{Checking assumptions of Theorem \ref{MR}}\label{CTGAR}

\emph{Mixing properties.} If $|\phi|<1$, the process $\Y_i$ is geometrically ergodic. For further details, we refer to \cite{Do94}.\\

\emph{Regularity conditions:}
It remains to prove that the assumptions of Theorem \ref{MR} hold. It is easy to see that the only difficulty is to check the moment condition and the local dominance \textbf{(C)-(T)} and the uniqueness assumption \textbf{(CT)}. The others assumptions are easily to verify since the function $l_{\theta}(x)$ is regular in $\theta$ belonging to $\Theta$.\\

\textbf{(CT)}: The limit contrast function  $\displaystyle \mathbf{P}m_{\theta} : \theta \in \Theta \mapsto \mathbf{P}m_{\theta}$ given by:

\begin{eqnarray*}
\theta \mapsto \mathbf{P}m_{\theta}&=&
||l_{\theta}||_{2}^{2}-2\left\langle l_{\theta}, l_{\theta_{0}}\right\rangle \\
&=&\frac{\phi^{2}\gamma}{4\sqrt{\pi}}-\sqrt{\frac{2}{\pi}}\frac{\phi \phi_{0}\gamma^{2}\gamma_{0}^{2}}{(\gamma^{2}+\gamma_{0}^2)^{\frac{3}{2}}},
\end{eqnarray*}
is differentiable for all  $\theta$ in $\Theta$ and $\nabla_{\theta}\mathbf{P}m_{\theta}=0_{\R^2}$ if and only if $\theta$ is equal to $\theta_0$ . More precisely its first derivatives are given by:

\begin{eqnarray*}
&&\frac{\partial \mathbf{P}m_{\theta}}{\partial \phi}=\frac{1}{4\sqrt{\pi}}\frac{\phi\gamma(2-\phi^2)}{(1-\phi^2)}-\sqrt{\frac{2}{\pi}}\phi_0\gamma_0^{2}(\gamma^2+\gamma_0^2)^{-3/2}\left(\frac{\gamma^2+\gamma^2\phi^2}{(1-\phi^2)}-\frac{3\phi^2\gamma^4}{(1-\phi^2)(\gamma^2+\gamma_0^2)}\right),\\
&&\frac{\partial \mathbf{P}m_{\theta}}{\partial \sigma^2}= \frac{\phi^2}{8\sqrt{\pi}\sigma(1-\phi^2)^{1/2}}-\sqrt{\frac{2}{\pi}}\frac{\phi_0\gamma_0^{2}}{(1-\phi^2)(\gamma^2+\gamma_0^2)^{3/2}}\left(\phi-\frac{3\phi\gamma^2}{(\gamma^2+\gamma_0^2)}\right),
\end{eqnarray*}
and

\begin{eqnarray*}
\nabla_{\theta}\mathbf{P}m_{\theta}=0_{\R^2} \Leftrightarrow \theta=\theta_0
\end{eqnarray*}
The partial derivatives of $l_{\theta}$ w.r.t $\theta$ are given by:

\begin{eqnarray*}
&&\frac{\p l_{\theta}}{\p \phi}(x)= \left(\left(\frac{-\phi^2}{1-\phi^2}+1\right)x + \frac{\phi^2}{(1-\phi^2)\gamma^2}x^3\right)\frac{1}{\sqrt{2\pi \gamma^2}}e^{-\frac{x^2}{2\gamma^2}},\\
&&\frac{\p l_{\theta}}{\p \sigma^2}(x)=\left(-\frac{\phi}{2(1-\phi^2)\gamma^2}x +\frac{\phi}{2(1-\phi^2)\gamma^4}x^3 \right)\frac{1}{\sqrt{2\pi \gamma^2}}e^{-\frac{x^2}{2\gamma^2}}.
\end{eqnarray*}

For the reader convenience let us introduce the following notations:
\begin{eqnarray}\label{a1}
&&a_1= \frac{-\phi^2}{(1-\phi^2)}+1 = \frac{1-2\phi^2}{(1-\phi^2)} \text{ and } a_2= \frac{\phi^2}{(1-\phi^2)\gamma^2},\\
&&b_1= \frac{-\phi}{2(1-\phi^2)\gamma^2}\text{ and } b_2=  \frac{\phi}{2(1-\phi^2)\gamma^4}.
\end{eqnarray}

We rewrite:

\begin{eqnarray*}
\nabla_{\theta}l_{\theta}(x) &=& \left(\frac{\p l_{\theta}}{\p \phi}(x), \frac{\p l_{\theta}}{\p \sigma^2}(x)\right)^{'}\\
&=&\left((a_1x+a_2x^3)\times g_{0,\gamma^2}(x), (b_1x+b_2x^3)\times g_{0,\gamma^2}(x)\right)^{'},
\end{eqnarray*}
where the function $g_{0,\gamma^2}$  defines the normal probability density of a centered random variable with variance $\gamma^2$. Now, we can use Corollary \ref{lele} to compute the Hessian matrix $V_{\theta_0}$:   \\

\begin{equation}
V_{\theta_0}=
2 \begin{pmatrix} 
\left\|\frac{\p l_{\theta}}{\p \phi}\right\|_{2}^{2} & \left\langle \frac{\p l_{\theta}}{\p \phi}, \frac{\p l_{\theta}}{\p \sigma^{2}}\right\rangle\\
\left\langle \frac{\p l_{\theta}}{\p \sigma^{2}},  \frac{\p l_{\theta}}{\p \phi}\right\rangle & \left\|\frac{\p l_{\theta}}{\p \sigma^{2}}\right\|_{2}^{2} 
\end{pmatrix}
\end{equation}

\begin{equation*}
=\frac{1}{\gamma_0 \sqrt{\pi}}\begin{pmatrix}
a_{1}^{2}\E[X^{2}]+2a_{1}a_{2}\E[X^{4}]+a_{2}^{2}\E[X^{6}] & a_{1}b_{1}\E[X^{2}]+a_{1}b_{2}\E[X^{4}]+a_{2}b_{1}\E[X^{4}]+a_{2}b_{2}\E[X^{6}] \\
a_{1}b_{1}\E[X^{2}]+a_{1}b_{2}\E[X^{4}]+a_{2}b_{1}\E[X^{4}]+a_{2}b_{2}\E[X^{6}]& b_{1}^{2}\E[X^{2}]+2b_{1}b_{2}\E[X^{4}]+b_{2}^{2}\E[X^{6}]
\end{pmatrix},
\end{equation*}
with $X \sim \mathcal{N}\left(0,\frac{\gamma_{0}^2}{2}\right)$. By replacing the terms $a_{1}, a_{2}, b_{1}$ and $b_{2}$ at the point $\theta_0$ we obtain:
\begin{equation}\label{hes}
V_{\theta_{0}}
=\frac{1}{8\sqrt{\pi}(1-\phi_0^2)^2}
\begin{pmatrix}
\gamma_{0}(7\phi_0^4-4\phi_0^2+4) & \frac{-5\phi^{5}_0+3\phi^{3}_0+2\phi_{0}}{2\gamma_{0}(1-\phi^{2}_0)}\\
\frac{-5\phi^{5}_0+3\phi^{3}_0+2\phi_0}{2\gamma_{0}(1-\phi^{2}_0)} & \frac{7\phi_0^2}{4\gamma_{0}^3}
\end{pmatrix},
\end{equation}
which has a positive determinant equal to $0.0956$ at the true value $\theta_{0}=(0.7,0.3)$. Hence, $V_{\theta_{0}}$ is non-singular. Furthermore, the strict convexity of the function $\mathbf{P}m_{\theta}$ gives that $\theta_0$ is a minimum.\\

\textbf{(C):} (Local dominance): We have:
\begin{eqnarray*}
\E\left[\sup_{\theta \in \Theta}\left|Y_{2}u^{*}_{l_{\theta}}(Y_{1})\right|\right]&=&\frac{1}{\sqrt{2\pi}}\E\left[\sup_{\theta \in \Theta}\left|  \frac{\phi \gamma^2}{(\gamma^2-\sigma_{\varepsilon}^2)^{(3/2)}}Y_{2}Y_{1}\exp\left(-\frac{Y_{1}^2}{2(\gamma^2-\sigma_{\varepsilon}^2)}\right)\right|\right].
\end{eqnarray*}

The multivariate normal density of the pair $\Y_{1}=(Y_{1},Y_{2})$ denoted $g_{(0, \mathcal{J}_{\theta_0})}$ is given by: 

\begin{equation*}
\frac{1}{2\pi} \det{(\mathcal{J}_{\theta_0})}^{-1/2} \exp\left(-\frac{1}{2}\y_{1}^{'} \mathcal{J}_{\theta_0}^{-1}\y_{1}\right),
\end{equation*}
with: 

\begin{equation*}
\mathcal{J}_{\theta_0}=\begin{pmatrix}
\sigma_{\varepsilon}^2+\gamma^2_0 &  \phi_0 \gamma^2_0\\
 \phi_0 \gamma^2_0 & \sigma_{\varepsilon}^2+\gamma^2_0 
\end{pmatrix}
\text{ and }
\mathcal{J}_{\theta_0}^{-1}=\frac{1}{(\sigma_{\varepsilon}^2+\gamma^2_0 )^{2}-\gamma^4_0\phi^{2}_0}\begin{pmatrix}
\sigma_{\varepsilon}^2+\gamma^2_0 &  -\phi_0 \gamma^2_0\\
 -\phi_0\gamma^2_0 & \sigma_{\varepsilon}^2+\gamma^2_0
\end{pmatrix}.
\end{equation*}

By definition of the parameter space $\Theta$ and as all moments of the pair $\Y_{1}$ exist, the quantity $\E\left[\sup_{\theta \in \Theta}\left|Y_{2}u^{*}_{l_{\theta}}(Y_{1})\right|\right]$ is finite.\\

Moment condition \textbf{(T)}: We recall that:

\begin{eqnarray*}
\nabla_{\theta}l_{\theta}(x) &=& \left(\frac{\p l_{\theta}}{\p \phi}(x), \frac{\p l_{\theta}}{\p \sigma^2}(x)\right)^{'}\\
&=&\left((a_1x+a_2x^3)\times g_{0,\gamma^2}(x), (b_1x+b_2x^3)\times g_{0,\gamma^2}(x)\right)^{'}.
\end{eqnarray*}

The Fourier transforms of the first derivatives are:

\begin{eqnarray*}
\left(\frac{\partial l_{\theta}}{\partial \phi}(x)\right)^{*}&&=\int \exp\left(i xy\right)\left(a_1y+a_2y^3\right)\times g_{0,\gamma^2}(y) dy\\
&&=-ia_1\E\left[iG\exp\left(i xG\right)\right]+ia_2\E\left[-iG^{3}\exp\left(i xG\right)\right] \text{ where G} \sim \mathcal{N}(0,\gamma^{2}) \\
&&=-ia_{1}\frac{\p}{\p x}\E\left[\exp\left(i xG\right)\right]+ia_{2}\frac{\p^{3}}{\p x^{3}}\E\left[\exp\left(i xG\right)\right]\\
&&=-ia_{1}\frac{\p}{\p x}\exp\left(-\frac{x^{2}}{2}\gamma^{2}\right)+ia_{2}\frac{\p^{3}}{\p x^{3}}\exp\left(-\frac{x^{2}}{2}\gamma^{2}\right)\\
&&=(ia_{1}\gamma^{2}x+3ia_{2}\gamma^{4}x-ia_{2}\gamma^{6}x^{3})\exp\left(-\frac{x^{2}}{2}\gamma^{2}\right),
\end{eqnarray*}
and
\begin{equation*}
\left(\frac{\partial l_{\theta}}{\partial \sigma^{2}}(x)\right)^{*}=(ib_{1}\gamma^{2}x+3ib_{2}\gamma^{4}x-ib_{2}\gamma^{6}x^{3})\exp\left(-\frac{x^{2}}{2}\gamma^{2}\right).
\end{equation*}
We can compute the function $u_{\nabla_{\theta}l_{\theta}}(x)$:

\begin{eqnarray*}
u_{\frac{\partial l_{\theta}}{\partial \phi}}(x)&=&\frac{1}{2\pi}\frac{\left(\frac{\partial l_{\theta}}{\partial \phi}(-x)\right)^{*}}{f^{*}_{\varepsilon}(x)}\\
&=&\frac{1}{\sqrt{2\pi}}(\gamma^{2}-\sigma^{2}_{\varepsilon})^{1/2}\exp\left(-\frac{x^2}{2}(\gamma^{2}-\sigma^{2}_{\varepsilon})\right)\times \left\{\frac{1}{\sqrt{2\pi}}\frac{1}{(\gamma^{2}-\sigma^{2}_{\varepsilon})^{1/2}}\left((-ia_{1}\gamma^{2}-3ia_{2}\gamma^{4})x+ia_{2}\gamma^{6}x^{3}\right)\right\}\\
&=&-i\overline{C}\left(A_1x-A_2x^{3}\right)g_{0,\frac{1}{(\gamma^{2}-\sigma^{2}_{\varepsilon})}}(x),
\end{eqnarray*}
with $\overline{C}=\frac{1}{\sqrt{2\pi}}\frac{1}{(\gamma^{2}-\sigma_{\varepsilon}^2)^{1/2}}$ and $A_1=a_{1}\gamma^{2}+3a_{2}\gamma^{4}=\gamma^2\frac{(1+\phi^2)}{(1-\phi^2)}$ and $A_2=a_{2}\gamma^{6}=\gamma^4\frac{\phi^2}{(1-\phi^2)}.$ The Fourier transform of the function $u_{\frac{\partial l_{\theta}}{\partial \phi}}(x)$ is given by: 

\begin{eqnarray}
u^{*}_{\frac{\partial l_{\theta}}{\partial \phi}}(x)&=&-i\overline{C}\int \exp\left(iyx\right)\left(A_1y -A_2y^{3}\right)g_{\left(0,\frac{1}{(\gamma^{2}-\sigma_{\varepsilon}^2)}\right)}(y)dy\nonumber\\
&=&-\overline{C} A_1\frac{\p}{\p x}\E\left[\exp\left(ixG\right)\right]-\overline{C} A_2\frac{\p^{3}}{\p x^{3}}\E\left[\exp\left(ixG\right)\right] \text{ where G} \sim \mathcal{N}\left(0, \frac{1}{(\gamma^{2}-\sigma^{2}_{\varepsilon})}\right)\nonumber \\
&=&-\overline{C} A_1\frac{\p}{\p x}\left(\exp\left(-\frac{x^{2}}{2(\gamma^{2}-\sigma^{2}_{\varepsilon})}\right)\right)-\overline{C} A_2\frac{\p^{3}}{\p x^{3}}\left(\exp\left(-\frac{x^{2}}{2(\gamma^{2}-\sigma^{2}_{\varepsilon})}\right)\right)\nonumber\\
&=&\left(\Psi^{\phi_0}_{1}x+\Psi^{\phi_0}_{2}x^{3}\right)\exp\left(-\frac{x^{2}}{2(\gamma^{2}-\sigma^{2}_{\varepsilon})}\right),\label{u1}
\end{eqnarray}
with $\Psi^{\phi_0}_{1}=\overline{C}\left(\frac{A_1}{(\gamma^{2}-\sigma^{2}_{\varepsilon})}-\frac{3A_2}{(\gamma^{2}-\sigma_{\varepsilon}^2)^{2}}\right)$ and $\Psi^{\phi_0}_{2}=\overline{C}\left(\frac{A_2}{(\gamma^{2}-\sigma_{\varepsilon}^2)^{3}}\right).$ By the same arguments, we obtain:

\begin{equation}
u^{*}_{\frac{\partial l_{\theta}}{\partial \sigma^2}}(x)
=\left(\Psi^{\sigma_0^{2}}_{1}x+\Psi^{\sigma_0^{2}}_{2}x^{3}\right)\exp\left(-\frac{x^{2}}{2(\gamma^{2}-\sigma_{\varepsilon}^2)}\right),\label{u2}
\end{equation}
with $\Psi^{\sigma_0^{2}}_{1}=\overline{C}\left(\frac{B_1}{(\gamma^{2}-\sigma_{\varepsilon}^2)}-\frac{3B_2}{(\gamma^{2}-\sigma_{\varepsilon}^2)^{2}}\right), \Psi^{\sigma_0^{2}}_{2}=\overline{C}\left(\frac{B_2}{(\gamma^{2}-\sigma_{\varepsilon}^2)^{3}}\right),
B_1=b_{1}\gamma^{2}+3b_{2}\gamma^{4}=\frac{\phi}{(1-\phi^2)}$ and  $B_2=b_{2}\gamma^{6}=\gamma^2\frac{\phi}{2(1-\phi^2)}.$\\

Hence, for some $\delta >0$, $\E\left[\left|Y_{2}u^*_{\nabla_{\theta}l_{\theta}}(Y_{1}) \right|^{2+\delta}\right]$ is finite if:

\begin{eqnarray*}
&&\E\left[\left|\left(\Psi^{\phi_0}_{1}Y_{1}Y_{2}+\Psi^{\phi_0}_{2}Y_{1}^{3}Y_{2}\right)\exp\left(-\frac{Y_{1}^{2}}{2(\gamma^{2}-\sigma_{\varepsilon}^2)}\right)\right|^{2+\delta}\right]<\infty,\\
&&\E\left[\left|\left(\Psi^{\sigma_0^{2}}_{1}Y_{1}Y_{2}+\Psi^{\sigma_0^{2}}_{2}Y_{1}^{3}Y_{2}\right)\exp\left(-\frac{Y_{1}^{2}}{2(\gamma^{2}-\sigma_{\varepsilon}^2)}\right)\right|^{2+\delta}\right]<\infty,
\end{eqnarray*}
which is satisfied by the existence of all moments of the pair $\Y_{1}$.  One can check that the Hessian local assumption \textbf{(T)} is also satisfied by the same arguments.

\subsubsection{Explicit form of the Covariance matrix}\label{EFCM}

\begin{lemma}\label{hessienne_application}
The matrix $\Sigma(\theta_0)$ in the  Gaussian AR(1) model is given by:
\begin{equation*}
\Sigma(\theta_0)=V^{-1}_{\theta_0}\Omega(\theta_{0})V^{-1}_{\theta_0}
\end{equation*}
with 
\begin{eqnarray*}
V_{\theta_{0}}
=\frac{1}{8\sqrt{\pi}(1-\phi_0^2)^2}
\begin{pmatrix}
\gamma_0(7\phi_0^4-4\phi_0^2+4) & \frac{-5\phi^{5}_0+3\phi^{3}_0+2\phi_{0}}{2\gamma_0(1-\phi^{2}_0)}\\
\frac{-5\phi^{5}_0+3\phi^{3}_0+2\phi_0}{2\gamma_0(1-\phi^{2}_0)} & \frac{7\phi_0^2}{4\gamma_0^3}
\end{pmatrix},
\end{eqnarray*}
and
\begin{equation*}
\Omega(\theta_{0})=\Omega_0(\theta_{0})+2\sum_{j>1}^{\infty}\Omega_{j-1}(\theta_{0}) = 4\left[P_{2}-P_1\right]+8\sum_{j>1}^{\infty}(\tilde{C}_{j-1}-P_1)
\end{equation*}
where:
\begin{equation*}
P_{1}=\begin{pmatrix}
\frac{\phi_{0}^2\gamma_{0}^2(2-\phi_{0}^2)^{2}}{64\pi(1-\phi_{0}^2)^2} & \frac{\phi_{0}^3(2-\phi_{0}^2)}{128\pi(1-\phi_{0}^2)^2}\\
\frac{\phi_{0}^3(2-\phi_{0}^2)}{128\pi(1-\phi_{0}^2)^2} & \frac{\phi_{0}^4}{256\pi(1-\phi_{0}^2)^2\gamma_{0}^2}
\end{pmatrix},
\end{equation*}
and $P_2$ is the $2\times 2$ symmetric matrix multiplied by a factor $\frac{1}{\sqrt{\pi (\gamma_0^{2}-\sigma^{2}_{\varepsilon})}}$ and its coefficients $(P^2_{lm})_{1\leq l,m \leq 2}$ are given by:

\begin{eqnarray*}
P^2_{11}&=&\left(\Psi^{\phi_0}_{1}\right)^{2}\mathcal{F}\tilde{V}_{1}\left(\tilde{V}_{2}+3\frac{\phi_{0}^{2}\gamma_{0}^{4}}{(\gamma_0^2+\sigma_{\varepsilon}^2)^2}\tilde{V}_{1}\right)+15\left(\Psi^{\phi_0}_{2}\right)^{2}\mathcal{F}\tilde{V}_{1}^{3}\left(\tilde{V}_{2}+7\frac{\phi_{0}^{2}\gamma_{0}^{4}}{(\gamma_0^2+\sigma_{\varepsilon}^2)^2}\tilde{V}_{1}\right)+6\Psi^{\phi_0}_{1}\Psi^{\phi_0}_{2}\mathcal{F}\tilde{V}_{1}^{2}
\left(\tilde{V}_{2}+5\frac{\phi_{0}^{2}\gamma_{0}^{4}}{(\gamma_0^2+\sigma_{\varepsilon}^2)^2}\tilde{V}_{1}\right).\\
P^2_{22}&=&\left(\Psi^{\sigma_0^{2}}_{1}\right)^{2}\mathcal{F}\tilde{V}_{1}\left(\tilde{V}_{2}+3\frac{\phi_{0}^{2}\gamma_{0}^{4}}{(\gamma_0^2+\sigma_{\varepsilon}^2)^2}\tilde{V}_{1}\right)+15\left(\Psi^{\sigma_0}_{2}\right)^{2}\mathcal{F}\tilde{V}_{1}^{3}
\left(\tilde{V}_{2}+7\frac{\phi_{0}^{2}\gamma_{0}^{4}}{(\gamma_0^2+\sigma_{\varepsilon}^2)^2}\tilde{V}_{1}\right)+6\Psi^{\sigma_0^{2}}_{1}\Psi^{\sigma_0^{2}}_{2}
\mathcal{F}\tilde{V}_{1}^{2}\left(\tilde{V}_{2}+5\frac{\phi_{0}^{2}\gamma_{0}^{4}}{(\gamma_0^2+\sigma_{\varepsilon}^2)^2}\tilde{V}_{1}\right).\\
P^2_{12}&=&\Psi^{\phi_0}_{1}\Psi^{\sigma_0^{2}}_{1}\mathcal{F}\tilde{V}_{1}\left(\tilde{V}_{2}+3\frac{\phi_{0}^{2}\gamma_{0}^{4}}{(\gamma_0^2+\sigma_{\varepsilon}^2)^2}\tilde{V}_{1}\right)+15\Psi^{\phi_0}_{2}\Psi^{\sigma_0^{2}}_{2}\mathcal{F}\tilde{V}_{1}^{3}
\left(\tilde{V}_{2}+7\frac{\phi_{0}^{2}\gamma_{0}^{4}}{(\gamma_0^2+\sigma_{\varepsilon}^2)^2}\tilde{V}_{1}\right)+3\Psi^{\phi_0}_{1}\Psi^{\sigma_0^{2}}_{2}\mathcal{F}
\tilde{V}_{1}^{2}\left(\tilde{V}_{2}+5\frac{\phi_{0}^{2}\gamma_{0}^{4}}{(\gamma_0^2+\sigma_{\varepsilon}^2)^2}\tilde{V}_{1}\right)\\
&+&3\Psi^{\sigma_0^{2}}_{1}\Psi^{\phi_0}_{2}\mathcal{F}\tilde{V}_{1}^{2}\left(\tilde{V}_{2}+5\frac{\phi_{0}^{2}\gamma_{0}^{4}}{(\gamma_0^2+\sigma_{\varepsilon}^2)^2}\tilde{V}_{1}\right),
\end{eqnarray*}
with $\mathcal{F}=\frac{1}{(\sigma_{\varepsilon}^2+\gamma_{0}^{2})^{2}-\gamma_{0}^{4}\phi_{0}^{2}}\tilde{V}_{1}^{1/2}\tilde{V}_{2}^{1/2}$,  $\tilde{V}_{1}^{-1}=\frac{2}{(\gamma_{0}^2-\sigma_{\varepsilon}^2)}+\left(\frac{\gamma_{0}^2+\sigma_{\varepsilon}^2}{(\sigma_{\varepsilon}^2+\gamma_{0}^2)^{2}-\gamma_{0}^4\phi_{0}^{2}}\right)\left(1-\frac{\phi_{0}^{2}\gamma_{0}^{4}}{(\gamma_0^2+\sigma_{\varepsilon}^2)^2}\right)$, $\tilde{V}_{2}=\frac{(\gamma_{0}^2+\sigma_{\varepsilon}^2)^{2}-\phi_{0}^{2}\gamma_{0}^{4}}{(\gamma_{0}^2+\sigma_{\varepsilon}^2)}$, and:

\begin{eqnarray*}
\Psi^{\phi_0}_{1}&=&\frac{1}{\sqrt{2\pi}}\frac{1}{(\gamma_{0}^{2}-\sigma_{\varepsilon}^2)^{3/2}}\left(\frac{(1+\phi_{0}^{2})\gamma_{0}^{2}}{(1-\phi_{0}^{2})}-\frac{3\phi_{0}^{2}\gamma_{0}^{4}}{(1-\phi_{0}^{2})(\gamma_{0}^{2}-\sigma_{\varepsilon}^2)}\right).\\
\Psi^{\sigma_0^{2}}_{1}&=&\frac{1}{\sqrt{2\pi}}\frac{1}{(\gamma_{0}^{2}-\sigma_{\varepsilon}^2)^{3/2}}\left(\frac{\phi_{0}}{(1-\phi_{0}^{2})}-\frac{3\phi_{0}\gamma_{0}^{2}}{2(1-\phi_{0}^{2})(\gamma_{0}^{2}-\sigma_{\varepsilon}^2)}\right).\\
\Psi^{\phi_0}_{2}&=&\frac{1}{\sqrt{2\pi}}\frac{1}{(\gamma_{0}^{2}-\sigma_{\varepsilon}^2)^{7/2}}\frac{\gamma_{0}^{4}\phi_{0}^{2}}{(1-\phi_{0}^{2})}.\\
\Psi^{\sigma_0^{2}}_{2}&=&\frac{1}{\sqrt{2\pi}}\frac{1}{(\gamma_{0}^{2}-\sigma_{\varepsilon}^2)^{7/2}}\frac{\gamma_{0}^{2}\phi_{0}}{2(1-\phi_{0}^{2})}
\end{eqnarray*}

$\newline$

The covariance terms are given by:

\begin{equation*}
\tilde{C}_{j-1}=
\frac{\phi_{0}^{2}}{2\pi\gamma_{0}^{2}}\begin{pmatrix}
(4\phi_{0}^4-4\phi_{0}^2+1)  \tilde{c}_1(j) + \frac{2\phi_{0}^2(1-2\phi_{0}^2)}{\gamma_{0}^2}\tilde{c}_2(j)+\frac{\phi_{0}^4}{\gamma_{0}^4}\tilde{c}_3(j)& \frac{\phi_{0}(2\phi_{0}^2-1)}{2\gamma_{0}^2}\tilde{c}_1(j) + \frac{\phi_{0}(1-3\phi_{0}^2)}{2\gamma_{0}^4}\tilde{c}_2(j)+ \frac{\phi_{0}^3}{2\gamma_{0}^6} \tilde{c}_3(j)\\
\frac{\phi_{0}(2\phi_{0}^2-1)}{2\gamma_{0}^2}\tilde{c}_1(j) + \frac{\phi_{0}(1-3\phi_{0}^2)}{2\gamma_{0}^4}\tilde{c}_2(i)+ \frac{\phi_{0}^3}{2\gamma_{0}^6} \tilde{c}_3(j) & \frac{\phi_{0}^2}{4\gamma_{0}^4} \tilde{c}_1(j)-\frac{\phi_{0}^2} {2\gamma_{0}^6}\tilde{c}_2(j) + \frac{\phi_{0}^2}{4\gamma_{0}^8}  \tilde{c}_3(j)\\
\end{pmatrix},
\end{equation*}

with:

\begin{eqnarray*}
&&\tilde{c}_1(j)=\frac{1}{\gamma_{0}}(2-\phi_{0}^{2j})^{-1/2} V_j^{3/2}\left(\mathcal{V} + \frac{3\phi_{0}^{2j}V_j}{(2-\phi_{0}^{2j})^2} \right),\\
&&\tilde{c}_2(j)= \frac{3}{\gamma_{0}}(2-\phi_{0}^{2j})^{-1/2} V_j^{5/2} \left(\mathcal{V} + 5\frac{\phi_{0}^{2j}V_j}{(2-\phi_{0}^{2j})^2} \right),\\
&&\tilde{c}_3(j)=\frac{3(2-\phi_{0}^{2j})^{-1/2}}{\gamma_{0}}V_j^{5/2}\left[3\mathcal{V}^2+5V_j(4\mathcal{V}+2)\frac{\phi_{0}^{2j}}{(2-\phi_{0}^{2j})^2}+35V_j^{2}\frac{\phi_{0}^{4j}}{(2-\phi_{0}^{2j})^4}\right],
\end{eqnarray*}

where:

\begin{equation*}
V_j = \frac{\gamma_{0}^2(1-\phi_{0}^{2j})(2-\phi_{0}^{2j})}{(2-\phi_{0}^{2j})^2-\phi_{0}^{2j}} \text{ and } \mathcal{V} = \frac{\gamma_{0}^2(1-\phi_{0}^{2j})}{2-\phi_{0}^{2j}}
\end{equation*}
 
Moreover $\lim_{j\rightarrow \infty}\Omega_{j-1}(\theta_0)=0_{\mathcal{M}_{2\times2}}$.
\end{lemma}

\begin{remark}
In practice, for the computing of the covariance matrix $\Omega_{j-1}(\theta)$ that appears in Corollary \ref{lele}, we have truncated the infinite sum ($q_{trunc}=100$).
\end{remark}

\begin{proof}

\textbf{Calculus of $\nabla m$}\\

For all $x \in \R$, the function $l_{\theta}(x)$ is two times differentiable w.r.t $\theta$ on the compact subset $\Theta$. More precisely, note that since $\gamma^2 = \sigma^2/(1-\phi^2)$, it follows from the definition of the subset $\Theta$ that $(\gamma^2-\sigma_{\varepsilon}^2)>0$. So that for all $\y_{i}$ in $\R^{2}$ the function $m_{\theta}(\y_{i}): \theta \in \Theta \mapsto m_{\theta}(\y_{i})$ is differentiable and:

\begin{eqnarray*}
\nabla_{\theta}(m_{\theta}(\y_{i}))&=&\left(\frac{\p m_{\theta}(\y_{i})}{\p \phi}, \frac{\p m_{\theta}(\y_{i})}{\p \sigma^{2}}\right)'\\
&=&\left(\frac{\p \left\|l_{\theta}\right\|_{2}^2}{\p \phi}-2y_{i+1}u^{*}_{\frac{\p l_{\theta}}{\p \phi}}(y_{i}), \frac{\p \left\|l_{\theta}\right\|_{2}^2}{\p \sigma^{2}}-2y_{i+1}u^{*}_{\frac{\p l_{\theta}}{\p \sigma^{2}}}(y_{i})\right)',\\
\end{eqnarray*}
with:

\begin{eqnarray*}
&&\frac{\partial}{\partial \phi}||l_{\theta}||_{2}^{2}=\frac{\phi\gamma(2-\phi^2)}{4\sqrt{\pi}(1-\phi^2)},\\
&&\frac{\partial}{\partial \sigma^2}||l_{\theta}||_{2}^{2}=\frac{\phi^2}{8\sqrt{\pi}(1-\phi^2)}.
\end{eqnarray*}

And, the function $u^{*}_{\frac{\p l_{\theta}}{\p \phi}}(x)$ and $u^{*}_{\frac{\p l_{\theta}}{\p \sigma^{2}}}(x)$ are given in Eq.(\ref{u1})-(\ref{u2}). Therefore,

\begin{equation}\label{dmt}
\nabla_{\theta}m_{\theta}(\y_{i})=\begin{pmatrix}\left(\frac{\phi_{0}\gamma_{0}(2-\phi_{0}^2)}{4\sqrt{\pi}(1-\phi_{0}^2)}-2y_{i+1}\left(\Psi^{\phi_{0}}_{1}y_{i}+\Psi^{\phi_{0}}_{2}y_{i}^{3}\right)\exp\left(-\frac{y_{i}^{2}}{2(\gamma_{0}^{2}-\sigma_{\varepsilon}^2)}\right)\right)\\ \left(\frac{\phi_{0}^2}{8\sqrt{\pi}(1-\phi_{0}^2)}-2y_{i+1}\left(\Psi^{\sigma_{0}^{2}}_{1}y_{i}+\Psi^{\sigma_{0}^{2}}_{2}y_{i}^{3}\right)\exp\left(-\frac{y_{i}^{2}}{2(\gamma_{0}^{2}-\sigma_{\varepsilon}^2)}\right)\right)\end{pmatrix}\text{ at the point } \theta_{0}.
\end{equation}

\newpage
\textbf{Calculus of $P_{1}$:} Recall that we have:

\begin{eqnarray*}
&&P_1=\E\left[b_{\theta_{0}}(X_1) \left(\nabla_{\theta}l_{\theta}(X_1) \right)\right]\E\left[b_{\phi_{0}}(X_1) \left(\nabla_{\theta}l_{\theta}(X_1) \right)\right]'\\
&&P_{2}=\E\left[Y^{2}_{2}\left(u^{*}_{\nabla_{\theta} l_{\theta}}(Y_1)\right)^{2}\right].
\end{eqnarray*}
And the moments $(\mu_{2k})_{k \in \mathbb{N}}$ of a centered Gaussian random variable with variance $\sigma^2$ are given by: 

$$\mu_{2k} = \left(\frac{(2k)!}{2^kk!}\right)\sigma^{2k}.$$

We define by $P(x)$ a  polynomial function of ordinary degree. We are interested in the calculus of $\displaystyle \E\left[ P(X) g_{0,\gamma^{2}}(X)\right],$
where $X \sim \mathcal{N}(0,\gamma^2)$. We have:

\begin{eqnarray*}
\E\left[ P(X) g_{0,\gamma^2}(X)\right] &=& \int P(x) \frac{1}{\sqrt{2\pi}\gamma} e^{-\frac{x^2}{2\gamma^2}} \frac{1}{\sqrt{2\pi}\gamma} e^{-\frac{x^2}{2\gamma^2}} dx\\
&=& \frac{1}{2\pi\gamma^{2}} \int P(x)  e^{-\frac{x^2}{\gamma^2}} dx\\
&=& \frac{1}{2\sqrt{\pi}\gamma} \E\left[ P(\bar{X}) \right], 
\end{eqnarray*}
where $\displaystyle \bar{X} \sim \mathcal{N}\left(0,\frac{\gamma^2}{2}\right)$.\\

Denote by $B_1$ the constant $\frac{1}{2\sqrt{\pi}\gamma_0}$. We obtain:
\begin{eqnarray*}\label{C2}
P_1&=& \begin{pmatrix}
\E\left[b_{\phi_{0}}(X_1)\frac{\p l_{\theta}}{\p \phi}(\theta,X_1) \right]^2 &\E\left[b_{\phi_{0}}(X_1)\frac{\p l_{\theta}}{\p \phi}(\theta,X_1) \right]\E\left[b_{\phi_{0}}(X_1)\frac{\p l_{\theta}}{\p \sigma^{2}}(\theta,X_1) \right]\\
\E\left[b_{\phi_{0}}(X_1)\frac{\p l_{\theta}}{\p \phi}(\theta,X_1) \right]\E\left[b_{\phi_{0}}(X_1)\frac{\p l_{\theta}}{\p \sigma^{2}}(\theta,X_1) \right] &\E\left[b_{\phi_{0}}(X_1)\frac{\p l_{\theta}}{\p \sigma^{2}}(\theta,X_1) \right]^2\\ 
\end{pmatrix}\\
&=&B_1^2\phi_{0}^2\begin{pmatrix}
\E\left[ H_{11}(\bar{X})\right]^{2} & \E\left[ H_{12}(\bar{X}) \right] \E\left[ H_{21}(\bar{X}) \right]\\
\E\left[ H_{21}(\bar{X}) \right] \E\left[ H_{12}(\bar{X}) \right] & \E\left[ H_{22}(\bar{X}) \right]^2\\ 
\end{pmatrix},\\
\end{eqnarray*}
where $\displaystyle \bar{X} \sim \mathcal{N}\left(0,\frac{\gamma_0^2}{2} \right)$. The polynomials $\left(H_{ij}(x)\right)_{1 \leq i,j \leq2}$ are given by:

\begin{eqnarray*}
&&H_{11}(x)=  \left(a_1x^2 + a_2x^4 \right), \\
&&H_{12}(x)= \left(b_1x^2 + b_2x^4 \right), \\
&&H_{21}(x)= \left(a_1x^2 + a_2x^4 \right),\\
&&H_{22}(x)=  \left( b_1x^2 + b_2x^4 \right).
\end{eqnarray*}

Lastly, by replacing the terms $B_1$, $a_1$, and $a_2$ by their expressions given in Eq.(\ref{a1}) at the point $\theta_{0}$, we obtain:

\begin{equation*}
P_1=\E\left[b_{\phi_{0}}(X_1) \left(\nabla_{\theta}l_{\theta}(X_1) \right)\right]\E\left[b_{\phi_{0}}(X_1) \left(\nabla_{\theta}l_{\theta}(X_1) \right)\right]'
=\begin{pmatrix}
\frac{\phi_{0}^2\gamma_{0}^2(2-\phi_0^2)^2}{64\pi(1-\phi_{0}^2)^2} & \frac{\phi_{0}^3(2-\phi_{0}^2)}{128\pi(1-\phi_{0}^2)^2}\\
\frac{\phi_{0}^3(2-\phi_{0}^2)}{128\pi(1-\phi_{0}^2)^2} & \frac{\phi_{0}^4}{256\pi(1-\phi_{0}^2)^2\gamma_{0}^2}
\end{pmatrix}.
\end{equation*}

\newpage
\textbf{Calculus of $P_2$:}\\

\begin{eqnarray*}
\E\left[\left(Y_{2}u^{*}_{\nabla_{\theta} l_{\theta}}(Y_{1})\right)\left(Y_{2}u^{*}_{\nabla_{\theta} l_{\theta}}(Y_{1})\right)'\right]&=&
\begin{pmatrix}
\E\left[Y^{2}_{2}\left(u^{*}_{\frac{\p l_{\theta}}{\p \phi} }(Y_{1})\right)^{2}\right] & \E\left[Y^{2}_{2}\left(u^{*}_{\frac{\p l_{\theta}}{\p \phi}}(Y_{1})\right)\left(u^{*}_{\frac{\p l_{\theta}}{\p \sigma^{2}} }(Y_{1})\right)\right]\\
\E\left[Y^{2}_{2}\left(u^{*}_{\frac{\p l_{\theta}}{\p \sigma^{2}}}(Y_{1})\right)\left(u^{*}_{\frac{\p l_{\theta}}{\p \phi}}(Y_{1})\right)\right] & \E\left[Y^{2}_{2}\left(u^{*}_{\frac{\p l_{\theta}}{\p \sigma^{2}}}(Y_{1})\right)^{2}\right]
\end{pmatrix}.
\end{eqnarray*}

We have:

\begin{eqnarray}\label{zz}
\left(2\pi \frac{(\gamma_0^2-\sigma_{\varepsilon}^2)}{2} \right)^{-1/2}\E\left[Y^{2}_{2}\left(u^{*}_{\frac{\p l_{\theta}}{\p \phi} }(Y_{1})\right)^{2}\right]&=&\E\left[Y^{2}_{2}\left(\Psi^{\phi_{0}}_{1}Y_{1}+\Psi^{\phi_{0}}_{2}Y^{3}_{1}\right)^{2}\times g_{\left(0, \frac{(\gamma_{0}^{2}-\sigma_{\varepsilon}^2)}{2}\right)}\right]\nonumber\\
&=&\left(\Psi^{\phi_{0}}_{1}\right)^{2}\E\left[Y^{2}_{2}Y^{2}_{1}\times g_{\left(0, \frac{(\gamma^{2}-\sigma_{\varepsilon}^2)}{2}\right)}\right]+\left(\Psi^{\phi_{0}}_{2}\right)^{2}\E\left[Y^{2}_{2}Y^{6}_{1}\times g_{\left(0, \frac{(\gamma_{0}^{2}-\sigma_{\varepsilon}^2)}{2}\right)}\right]\nonumber\\
&& \quad +2\Psi^{\phi_{0}}_{1}\Psi^{\phi_{0}}_{2}\E\left[Y^{2}_{2}Y^{4}_{1}\times g_{\left(0, \frac{(\gamma_{0}^{2}-\sigma_{\varepsilon}^2)}{2}\right)}\right].
\end{eqnarray}

$\newline$
The density of $\Y_{1}$ is $g_{(0, \mathcal{J}_{\theta_0})}$. Then, $g_{(0,\mathcal{J}_{\theta_0})}\times \exp\left(-\frac{y^{2}_{1}}{(\gamma_{0}^2-\sigma_{\varepsilon}^2)}\right)$ is equal to:

\begin{eqnarray*}
&&\frac{1}{2\pi}\frac{1}{\left( (\sigma_{\varepsilon}^2+\gamma_{0}^2 )^{2}-\gamma_{0}^4\phi_{0}^{2} \right)^{1/2}}\exp\left( -\frac{1}{2} \frac{1}{\left( (\sigma_{\varepsilon}^2+\gamma_{0}^2 )^{2}-\gamma_{0}^4\phi_{0}^{2} \right)} \left((\sigma_{\varepsilon}^2+\gamma_{0}^2)(y^{2}_{1}+y^{2}_{2})-2\phi_{0}\gamma_{0}^{2}y_{1}y_{2}\right)\right)\\
&& \quad \times \exp\left( -\frac{1}{2}\frac{2}{(\gamma_{0}^{2}-\sigma_{\varepsilon}^2)} y^{2}_{1}\right)\\
&=&\frac{1}{2\pi}\frac{1}{\left( (\sigma_{\varepsilon}^2+\gamma_{0}^2 )^{2}-\gamma_{0}^4\phi_{0}^{2} \right)^{1/2}}
\times \exp\left(-\frac{1}{2}y^{2}_{1}\left(\frac{2}{(\gamma_{0}^{2}-\sigma_{\varepsilon}^2)}-\frac{(\gamma_{0}^{2}+\sigma_{\varepsilon}^2)}{\left( (\sigma_{\varepsilon}^2+\gamma_{0}^2 )^{2}-\gamma_{0}^4\phi_{0}^{2} \right)}\right)\right)\\
&& \quad \times \exp\left(-\frac{1}{2}y^{2}_{2}\left(\frac{(\gamma_{0}^{2}+\sigma_{\varepsilon}^2)}{\left( (\sigma_{\varepsilon}^2+\gamma_{0}^2 )^{2}-\gamma_{0}^4\phi_{0}^{2} \right)}\right)\right) \exp\left(-\frac{1}{2}y_{1}y_{2}\left(\frac{2\phi_{0}\gamma_{0}^{2}}{\left( (\sigma_{\varepsilon}^2+\gamma_{0}^2 )^{2}-\gamma_{0}^4\phi_{0}^{2} \right)}\right)\right)\\
&=&\frac{1}{2\pi}\frac{1}{\left( (\sigma_{\varepsilon}^2+\gamma_{0}^2 )^{2}-\gamma_{0}^4\phi_{0}^2 \right)^{1/2}}\exp\left( -\frac{1}{2}\left(\tilde{V}_{2}^{-1}\left(y_{2}-\frac{\phi_{0} \gamma_{0}^{2}}{\gamma_0^2 + \sigma_{\varepsilon}^2}y_{1}\right)^{2}\right)\right)\times \exp\left(-\frac{1}{2}y_{1}^{2}\tilde{V}_{1}^{-1}\right),
\end{eqnarray*}
with $\tilde{V}_{1}^{-1}=\frac{2}{(\gamma_{0}^2-\sigma_{\varepsilon}^2)}+\left(\frac{\gamma_{0}^2+\sigma_{\varepsilon}^2}{(\sigma_{\varepsilon}^2+\gamma_{0}^2)^{2}-\gamma_{0}^4\phi_{0}^{2}}\right)\left(1-\frac{\phi_{0}^{2}\gamma_{0}^{4}}{(\gamma_0^2 + \sigma_{\varepsilon}^2)^2}\right)$ and $\tilde{V}_{2}=\frac{(\gamma_{0}^2+\sigma_{\varepsilon}^2)^{2}-\phi_{0}^{2}\gamma_{0}^{4}}{(\gamma_{0}^2+\sigma_{\varepsilon}^2)}$.\\

Then, we obtain:

\begin{eqnarray*}
g_{(0,\mathcal{J}_{\theta_0})}\times \exp\left(-\frac{y^{2}_{1}}{(\gamma_{0}^2-\sigma_{\varepsilon}^2)}\right) &=&\frac{1}{((\sigma_{\varepsilon}^2+\gamma_{0}^2 )^{2}-\gamma_{0}^4\phi_{0}^{2})^{1/2}}\tilde{V}_{1}^{1/2}\tilde{V}_{2}^{1/2}g_{(\phi_{0} \gamma_{0}^{2}y_{1}/(\gamma_0^2 + \sigma_{\varepsilon}^2),\tilde{V}_{2})}(y_{2})g_{(0,\tilde{V}_{1})}(y_{1}).
\end{eqnarray*}

In the following, we set $\mathcal{F}=\frac{1}{(\sigma_{\varepsilon}^2+\gamma_{0}^2 )^{2}-\gamma_{0}^4\phi_{0}^{2}}\tilde{V}_{1}^{1/2}\tilde{V}_{2}^{1/2}$. Now, we can compute the moments:

\begin{eqnarray*}
\left(\Psi^{\phi_{0}}_{1}\right)^{2}\E\left[Y^{2}_{2}Y^{2}_{1}\exp\left(-\frac{Y^{2}_{1}}{(\gamma_{0}^2-\sigma_{\varepsilon}^2)}\right)\right]&=&\left(\Psi^{\phi_{0}}_{1}\right)^{2}\mathcal{F} \int y_{1}^{2}g_{(0,\tilde{V}_{1})}(y_{1}) dy_{1} \int y_{2}^{2}g_{(\phi_{0} \gamma_{0}^{2}y_{1}/(\gamma_0^2 + \sigma_{\varepsilon}^2),\tilde{V}_{2})}(y_{2})dy_{2}\\
&=&\left(\Psi^{\phi_{0}}_{1}\right)^{2}\mathcal{F}\int y_{1}^{2}g_{(0,V_{1})}(y_{1}) dy_{1} \E\left[G^{2}\right] \text{ where G} \sim \mathcal{N}(\phi_{0} \gamma_{0}^{2}y_{1}/(\gamma_0^2 + \sigma_{\varepsilon}^2), \tilde{V}_{2})\\
&=&\left(\Psi^{\phi_{0}}_{1}\right)^{2}\mathcal{F}\int \left(\tilde{V}_{2}y_{1}^{2}+\frac{\phi_{0}^{2}\gamma_{0}^{4}}{(\gamma_0^2 + \sigma_{\varepsilon}^2)^2}y_{1}^{4}\right)g_{(0,\tilde{V}_{1})}(y_{1})dy_{1}\\
&=&\left(\Psi^{\phi_{0}}_{1}\right)^{2}\mathcal{F}\tilde{V}_{2}\tilde{V}_{1}+3\left(\Psi^{\phi_{0}}_{1}\right)^{2}\mathcal{F}\frac{\phi_{0}^{2}\gamma_{0}^{4}}{(\gamma_0^2 + \sigma_{\varepsilon}^2)^2}\tilde{V}_{1}^{2}\\
&=&\left(\Psi^{\phi_{0}}_{1}\right)^{2}\mathcal{F}\tilde{V}_{1}\left(\tilde{V}_{2}+3\frac{\phi_{0}^{2}\gamma_{0}^{4}}{(\gamma_0^2 + \sigma_{\varepsilon}^2)^2}\tilde{V}_{1}\right).
\end{eqnarray*}

In a similar manner, we have:

\begin{eqnarray*}
\left(\Psi^{\phi_{0}}_{2}\right)^{2}\E\left[Y^{2}_{2}Y^{6}_{1}\exp\left(-\frac{Y^{2}_{1}}{(\gamma_{0}^2-\sigma_{\varepsilon}^2)}\right)\right]&=&\left(\Psi^{\phi_{0}}_{2}\right)^{2}\mathcal{F}\int y_{1}^{6}g_{(0,V_{1})}(y_{1}) dy_{1} \E\left[G^{2}\right] \text{ where }G \sim \mathcal{N}(\phi_{0} \gamma_{0}^{2}y_{1}/(\gamma_0^2 + \sigma_{\varepsilon}^2), \tilde{V}_{2})\\
&=&\left(\Psi^{\phi_{0}}_{2}\right)^{2}\mathcal{F}\int \left(\tilde{V}_{2}y_{1}^{6}+\frac{\phi_{0}^{2}\gamma_{0}^{4}}{(\gamma_0^2 + \sigma_{\varepsilon}^2)^2}y_{1}^{8}\right)g_{(0,\tilde{V}_{1})}(y_{1})dy_{1}\\
&=&15\left(\Psi^{\phi_{0}}_{2}\right)^{2}\mathcal{F}\tilde{V}_{2}\tilde{V}_{1}^{3}+105\left(\Psi^{\phi_{0}}_{2}\right)^{2}\mathcal{F}\frac{\phi_{0}^{2}\gamma_{0}^{4}}{(\gamma_0^2 + \sigma_{\varepsilon}^2)^2}\tilde{V}_{1}^{4}\\
&=&15\left(\Psi^{\phi_{0}}_{2}\right)^{2}\mathcal{F}\tilde{V}_{1}^{3}\left(\tilde{V}_{2}+7\frac{\phi_{0}^{2}\gamma_{0}^{4}}{(\gamma_0^2 + \sigma_{\varepsilon}^2)^2}\tilde{V}_{1}\right),
\end{eqnarray*}
and

\begin{eqnarray*}
2\Psi^{\phi_{0}}_{1}\Psi^{\phi_{0}}_{2}\E\left[Y^{2}_{2}Y^{4}_{1}\exp\left(-\frac{Y^{2}_{1}}{(\gamma_{0}^2-\sigma_{\varepsilon}^2)}\right)\right]&=&2\Psi^{\phi_{0}}_{1}\Psi^{\phi_{0}}_{2}\mathcal{F}\int y_{1}^{4}g_{(0,V_{1})}(y_{1}) dy_{1} \E\left[G^{2}\right] \text{ where }G \sim \mathcal{N}(\phi_{0} \gamma_{0}^{2}y_{1}/(\gamma_0^2 + \sigma_{\varepsilon}^2), \tilde{V}_{2})\\
&=&2\Psi^{\phi_{0}}_{1}\Psi^{\phi_{0}}_{2}\mathcal{F}\int \left(\tilde{V}_{2}y_{1}^{4}+\frac{\phi_{0}^{2}\gamma_{0}^{4}}{(\gamma_0^2 + \sigma_{\varepsilon}^2)^2}y_{1}^{6}\right)g_{(0,\tilde{V}_{1})}(y_{1})dy_{1}\\
&=&6\Psi^{\phi_{0}}_{1}\Psi^{\phi_{0}}_{2}\mathcal{F}\tilde{V}_{2}\tilde{V}_{1}^{2}+30\Psi^{\phi_{0}}_{1}\Psi^{\phi_{0}}_{2}\mathcal{F}
\frac{\phi_{0}^{2}\gamma_{0}^{4}}{(\gamma_0^2 + \sigma_{\varepsilon}^2)^2}\tilde{V}_{1}^{3}\\
&=&6\Psi^{\phi_{0}}_{1}\Psi^{\phi_{0}}_{2}\mathcal{F}\tilde{V}_{1}^{2}\left(\tilde{V}_{2}+5\frac{\phi_{0}^{2}\gamma_{0}^{4}}{(\gamma_0^2 + \sigma_{\varepsilon}^2)^2}\tilde{V}_{1}\right).
\end{eqnarray*}

By replacing all the terms of Eq.(\ref{zz}) we obtain:

\begin{eqnarray}\label{a}
\E\left[Y^{2}_{2}\left(u^{*}_{\frac{\p l_{\theta}}{\p \phi} }(Y_1)\right)^{2}\right]&=&\left(\Psi^{\phi_{0}}_{1}\right)^{2}\mathcal{F}\tilde{V}_{1}\left(\tilde{V}_{2}+3\frac{\phi_{0}^{2}\gamma_{0}^{4}}{(\gamma_0^2+\sigma_{\varepsilon}^2)^2}\tilde{V}_{1}\right)+15\left(\Psi^{\phi_{0}}_{2}\right)^{2}\mathcal{F}\tilde{V}_{1}^{3}
\left(\tilde{V}_{2}+7\frac{\phi_{0}^{2}\gamma_{0}^{4}}{(\gamma_0^2 + \sigma_{\varepsilon}^2)^2}\tilde{V}_{1}\right)\nonumber\\
&+&6\Psi^{\phi_{0}}_{1}\Psi^{\phi_{0}}_{2}\mathcal{F}\tilde{V}_{1}^{2}\left(\tilde{V}_{2}+5\frac{\phi_{0}^{2}\gamma_{0}^{4}}{(\gamma_0^2 + \sigma_{\varepsilon}^2)^2}\tilde{V}_{1}\right),
\end{eqnarray}
and

\begin{eqnarray}\label{b}
\E\left[Y^{2}_{2}\left(u^{*}_{\frac{\p l_{\theta}}{\p \sigma^{2}} }(Y_{1})\right)^{2}\right]&=&\left(\Psi^{\sigma_0^{2}}_{1}\right)^{2}\mathcal{F}\tilde{V}_{1}\left(\tilde{V}_{2}+3
\frac{\phi_{0}^{2}\gamma_{0}^{4}}{(\gamma_0^2 + \sigma_{\varepsilon}^2)^2}\tilde{V}_{1}\right)
+15\left(\Psi^{\sigma_0^{2}}_{2}\right)^{2}\mathcal{F}\tilde{V}_{1}^{3}\left(\tilde{V}_{2}+7\frac{\phi_{0}^{2}\gamma_{0}^{4}}{(\gamma_0^2 + \sigma_{\varepsilon}^2)^2}\tilde{V}_{1}\right)\nonumber\\
&+&6\Psi^{\sigma_0^{2}}_{1}\Psi^{\sigma_0^{2}}_{2}\mathcal{F}\tilde{V}_{1}^{2}\left(\tilde{V}_{2}+5\frac{\phi_{0}^{2}\gamma_{0}^{4}}{(\gamma_0^2 + \sigma_{\varepsilon}^2)^2}\tilde{V}_{1}\right),
\end{eqnarray}
and

\begin{eqnarray}
&&\E\left[Y^{2}_{2}\left(u^{*}_{\frac{\p l_{\theta}}{\p \phi} }(Y_{1})\right)\left(u^{*}_{\frac{\p}{\p \sigma^{2}} l_{\theta}}(Y_{1})\right)\right]\nonumber\\
&&=\Psi^{\phi_{0}}_{1}\Psi^{\sigma_0^{2}}_{1}\E\left[Y^{2}_{2}Y^{2}_{1}\times g_{\left(0, \frac{(\gamma_{0}^{2}-\sigma_{\varepsilon}^2)}{2}\right)}\right]+\Psi^{\phi_{0}}_{2}\Psi^{\sigma_0^{2}}_{2}\E\left[Y^{2}_{2}Y^{6}_{1}\times g_{\left(0, \frac{(\gamma_{0}^{2}-\sigma_{\varepsilon}^2)}{2}\right)}\right]\nonumber\\
&&\quad+\Psi^{\phi_{0}}_{1}\Psi^{\sigma_0^{2}}_{2}\E\left[Y^{2}_{2}Y^{4}_{1}\times g_{\left(0, \frac{(\gamma_{0}^{2}-\sigma_{\varepsilon}^2)}{2}\right)}\right]+\Psi^{\phi_{0}}_{2}\Psi^{\sigma_0^{2}}_{1}\E\left[Y^{2}_{2}Y^{4}_{1}\times g_{\left(0, \frac{(\gamma_{0}^{2}-\sigma_{\varepsilon}^2)}{2}\right)}\right]\nonumber\\
&&=\Psi^{\phi_{0}}_{1}\Psi^{\sigma_0^{2}}_{1}\mathcal{F}\tilde{V}_{1}\left(\tilde{V}_{2}+3\frac{\phi_{0}^{2}\gamma_{0}^{4}}{(\gamma_0^2 +\sigma_{\varepsilon}^2)^2}\tilde{V}_{1}\right)+15\Psi^{\phi_{0}}_{2}\Psi^{\sigma_0^{2}}_{2}\mathcal{F}\tilde{V}_{1}^{3}\left(\tilde{V}_{2}
+7\frac{\phi_{0}^{2}\gamma_{0}^{4}}{(\gamma_0^2 + \sigma_{\varepsilon}^2)^2}\tilde{V}_{1}\right)\nonumber\\
&&\quad+3\Psi^{\phi_{0}}_{1}\Psi^{\sigma_0^{2}}_{2}
\mathcal{F}\tilde{V}_{1}^{2}\left(\tilde{V}_{2}+5\frac{\phi_{0}^{2}\gamma_{0}^{4}}{(\gamma_0^2 + \sigma_{\varepsilon}^2)^2}\tilde{V}_{1}\right) +3\Psi^{\phi_{0}}_{2}\Psi^{\sigma_0^{2}}_{1}
\mathcal{F}\tilde{V}_{1}^{2}\left(\tilde{V}_{2}+5\frac{\phi_{0}^{2}\gamma_{0}^{4}}{(\gamma_0^2 + \sigma_{\varepsilon}^2)^2}\tilde{V}_{1}\right) .\label{c}
\end{eqnarray}

$\newline$
\textbf{Calculus of $\C \left( \nabla_{\theta}m_{\theta}(Y_{1}),\nabla_{\theta}m_{\theta}(Y_{j}) \right)$:} We want to compute:

$$\C\left(\nabla_{\theta}m_{\theta}(Y_{1}), \nabla_{\theta}m_{\theta}(Y_{j}) \right)=4\left[\tilde{C}_{j-1}-P_1\right].$$

Since we have already computed the terms of the matrix $P_1$, it remains to compute the terms of the covariance matrix $\tilde{C}_{j-1}$ given by:

\begin{equation*}
\tilde{C}_{j-1}=\E\left[b_{\phi_{0}}(X_1)b_{\phi_{0}}(X_j)\left(\nabla_{\theta}l_{\theta}(X_1)\right)\left(\nabla_{\theta}l_{\theta}(X_j)\right)'\right].
\end{equation*}

For all $j>1$, the pair $(X_1,X_j)$ has a multivariate normal density $g_{(0, \mathcal{W})}$ where $\mathcal{W}$ is given by:

\begin{equation*}
\mathcal{W}=\gamma_{0}^2\begin{pmatrix}
1 &  \phi_{0}^j\\
 \phi_{0}^j & 1\\ 
\end{pmatrix}
\text{ and } 
\mathcal{W}^{-1}=\frac{1}{\gamma_{0}^2(1-\phi_{0}^{2j})}\begin{pmatrix}
1 & - \phi_{0}^j\\
- \phi_{0}^j & 1\\ 
\end{pmatrix}.
\end{equation*}

The density of the couple $(X_1,X_j)$ is:
 
\begin{equation*}
g_{(0, \mathcal{W})}(x_1,x_j) =\frac{1}{2\pi} \det{(\mathcal{W})}^{-1/2} \exp\left(-\frac{1}{2} (x_1,x_j)^{'}\mathcal{W}^{-1}(x_1,x_j)\right).
\end{equation*}

We start by computing:

\begin{equation*}
g_{(0, \mathcal{W})}(x_1,x_j) \times \exp\left(-\frac{1}{2\gamma_{0}^2}\left(x_1^2+x_j^2\right) \right).
\end{equation*}

We have:

\begin{eqnarray*}
&&g_{(0, \mathcal{W})}(x_1,x_j) \times \exp\left(-\frac{1}{2\gamma_{0}^2}(x_1^2+x_j^2) \right)\\
&& = \frac{1}{2\pi}\det{(\mathcal{W})}^{-1/2} \exp\left(-\frac{1}{2(1-\phi_{0}^{2j})\gamma_{0}^2}\left(x_1^2(1-\phi_{0}^{2j})+x_j^2(1-\phi_{0}^{2j})+x_1^2-2\phi_{0}^jx_1x_j+x_j^2\right) \right),\\
&&=\frac{1}{2\pi}\det{(\mathcal{W})}^{-1/2}
\exp\left(-\frac{1}{2(1-\phi_{0}^{2j})\gamma_{0}^2}\left[\left(x_1^2(1-\phi_{0}^{2j})+x_1^2-2\phi_{0}^jx_1x_j\right)+\left(x_j^2(1-\phi_{0}^{2j})+x_j^2\right)\right] \right),\\
&&=\frac{1}{2\pi}\det{(\mathcal{W})}^{-1/2}
\exp\left(-\frac{1}{2(1-\phi_{0}^{2j})\gamma_{0}^2}\left[(2-\phi_{0}^{2j})\left(x_1^2-2\frac{\phi_{0}^j}{(2-\phi_{0}^{2j})}x_1x_j\right)+\left(x_j^2(1-\phi_{0}^{2j})+x_j^2\right)\right] \right),\\
&&=\frac{1}{2\pi}\det{(\mathcal{W})}^{-1/2}
\exp\left(-\frac{(2-\phi_{0}^{2j})}{2(1-\phi_{0}^{2j})\gamma_{0}^2}\left(x_1 - \frac{\phi_{0}^j x_j}{(2-\phi_{0}^{2j})}\right)^2 \right)\\
&&\quad \times \exp\left(-\frac{(2-\phi_{0}^{2j})}{2(1-\phi_{0}^{2j})\gamma_{0}^2}x_j^2\left(1-\frac{\phi_{0}^{2j}}{(2-\phi_{0}^{2j})^2}\right)\right).
\end{eqnarray*}

For all $j>1$, we define:

\begin{equation*}
V_j=\frac{\gamma_{0}^2(1-\phi_{0}^{2j})(2-\phi_{0}^{2j})}{(2-\phi_{0}^{2j})^2-\phi_{0}^{2j}} \text{ and } \mathcal{V} = \frac{\gamma_{0}^2(1-\phi_{0}^{2j})}{(2-\phi_{0}^{2j})}.
\end{equation*}

We can rewrite:

\begin{eqnarray*}
&&g_{(0, \mathcal{W})}(x_1,x_j) \times \exp\left(-\frac{1}{2\gamma_{0}^2}(x_1^2+x_j^2) \right)\\
&&=\frac{V_j^{1/2}}{\gamma_{0}} \frac{1}{\sqrt{2\pi}V_j^{1/2}} \exp\left(-\frac{1}{2V_j} x_j^2 \right) \times \frac{1}{(2-\phi_{0}^{2j})^{1/2}\sqrt{2\pi}\mathcal{V}^{1/2}}\exp\left( -\frac{1}{2 \mathcal{V}}\left(x_1 - \frac{\phi_{0}^ix_j}{(2-\phi_{0}^{2j})}\right)^2\right).
\end{eqnarray*}

So, by Fubini's Theorem, we obtain:

\begin{eqnarray*}
&&\E\left[X_1^2 X_j^2 \exp\left(-\frac{1}{2\gamma_{0}^2}\left(X_1^2+X_j^2\right) \right)\right]\\
&&=\frac{1}{\gamma_{0}}V_j^{1/2}\int x_j^2\frac{1}{\sqrt{2\pi}V_j^{1/2}} \exp\left(-\frac{1}{2V_j} x_j^2 \right) \int x_1^2 \frac{(2-\phi_{0}^{2j})^{-1/2}}{\sqrt{2\pi}\mathcal{V}^{1/2}}\exp\left( -\frac{1}{2\mathcal{V}}\left(x_1 - \frac{\phi_{0}^jx_j}{(2-\phi_{0}^{2j})}\right)^2\right)dx_1dx_j,\\
&&=\frac{1}{\gamma_{0}}V_j^{1/2}\int x_j^2\frac{1}{\sqrt{2\pi}V_j^{1/2}} \exp\left(-\frac{1}{2V_j} x_j^2 \right)(2-\phi_{0}^{2j})^{-1/2} \E[G^2]dx_j,
\end{eqnarray*}
where $\displaystyle G \sim \mathcal{N}\left(\frac{\phi_{0}^jx_j}{(2-\phi_{0}^{2j})},\mathcal{V} \right)$. Thus,
$\E[G^2] = \mathcal{V} + \left(\frac{\phi_{0}^jx_j}{(2-\phi_{0}^{2j})}\right)^2$. We obtain:

\begin{eqnarray}\label{c1i}
&&\E\left[X_1^2 X_j^2 \exp\left(-\frac{1}{2\gamma_{0}^2}\left(X_1^2+X_j^2\right) \right)\right]\nonumber\\
&&=(2-\phi_0^{2j})^{-1/2}\frac{V_j^{1/2}}{\gamma_{0}}(2-\phi_{0}^{2j})^{-1/2}\int x_j^2\left( \mathcal{V} + \left(\frac{\phi_{0}^jx_j}{(2-\phi_{0}^{2j})}\right)^2\right) \frac{1}{\sqrt{2\pi}V_j^{1/2}} \exp\left(-\frac{1}{2V_j} x_j^2 \right) dx_j\nonumber\\
&&=\frac{V_j^{1/2}}{\gamma_{0}} (2-\phi_{0}^{2j})^{-1/2}\left(\mathcal{V} \E[G_j^2] + \frac{\phi_{0}^{2j}}{(2-\phi_{0}^{2j})^2 \gamma_0}\E[G_j^4]\right)\nonumber\\
&&=\frac{V_j^{3/2}}{\gamma_{0}}(2-\phi_{0}^{2j})^{-1/2} \left(\mathcal{V} + \frac{3\phi_{0}^{2j}V_j}{(2-\phi_{0}^{2j})^2} \right)\nonumber\\
&&=\tilde{c}_1(j),
\end{eqnarray}
where $G_j \sim \mathcal{N}\left(0,V_j \right)$. Additionally, we have:

\begin{eqnarray} \label{c2i}
&&\E\left[X_1^2 X_j^4 \exp\left(-\frac{1}{2\gamma_{0}^2}\left(X_1^2+X_j^2\right) \right)\right]\nonumber\\
&&=\frac{ V_j^{1/2}}{\gamma_{0}}(2-\phi_{0}^{2j})^{-1/2}\mathcal{V} \E[G_j^4] + V_j^{1/2} \frac{(2-\phi_{0}^{2j})^{-1/2}\phi_{0}^{2j}}{(2-\phi_{0}^{2j})^2}\E[G_j^6],\nonumber\\
&&=\frac{3V_j^{5/2}}{\gamma_{0}}(2-\phi_{0}^{2j})^{-1/2}  \left(\mathcal{V} + 5\frac{\phi_{0}^{2j}V_j}{(2-\phi_{0}^{2j})^2} \right),\nonumber\\
&&=\tilde{c}_2(j).
\end{eqnarray}

Now, we are interested in $\E\left[X_1^4 X_j^4 \exp\left(-\frac{1}{2\gamma_{0}^2}(X_1^2+X_j^2) \right)\right]$. In a similar manner, we obtain:

\begin{eqnarray}\label{mom4}
&&\E\left[X_1^4 X_j^4 \exp\left(-\frac{1}{2\gamma_{0}^2}\left(X_1^2+X_j^2 \right)\right)\right]\nonumber\\
&&=\frac{V_j^{1/2}}{\gamma_{0}}\int x_j^4\frac{1}{\sqrt{2\pi}V_j^{1/2}} \exp\left(-\frac{1}{2V_j} x_j^2 \right)(2-\phi_{0}^{2j})^{-1/2} \E[G^4]dx_j,
\end{eqnarray}
where $\displaystyle G \sim \mathcal{N}\left(\frac{\phi_{0}^jx_j}{(2-\phi_{0}^{2j})},\mathcal{V} \right)$. We use the fact that the moments of a random variable $X\sim \mathcal{N}(\mu,v)$ are:

\begin{eqnarray*}
&&\E\left[X^n\right]=(n-1)v\E\left[X^{n-2}\right]+\mu\E\left[X^{n-1}\right]\\
&&\E[G^4] =3 \mathcal{V} \E[G^2]+ \left(\frac{\phi_{0}^jx_j}{(2-\phi_{0}^{2j})}\E[G^3]\right)\\
&& \qquad \ \ =3 \mathcal{V}^{2}+(4\mathcal{V}+2)\frac{\phi_{0}^{2j}x_j^2}{(2-\phi_{0}^{2j})^2}+\frac{\phi_{0}^{4j}x_j^4}{(2-\phi_{0}^{2j})^4}.
\end{eqnarray*}

By replacing $\E[G^4]$ in equation (\ref{mom4}), we have:

\begin{eqnarray}
&&\E\left[X_1^4 X_j^4 \exp\left(-\frac{1}{2\gamma_{0}^2}\left(X_1^2+X_j^2 \right)\right)\right]\nonumber\\
&&=\frac{3(2-\phi_{0}^{2j})^{-1/2}}{\gamma_{0}}V_j^{5/2}\left[3\mathcal{V}^2+5V_j(4\mathcal{V}+2)\frac{\phi_{0}^{2j}}{(2-\phi_{0}^{2j})^2}+35V_j^{2}\frac{\phi_{0}^{4j}}{(2-\phi_{0}^{2j})^4}\right],\nonumber\\
&&=\tilde{c}_3(j). \label{c3i}
\end{eqnarray}

For all $j>1$, the matrix $\tilde{C}_{j-1}$ is given by:

\begin{equation*}
\tilde{C}_{j-1}=\frac{\phi_{0}^2}{2\pi\gamma_{0}^2}
\begin{pmatrix}
a_1^2  \tilde{c}_1(j)+2a_1a_2\tilde{c}_2(j)+a_2^2\tilde{c}_3(j)& a_1b_1\tilde{c}_1(j) + a_1b_2+a_2b_1\tilde{c}_2(j)+ a_2b_2 \tilde{c}_3(j)\\
a_1b_1 \tilde{c}_1(j) + a_1b_2+a_2b_1\tilde{c}_2(j)+ a_2b_2 \tilde{c}_3(j) &  b_1^2\tilde{c}_1(j)+2b_1b_2 \tilde{c}_2(j) + b_2^2\tilde{c}_3(j)\\
\end{pmatrix},
\end{equation*}
where the coefficients $\tilde{c}_1(j)$, $\tilde{c}_2(j)$, and $\tilde{c}_3(j)$ are given by (\ref{c1i}), (\ref{c2i}) and (\ref{c3i}).\\

Finally, by replacing the terms $a_1$, $a_2$, $b_1$ and $b_2$, the matrix $\tilde{C}_{j-1}$ is equal to:

\begin{equation*}
\tilde{C}_{j-1}=A
\begin{pmatrix}
(4\phi_{0}^4-4\phi_{0}^2+1)  \tilde{c}_1(j) + \frac{2\phi_{0}^2(1-2\phi_{0}^2)}{\gamma_{0}^2}\tilde{c}_2(j)+\frac{\phi_{0}^4}{\gamma_{0}^4}\tilde{c}_3(j)& \frac{\phi_{0}(2\phi_{0}^2-1)}{2\gamma_{0}^2}\tilde{c}_1(j) + \frac{\phi_{0}(1-3\phi_{0}^2)}{2\gamma_{0}^4}\tilde{c}_2(j)+ \frac{\phi_{0}^3}{2\gamma_{0}^6} \tilde{c}_3(j)\\
\frac{\phi_{0}(2\phi_{0}^2-1)}{2\gamma_{0}^2}\tilde{c}_1(j) + \frac{\phi_{0}(1-3\phi_{0}^2)}{2\gamma_{0}^4}\tilde{c}_2(j)+ \frac{\phi_{0}^3}{2\gamma_{0}^6} \tilde{c}_3(j) & \frac{\phi_{0}^2}{4\gamma_{0}^4} \tilde{c}_1(j)+\frac{-\phi_{0}^2} {2\gamma_{0}^6}\tilde{c}_2(j) + \frac{\phi_{0}^2}{4\gamma_{0}^8} \tilde{c}_3(j)\\
\end{pmatrix},
\end{equation*}
where $A=\frac{\phi_{0}^2}{2\pi\gamma_{0}^2(1-\phi_0^2)^2}$.\\
$\newline$
\textbf{Asymptotic behaviour of the covariance matrix $\Omega_{j-1}(\theta_0)$:} By  the stationary assumption $|\phi_{0}|<1$, the limits of the following terms are:
\begin{equation*}
\lim_{j\rightarrow \infty}V_{j}=\frac{\gamma_{0}^2}{2} \text{ and } \lim_{j \rightarrow  \infty}\mathcal{V}=\frac{\gamma_{0}^2}{2},
\end{equation*}
and

\begin{equation*}
\lim_{j\rightarrow \infty}\tilde{c}_{1}(j)=\frac{\gamma_{0}^4}{8}, 
\lim_{j\rightarrow \infty}\tilde{c}_{2}(j)=\frac{3\gamma_{0}^6}{16},
\lim_{j\rightarrow \infty}\tilde{c}_{3}(j)=\frac{9\gamma_{0}^8}{32}.
\end{equation*}

Therefore,

\begin{eqnarray*}
&&\lim_{j\rightarrow \infty}\tilde{C}_{j-1}=\begin{pmatrix}
\frac{\phi_{0}^2\gamma_{0}^2(2-\phi_{0}^2)^2}{64\pi(1-\phi_{0}^2)^2} & \frac{\phi_{0}^3(2-\phi_{0}^2)}{128\pi(1-\phi_{0}^2)^2}\\
$\newline$
\frac{\phi_{0}^3(2-\phi_{0}^2)}{128\pi(1-\phi_{0}^2)^2} & \frac{\phi_{0}^4}{256\pi(1-\phi_{0}^2)^2\gamma_{0}^2}
\end{pmatrix}=P_1.
\end{eqnarray*}

We obtain:

\begin{eqnarray*}
\lim_{j\rightarrow \infty}{\rm Cov}\left(\nabla_{\theta}m_{\theta_0}(Y_{1}),\nabla_{\theta}m_{\theta_0}(Y_{j})\right)&=&4 \lim_{j\rightarrow \infty}(\tilde{C}_{j-1}-P_1)\\
&=&0_{\mathcal{M}_{2\times2}}.
\end{eqnarray*}

We conclude that the covariance between the two vectors $\nabla_{\theta}m_{\theta_0}(Y_{1}), \nabla_{\theta}m_{\theta_0}(Y_{j})$ vanishes when the lag between the two observations $Y_{1}$ and $Y_{j}$ goes to the infinity. \\

\textbf{Calculus of $V_{\theta_{0}}$:} The Hessian matrix $V_{\theta_0}$ is given in Eq. (\ref{hes}).\\
\end{proof}

\subsection{The SV model}\label{AppSV}

\subsubsection{Contrast function}

The $\mathbb{L}_2$-norm and the Fourier transform of the function $l_{\theta}$ are the same as the Gaussian AR(1) model. The only difference is the law of the measurement noise which is a log-chi-square for the log-transform SV model.\\ 

Consider the random variable $\varepsilon=\beta\log(X^2)-\tilde{\mathcal{E}}$ where $\tilde{\mathcal{E}}=\beta\E[\log(X^{2})]$ such that $\varepsilon$ is centered. The random variable $X$ is a standard Gaussian random. The Fourier transform of $\varepsilon$ is given by:

\begin{eqnarray*}
\E\left[\exp\left(i\varepsilon y\right)\right]&=&\exp\left(-i\tilde{\mathcal{E}}y\right)\E\left[X^{2i\beta y}\right]\\
&=&\exp\left(-i\tilde{\mathcal{E}}y\right)\frac{1}{\sqrt{2\pi}}\int_{-\infty}^{+\infty} x^{2i\beta y}\exp\left(-\frac{x^2}{2}\right)dx
\end{eqnarray*}
By a change of  variable $z=\frac{x^2}{2}$, one has:

\begin{eqnarray*}
\E\left[\exp\left(i \varepsilon y\right)\right]&=&\exp\left(-i\tilde{\mathcal{E}}y\right)\frac{2^{i\beta y}}{\sqrt{\pi}}\underbrace{\int_{0}^{+\infty} z^{i\beta y-\frac{1}{2}}e^{-z}dz}_{\Gamma\left(\frac{1}{2}+i\beta y\right)}\\
&=&\exp\left(-i\tilde{\mathcal{E}}y\right)\frac{2^{i\beta y}}{\sqrt{\pi}}\Gamma\left(\frac{1}{2}+i\beta y\right),
\end{eqnarray*}

and the expression (\ref{contraste_application2}) of the contrast function follows with $u_{l_{\theta}}(y)=\frac{1}{2\sqrt{\pi}}\left(\frac{-i\phi y\gamma^2\exp\left(\frac{-y^2}{2}\gamma^2\right)}{\exp\left(-i\tilde{\mathcal{E}}y\right)2^{i\beta y}\Gamma\left(\frac{1}{2}+i\beta y\right)}\right)$.\\

\subsubsection{Checking assumption of Theorem \ref{MR}}

\emph{Regularity conditions:} The proof is essentially the same as for the Gaussian case since the functions $l_{\theta}(x)$ and $\mathbf{P}m_{\theta}$ are the same. We need only to check the assumptions \textbf{(C)} and \textbf{(T)}. These assumptions are satisfied since Fan (see \cite{fan2}) showed that the noises $\varepsilon_i$ have a Fourier transform $f^{*}_{\varepsilon}$ which satisfies :

\begin{equation*}
 |f^{*}_{\varepsilon}(x)|=\sqrt{2}\exp\left(-\frac{\pi}{2} |x|\right)\left(1+O\left(\frac{1}{|x|}\right)\right), \quad |x|\rightarrow \infty,             
\end{equation*}
which means that $f_{\varepsilon}$ is super-smooth in its terminology. Furthermore, by the compactness of the parameter space $\Theta$ and as the functions $l^{*}_{\theta}$, and for $j,k \in \left\{1, 2\right\}$, the functions $(\frac{\p l_{\theta}}{\p \theta_j })^*$  $(\frac{\p^{2} l_{\theta}}{\p\theta_j \p\theta_k})^{*}$, have the following form $C_1(\theta)P(x) \exp\left(-C_2(\theta)x^{2}\right)$ where $C_1(\theta)$ and $C_2(\theta)$ are two constants well defined in the parameter space $\Theta$ with $C_2(\theta)>0$, we obtain: 

\begin{equation*}
\left\lbrace\begin{array}{ll}
\E\left(\left|Y_{2}u^*_{\nabla_{\theta}l_{\theta}}(Y_{1}) \right|^{2+\delta}\right)<\infty \qquad\qquad \text{ for some } \delta >0,\\
\E\left(\sup_{\theta \in \mathcal{U}}\left\|Y_{2}u^{*}_{\nabla_{\theta}^{2}l_{\theta}}(Y_{1})\right\|\right)<\infty \qquad \text{ for some neighbourhood } \mathcal{U} \text{ of } \theta_0 .
\end{array}
\right.
\end{equation*}

\subsubsection{Expression of the Covariance matrix:}

 As, the functions $l_{\theta}(x)$ and $\mathbf{P}m_{\theta}$ are the same for the two models, the expressions of the matrix $V_{\theta_0}$ and $\Omega_{j}(\theta_0)$ are given in Lemma \ref{hessienne_application}. We need only to use an estimator of $P_{2}=\E[Y^{2}_{2}(u^{*}_{\nabla l_{\theta}}(Y_{1}))^2]$ since we can just approximate $u^{*}_{\nabla l_{\theta}}(y)$. A natural and consistent estimator of $P_2$ is given by:  

\begin{equation}\label{est_P2}
 \widehat{P}_{2}=\frac{1}{n}\sum_{i=1}^{n-1}\left(Y^{2}_{i+1}(u^{*}_{\nabla l_{\theta}}(Y_i))^2\right),
\end{equation}

\begin{remark}
 In some models, the covariance matrix $\Omega_{j}(\hat{\theta}_n)$ cannot be explicitly computable. We refer the reader to \cite{Fy00} chapter 6 Section 6.6 p.408 for this case.\\
 \end{remark}

\section{EM algorithm}

We first refer to \cite{dempster} for general details on the EM algorithm. The EM algorithm is an iterative procedure for maximizing the log-likelihood $l(\theta)=\log(\P_{\theta}(Y_{1:n}))$. Suppose that after the $k^{th}$ iteration, the estimate for $\theta$ is given by $\theta_k$. Since the objective is to maximize $l(\theta)$, we want to compute an updated $\theta$ such that:
\begin{equation*}
l(\theta)>l(\theta_k)
\end{equation*}

Hidden variables can be introduced for making the ML estimation tractable. Denote the hidden random variables $U_{1:n}$ and a given realization by $u_{1:n}$. The total probability $\P_{\theta}(Y_{1:n})$ can be written as:

\begin{equation*}
\P_{\theta}(Y_{1:n})=\sum_{u_{1:n}}^{}\P_{\theta}(Y_{1:n}\vert u_{1:n})\P_{\theta}(u_{1:n})
\end{equation*}Hence,

\begin{eqnarray}
l(\theta)-l(\theta_k)&=&\log(\P_{\theta}(Y_{1:n}))-\log(\P_{\theta_k}(Y_{1:n}))\nonumber\\
&=&\log\left( \sum_{u_{1:n}}^{}\P_{\theta}(Y_{1:n}\vert u_{1:n})\P_{\theta}(u_{1:n})\right)-\log(\P_{\theta_k}(Y_{1:n}))\nonumber\\
&=&\log\left( \sum_{u_{1:n}}^{}\P_{\theta}(Y_{1:n}\vert u_{1:n})\P_{\theta}(u_{1:n})\frac{\P_{\theta_k}(u_{1:n}\vert Y_{1:n})}{\P_{\theta_k}(u_{1:n}\vert Y_{1:n})}\right)-\log(\P_{\theta_k}(Y_{1:n}))\nonumber\\
&=&\log\left( \sum_{u_{1:n}}^{} \P_{\theta_k}(u_{1:n}\vert Y_{1:n})\frac{\P_{\theta}(Y_{1:n}\vert u_{1:n})\P_{\theta}(u_{1:n}) }{\P_{\theta_k}(u_{1:n}\vert Y_{1:n})}\right)-\log(\P_{\theta_k}(Y_{1:n}))\label{densit1}\\
&\geq& \sum_{u_{1:n}}^{} \P_{\theta_k}(u_{1:n}\vert Y_{1:n})\log\left( \frac{\P_{\theta}(Y_{1:n}\vert u_{1:n})\P_{\theta}(u_{1:n}) }{\P_{\theta_k}(u_{1:n}\vert Y_{1:n})}\right)-\log(\P_{\theta_k}(Y_{1:n}))\label{densit2}\\
&=&\sum_{u_{1:n}}^{} \P_{\theta_k}(u_{1:n}\vert Y_{1:n})\log\left( \frac{\P_{\theta}(Y_{1:n}\vert u_{1:n})\P_{\theta}(u_{1:n}) }{\P_{\theta_k}(u_{1:n}\vert Y_{1:n})}\right)-\log(\P_{\theta_k}(Y_{1:n}))\sum_{u_{1:n}}^{} \P_{\theta_k}(u_{1:n}\vert Y_{1:n}) \label{densit3}\\
&=&\sum_{u_{1:n}}^{} \P_{\theta_k}(u_{1:n}\vert Y_{1:n})\log\left( \frac{\P_{\theta}(Y_{1:n}\vert u_{1:n})\P_{\theta}(u_{1:n}) }{\P_{\theta_k}(u_{1:n}\vert Y_{1:n})\P_{\theta_k}(Y_{1:n})}\right)\nonumber\\
&=&\Delta(\theta, \theta_k).\nonumber
\end{eqnarray} In going from Eq.(\ref{densit1}) to Eq.(\ref{densit2}) we use the Jensen inequality: $\log  \sum_{i=1}^{n} \lambda_i x_i \geq \sum_{i=1}^{n} \lambda_i \log(x_i )$ for constants $\lambda_i \geq 0$ with $\sum_{i=1}^{n} \lambda_i=1$. And in going from Eq.(\ref{densit2}) to Eq.(\ref{densit3})  we use the fact that $\sum_{u_{1:n}}^{} \P_{\theta_k}(u_{1:n}\vert Y_{1:n})=1$. Hence, 

\begin{equation*}
l(\theta)\geq l(\theta_k)+\Delta(\theta, \theta_k)=\mathcal{L}(\theta, \theta_k) \text{ and } \Delta(\theta, \theta_k) =0 \text{ for } \theta=\theta_k
\end{equation*} The function $\mathcal{L}(\theta, \theta_k)$ is bounded by the log-likelihood function $l(\theta)$ and they are equal when $\theta=\theta_k$. Consequently, any $\theta$ which increases $\mathcal{L}(\theta, \theta_k)$ will increases $l(\theta)$. The EM algorithm selects $\theta$ such that $\mathcal{L}(\theta, \theta_k)$ is maximized. We denote this updated value $\theta_{k+1}$. Thus,

\begin{eqnarray}
\theta_{k+1}&=&\arg\max_{\theta}\left\{l(\theta_k)+ \sum_{u_{1:n}}^{} \P_{\theta_k}(u_{1:n}\vert Y_{1:n})\log\left( \frac{\P_{\theta}(Y_{1:n}\vert u_{1:n})\P_{\theta}(u_{1:n}) }{\P_{\theta_k}(u_{1:n}\vert Y_{1:n})\P_{\theta_k}(Y_{1:n})}\right) \right\}\nonumber\\
&=&\arg\max_{\theta}\left\{\sum_{u_{1:n}}^{} \P_{\theta_k}(u_{1:n}\vert Y_{1:n})\log\P_{\theta}(Y_{1:n}\vert u_{1:n})\P_{\theta}(u_{1:n})  \right\} \text{ if we drop the terms which don't depend on } \theta.\nonumber\\
&=&\arg\max_{\theta} \left\{\E[\log\P_{\theta}(Y_{1:n}\vert u_{1:n})\P_{\theta}(u_{1:n})]\right\}  \text{ where the expectation is according to } \P_{\theta_k}(u_{1:n}\vert Y_{1:n}).\label{maxim}
\end{eqnarray} 

\subsection{Simulated Expectation Maximization Estimator}\label{siem}
Here, we describe the SIEMLE proposed by Kim, Shepard and Chib \cite{shep} for the SV model, these authors retain the linear log-transform model given in (\ref{svappl}). However, instead of approximating the log-chi-square distribution of $\varepsilon_i$ with a Gaussian distribution, they approximate $\varepsilon_i$ by a mixture of seven Gaussian. The distribution  of the noise is given by:

\begin{eqnarray*}
f_{\varepsilon_i}(x)&\approx& \sum_{j=1}^{7}q_j \times g_{(m_j,v_j^2)}(x),\\
&\approx& \sum_{j=1}^{7}q_jf_{\varepsilon_i|s_i=j}(x)
\end{eqnarray*}
where $g_{(m,v)}(x)$ denotes the Gaussian distribution of $\varepsilon_i$ with mean $m$ and variance $v$, and $f_{\varepsilon_i|s_i=j}(x)$ is a Gaussian distribution conditional to an indicator variable $s_i$  at time $i$ and the variables $q_j, j=1\cdots, 7$ are the given weights attached to each component and such that $\sum_{j=1}^{7}q_j=1$. Note that, most importantly, given the indicator variable $s_i$ at each time $i$, the log-transform model is Gaussian. That is:
\begin{equation*}
f_{\theta}(Y_i|s_i=j,X_i)\sim g_{(X_i+m_j, v_j^2)}.
\end{equation*}

Then, conditionally to the indicator variable $s_i$, the SV model becomes a Gaussian state-space model and the Kalman filter can be used in the SIEMLE in order to compute the log-likelihood function given by:

\begin{equation*}
\log f_{\theta}(Y_{1:n}|s_{1:n})=
-\frac{n}{2}\log(2\pi)-\frac{1}{2}\sum_{i=1}^{n}\log F_i-\frac{1}{2}\sum_{i=1}^{n}\frac{\nu_{i}^{2}}{F_i},
\end{equation*} with $\nu_i=(Y_i-\hat{Y}_i^{-}-m_{s_i})$ and $F_i=\V_{\theta}[\nu_i]=P_i^{-}+v^{2}_{s_i}$. The quantities $\hat{Y}_i^{-}=\E_{\theta}[Y_i| Y_{1:i-1}]$ and $P_i^{-}=\V_{\theta}[(X_i-\hat{X}_{i}^{-})^{2}]$ are computed by the Kalman filter.\\

Hence, if we consider that the missing data $u_{1:n}$ for the EM correspond to the indicator variables $s_{1:n}$, then according to Eq.(\ref{maxim}) and since $f(s_{1:n})$ do not depend on $\theta$, the Maximization step is:

\begin{equation*}
\theta_{k+1}=\arg\max_{\theta} \left\{\E[\log\P_{\theta}(Y_{1:n}|s_{1:n})]\right\}=\arg\max_{\theta}Q(\theta,\theta_k)
\end{equation*} where the expectation is according to $\P_{\theta_k}(s_{1:n}\vert Y_{1:n})$. Nevertheless, for the SV model, the problem with the EM algorithm is that the density $f_{\theta}(s_{1:n}|Y_{1:n})$ is unknown. The main idea consists in introducing a Gibbs algorithm to obtain $\tilde{M}$ draws $s^{(1)}_{1:n},\cdots, s^{(\tilde{M})}_{1:n}$ from the law $f_{\theta}(s_{1:n}|Y_{1:n})$. Hence, the objective function $Q(\theta,\theta_{k})$ is approximated by:
\begin{equation*}
\tilde{Q}(\theta,\theta_{k})=\frac{1}{\tilde{M}}\sum_{l=1}^{\tilde{M}}\log f_{\theta}(Y_{1:n}|s_{1:n}^{(l)})
\end{equation*} 

Then, the simulated EM algorithm for the SV model is as follows: Let $C>0$ be a threshold to stop the algorithm and $\theta_{k}$ a given arbitrary value of the parameter. While $|\theta_{k}-\theta_{k-1}|>C,$ \\

\begin{enumerate}
\item Apply the Gibbs sampler as follows:\\

\textbf{The Gibbs Sampler:} Choose arbitrary starting values $X_{1:n}^{(0)}$, and let $l=0$.
\begin{enumerate}
\item Sample  $s_{1:n}^{(l+1)}\sim f_{\theta_{k}}(s_{1:n}|Y_{1:n},X_{1:n}^{(l)})$.
\item Sample $X_{1:n}^{(l+1)}\sim f_{\theta_{k}}(X_{1:n} |Y_{1:n} ,s_{1:n}^{(l+1)})$.
\item Set $l=l+1$ and goto (a).
\end{enumerate}

\item $\theta_{k+1}=\arg\max_{\theta} \tilde{Q}(\theta,\theta_{k})$.
\end{enumerate}

Step (a): to sample the vector $s_{1:n}$ from its full conditional density, we sample each $s_i$ independently. We have:

\begin{equation*}
f_{\theta_{k}}(s_{1:n}|Y_{1:n},X_{1:n})=\prod_{r=1}^{n}f_{\theta_{k}}(s_{r}|Y_{r},X_{r})\propto \prod_{r=1}^{n}f_{\theta_{k}}(Y_{r}|s_{r},X_{r})f(s_r),
\end{equation*}
and $f_{\theta_{k}}(Y_{r}|s_{r}=j,X_{r})\propto g_{(X_{r}+m_{j}, v_{j}^2)}$ for $j=1\cdots,7.$ And the step (b) of the Gibbs sampler is conducted by the Kalman filter since the model is Gaussian.\\

\begin{acknowledgements}
I thank my co-director Patricia Reynaud-Bouret for her idea about this paper and for her help and generosity. I thank also my director Fr{\'e}d{\'e}ric Patras for his supervisory throughout this paper. I thank him for his careful reading of the paper and his large comments. I would like to thank F. Comte, A. Samson, N. Chopin, M. Miniconi and F. Pelgrin for their suggestions and their interest about this framework.\\
I thank also the referees for careful reading and constructive suggestions which were helpful in improving substantially this paper.
\end{acknowledgements}

\bibliographystyle{alpha}      % basic style, author-year citations
\bibliography{HMM_ESTI}   % name your BibTeX data base

\end{document}